\documentclass[longbibliography]{amsart}
\usepackage[unicode=true,pdfusetitle, bookmarks=true,bookmarksnumbered=false,bookmarksopen=false, breaklinks=false,pdfborder={0 0 0},backref=false,colorlinks=false]{hyperref}
\hypersetup{colorlinks,linkcolor=myurlcolor,citecolor=myurlcolor,urlcolor=myurlcolor}
\usepackage{graphics,epstopdf,graphicx,amsthm,amsmath,amssymb,mathptmx,braket,colortbl,color,bm,framed,mathrsfs}
\usepackage[up]{subfigure}
\usepackage{tikz}
\usepackage{multirow}
\usepackage{enumerate}

\usepackage{tabularx}
\usepackage{stmaryrd}

\usepackage{booktabs} 
\let\hbar\undefined
\usepackage{fouriernc} 

\makeatletter
\def\l@section{\@tocline{1}{0pt}{1pc}{}{}}
\def\l@subsection{\@tocline{2}{0pt}{1pc}{4.6em}{}}
\def\l@subsubsection{\@tocline{3}{0pt}{1pc}{7.6em}{}}
\renewcommand{\tocsection}[3]{%
  \indentlabel{\@ifnotempty{#2}{\makebox[2.3em][l]{%
    \ignorespaces#1 #2.\hfill}}}#3}
\renewcommand{\tocsubsection}[3]{%
  \indentlabel{\@ifnotempty{#2}{\hspace*{2.3em}\makebox[2.3em][l]{%
    \ignorespaces#1 #2.\hfill}}}#3}
\renewcommand{\tocsubsubsection}[3]{%
  \indentlabel{\@ifnotempty{#2}{\hspace*{4.6em}\makebox[3em][l]{%
    \ignorespaces#1 #2.\hfill}}}#3}
\makeatother
\definecolor{myurlcolor}{rgb}{0,0,0.9}

\newcommand{\proj}[1]{| #1\rangle\!\langle #1 |}

\newcommand{\inner}[2]{\langle #1 , #2\rangle}

\newcommand{\ep}[1]{\langle #1 \rangle}

\DeclareMathOperator{\trace}{Tr}
\newcommand{\Ptr}[2]{\trace_{#1}\Pa{#2}}
\newcommand{\Tr}[1]{\Ptr{}{#1}}

\newcommand{\Pa}[1]{\left[#1\right]}

\newcommand{\norm}[1]{\left\lVert #1 \right\rVert}

\theoremstyle{plain}
\newtheorem{thm}{Theorem}
\newtheorem{lem}[thm]{Lemma}
\newtheorem{prop}[thm]{Proposition}
\newtheorem{cor}[thm]{Corollary}

\newtheorem{Def}[thm]{Definition}
\newtheorem{Rem}[thm]{Remark}

\newcommand*{\myproofname}{Proof}

\def\ot{\otimes}
\def\complex{\mathbb{C}}
\def\CN{\mathbb{C}}
\def\real{\mathbb{R}}
\def\R{\mathbb{R}}
\def\Z{\mathbb{Z}}

\newcommand{\CEE}{\mathcal E}

\newcommand{\CHH}{\mathcal H} 

\newcommand{\CLL}{\mathcal L} 

\newcommand{\CMM}{\mathcal M}

\newcommand{\CNN}{\mathcal N}

\newcommand{\CRR}{\mathcal R}

\newcommand{\be}{\begin{equation}}
\newcommand{\ee}{\end{equation}}

\renewcommand{\ge}{\geqslant}
\renewcommand{\geq}{\geqslant}
\renewcommand{\leq}{\leqslant}
\renewcommand{\le}{\leqslant}

\DeclareMathAlphabet{\mathcal}{OMS}{cmsy}{m}{n}

\makeatother

\begin{document}

\title{Discrete Quantum Gaussians and Central Limit Theorem}


  %

   \author{Kaifeng Bu$^{1,2}$}
  \email{bu.115@osu.edu}
 
   \author{Weichen Gu$^{1,3}$}
  \email{gu.1213@osu.edu}

   \author{Arthur Jaffe$^2$}
  \email{Arthur\_Jaffe@harvard.edu}

  \address[$1$]{Department of Mathematics, The Ohio State University, Columbus, Ohio 43210, USA }
  
   \address[$2$]{Departments of Mathematics and Physics, Harvard University, Cambridge, Massachusetts 02138, USA}

   \address[$3$]{Department of Mathematics and Statistics, University of New Hampshire, Durham, New Hampshire  03824, USA}

\begin{abstract}

We introduce a  quantum convolution and a conceptual framework to study states in discrete-variable (DV) quantum systems.  All our results suggest that stabilizer states play a role in DV quantum systems similar to the role that Gaussian states play in continuous-variable systems. Hence we suggest the name ``discrete quantum Gaussians'' for stabilizer states. For example, we prove that the convolution of two stabilizer  states is another stabilizer state, and that stabilizer states extremize both quantum entropy and Fisher information. We establish a ``maximal entropy principle,'' a ``second law of thermodynamics for quantum  convolution,'' and a ``quantum central limit theorem.'' The latter is based on iterating the convolution of a zero-mean quantum state, which we prove converges to a stabilizer state. We bound the exponential rate of convergence of the quantum central limit theorem by the ``magic gap,'' defined in terms of the support of the characteristic function of the state. We elaborate on our general results with a discussion of some examples, as well as extending many of them to quantum channels.

\end{abstract}

\maketitle

\setcounter{tocdepth}{3}
\tableofcontents

\section{Introduction}
For continuous-variable (CV) quantum information with $n$ degrees of freedom, the Hilbert space is $L^{2}\left(\mathbb{R}^{n}\right)$.  A state is a positive linear functional on the bounded operators taking the identity to 1.
For discrete-variable (DV) quantum information taking $d$ possible values, the Hilbert space for $n$ particles is $\left(\mathbb{C}^d\right)^{\otimes n}$.
A state for such a DV system is a positive transformation on the Hilbert space with unit trace. A goal of this work is to unify these two perspectives. 

To achieve this end, we construct a convolution for DV systems and establish a framework for investigating stabilizer states as limits of repeated convolutions. Stabilizer  states were introduced by Gottesman for qubits (the case $d=2$) to understand quantum error correction. We show  that stabilizer states and quantum channels play a role  in DV quantum systems analogous to the role played by Gaussian states and unitaries in CV quantum systems.

Although many continuous versions of a quantum central limit theorem have been proved, central limit theorems for  DV quantum systems remain relatively unexplored. We address this gap here and present a  central limit theorem specifically tailored for DV quantum systems, one that converges  to stabilizer states.

In classical probability theory, Gaussian distributions (or random variables) play a pivotal role, underpinning several fundamental concepts and results. For instance, the maximal entropy principle  asserts that Gaussian random variables achieve maximum entropy among all random variables with a given expectation and covariance~\cite{CT06}. Gaussians are invariant under convolution, and moreover the iteration of convolution gives rise to the central limit theorem (CLT). Normalized sums of independent and identically distributed random variables, $\frac{1}{\sqrt{N}}\sum^N_{i=1}X_i$, converge to a Gaussian random variable. Furthermore, the Shannon entropy of the normalized sum in the CLT is an increasing function of $N$. This fact, originally conjectured by Lieb~\cite{McKean1966,Lieb78}, was  proved by Artstein, Ball, Barthe, and Naor~\cite{artstein2004JAMS}.

The possibility of a central limit theorem for CV quantum systems dates back to  Cushen and Hudson~\cite{Cushen71} and to Hepp and Lieb~\cite{Lieb73,Lieb1973}. 
Many other CV quantum versions of the central limit theorem have been found, see~\cite{Giri78,Goderis89,Matsui02,Cramer10,Jaksic09,Arous13,Michoel04,GoderisPTRT89,JaksicJMP10,Accardi94,Liu16,JiangLiuWu19,Hayashi09,CampbellPRA13,BekerCMP21}. For example, if the convolution is defined by the beam splitter in CV quantum information, the CLT  converges to a Gaussian state~\cite{Hayashi09,CampbellPRA13,BekerCMP21}. 
Another CLT occurs in subfactor theory,  with  convergence to a biprojection; in this case the positivity preserving property of convolution is Liu's  quantum Schur product theorem~\cite{Liu16,JiangLiuWu19}.

A different point of view arises in  Voiculescu's free probability theory~\cite{VDN92,voiculescu2016free}, leading to the free central limit theorem. Repeated, normalized, (additive) free convolution of a probability measure  converges (with some additional assumptions) to 
a semicircle distribution~\cite{voiculescu1986addition,voiculescu1987multiplication}. The free entropy, introduced by Voiculescu~\cite{voiculescu1993analogues}, is maximized by  random variables satisfying the semicircle law~\cite{voiculescu1997analogues}.
The free entropy increases monotonically under repeated free convolution~\cite{shlyakhtenko2007free,shlyakhtenko2007shannon}.

\subsection{Main Results}
\begingroup
\setlength{\tabcolsep}{5pt} 
\renewcommand{\arraystretch}{1.5} 

\begin{table}[h!tbp]
\centering
 \resizebox{\textwidth}{10mm}{
\begin{tabular}{ |c|c|c| } 
\hline
& \bf DV quantum systems (this work) &\bf CV classical results\\
\hline
\bf Convolution & quantum convolution $\boxtimes$ & classical convolution  $*$\\
\hline
\bf Gaussians & Stabilizer states & Gaussian distributions\\
\hline
\bf Central limit theorem & Converge to stabilizer states &  Converge to Gaussian distributions\\
\hline
\end{tabular}}
\vskip 5pt
\caption{\label{tab:sum_high}Short summary of  main results}

\end{table}

\endgroup

We prove a central limit theorem for DV quantum systems, with the limit of repeated convolution being a stabilizer state. This motivates our calling these states ``discrete quantum Gaussians.'' 
Our quantum CLT is based on a twisted convolution that we denote $\rho\boxtimes \sigma$.
We establish different properties of $\rho\boxtimes \sigma$, with the general goal to reveal different ways in which stabilizer states are extremal.

\begin{enumerate}[(1)] 
\item{}
Given a state $\rho$, in  Definition \ref{def:mean_state} we define its mean state (MS) $\mathcal{M}(\rho)$.   In Definition \ref{def:MSPS} we define the set of minimal stabilizer-projection states (MSPS)  associated to $\rho$ and show that $\mathcal{M}(\rho)\in \text{MSPS}$. We show  in Theorem \ref{thm:main1} that $\mathcal{M}(\rho)$  is the closest state in MSPS to $\rho$ as measured by the relative R\'enyi entropy. We establish an extremality condition: within all quantum states having the same  MS up to conjugation by   a Clifford unitary, the MSPS attains the maximum  R\'enyi entropy. 

\item{}
We introduce the notion of the  \textit{magic gap} of a quantum state, as  the difference between the first  and second largest absolute values in the support of the characteristic function of the state, see  Definition \ref{def:ma_gap}. We prove that the magic gap  provides  a  lower bound on the number of non-Clifford gates in the synthesis of a given unitary. We formulate these  results in Propositions~\ref{prop:magap} and \ref{prop:gap_syn}.

\item{}
Our convolution $\boxtimes$ is given in Definition~\ref{def:convo}. It is different from the convolution considered earlier in~\cite{Audenaert16, Carlen16}. An important  basic property of our convolution is: the convolution of two stabilizer states is a stabilizer state. 
We prove an  inequality on the spectrum, which implies a series of  inequalities for generalized quantum R\'enyi entropy and subentropy.  We also prove an inequality on quantum Fisher information in Theorem \ref{thm:fisher}.

\item{}
Stabilizer states play an important role for the convolution $\boxtimes$. We show in Theorem \ref{thm:min_outen} that  the  convolutional channel  achieves minimal output entropy, if and only if the input states  are pure stabilizer states. We also study the Holevo channel capacity of the convolutional channel. We show in Theorem \ref{thm:holv_stab} that 
the convolutional channel achieves  maximal Holevo capacity, if and only if the state is a stabilizer state. These results provide some further understanding and insight for  stabilizer states. 

\item{}
Our convolutional approach includes two characteristic examples, the DV beam splitter and the DV amplifier. These examples are important, because they share a similar structure to their CV counterparts and give evidence unifying CV and DV quantum information theory. We compare our  DV  results on the beam splitter to the known results for CV quantum systems in \S\ref{subsec:examp}, Table~\ref{tab:sum_B}. We also compare CV and DV cases for the amplifier in \S\ref{subsec:examp},  Table~\ref{tab:sum_S}. 

\item{}
In  \S\ref{sec:CLL}, we use our discrete convolution to establish a quantum central limit theorem for DV systems. To the best of our knowledge, this represents the first central limit theorem for DV quantum systems.

\item{}
We find the second law of quantum R\'enyi entropy under convolution, i.e.,  quantum R\'enyi entropy $H_{\alpha}(\boxtimes^N\rho)$ is nondecreasing w.r.t. the number of convolutions $N$. Moreover,
we find that the repeated convolution of any zero-mean quantum state converges to the MS. The exponential rate of convergence to the limit is bounded by the magic gap of the state, stated  precisely in Theorem \ref{thm:CLT_gap}. 

\item{}
We generalize many of our results for states to quantum channels in \S\ref{sect:ConvolutionChannels}.
\end{enumerate}

\begin{Rem}
In this paper, we introduce a general mathematical framework for quantum convolution, including rigorous and detailed proofs.  In  a companion paper \cite{BGJ23a},
we announced a subset of these results established here, but without proofs. In that work the quantum convolution was based on specific choices, such as discrete beam splitters and amplifiers, with a focus on the physical motivation and basic ideas of quantum convolution.
\end{Rem}

\section{Preliminaries}\label{sec:pre}
We study $n$-qudit systems with Hilbert space $\CHH^{\ot n}$, 
where $\CHH\simeq \CN^d$ and $d$ is prime. \footnote{The assumption that $d$ is prime simplifies the analysis; in this case $\Z_d$ is a field. For example, the Clifford unitaries for a system with prime-power dimension have a complete characterization; this is unknown for arbitrary $d$.} The total dimension of $\CHH^{\ot n}$ is $d^n$, i.e., a power of a prime number. Let $\mathcal B(\mathcal{H}^{\ot n})$ denote the set of all linear operators on $\mathcal{H}^{\ot n}$,
and
$\mathcal{D}(\mathcal{H}^{\ot n})$ denote the set of  all quantum states on $\CHH^{\ot n}$, namely positive operators with unit trace.
{If the quantum state $\rho$ is rank-one, then it is called a pure state, also denoted  $\rho=\proj{\psi}$. Here $\ket{\psi}$ denotes a unit vector in $\mathcal{H}^{\ot n}$.
For simplicity, with pure states we sometimes use the word  ``state'' to refer either to the vector  $\ket{\psi}$ or to the matrix  $\proj{\psi}$.}

Any quantum state $\rho$ has a spectral decomposition 
\begin{align*}
    \rho=\sum_i\lambda_i\proj{\psi_i},
\end{align*}
where $\set{\lambda_i}_i$ are eigenvalues, and $\set{\ket{\psi_i}}_i$ are the eigenvectors. 
Here, we use $\vec \lambda_{\rho}$ to denote the vector of eigenvalues
\begin{align}\label{eq:eigen_vec}
    \vec \lambda_{\rho}=(\lambda_1,...,\lambda_{d^n}).
\end{align}
The Schatten $L^{p}$ norm $\norm{\cdot}_p$ of a linear operator $O$ is 
\begin{eqnarray*}
\norm{O}_p
=\left(\Tr{|O|^p}\right)^{\frac 1p},
\quad\text{where }|O|= \sqrt{O^\dag O}\;.
\end{eqnarray*}
Here $O^\dag$ denotes the Hermitian adjoint. 
\begin{Def}
    A quantum channel is a completely positive, trace-preserving (CPTP) map $\mathcal{B}(\mathcal{H}^{\ot n})\to \mathcal{B}(\mathcal{H}^{\ot n})$. Let $ L(\mathcal{H}^{\ot n})$ denote the set of quantum channels on $\mathcal{H}^{\ot n}$.
\end{Def}
Define a linear map $\Theta:  L(\mathcal{H}^{\ot n})\to  L(\mathcal{H}^{\ot n})$ acting on quantum channels to be a linear supermap. If $id\ot \Theta$ is CPTP preserving for  any identity supermap $id$, then 
$\Theta$ is a superchannel, and it will map any quantum channel to a quantum channel.
In the Hilbert space $\mathcal{H}$, we define one  orthonormal set to be the computational basis and denote it in Dirac notation
\be
\set{\ket{k}}_{k\in \mathbb{Z}_d}\;.
\ee
 The Pauli matrices $X$ and $Z$ are unitary transformations that act as
\[ X |k\rangle = |k+1\rangle\;, 
\quad 
Z |k\rangle = \chi(k) |k\rangle,\;\;\;\forall k\in \Z_d\;.
\]
Here 
\be
\chi(k)=\xi^k_d\;,\quad
\text{and}\quad
\xi_d=\exp(2\pi i /d)\;,
\ee
is a $d$-th root of unity. The Weyl operators 
are defined as 
\begin{eqnarray}\label{eqn:wpq}
w(p,q)=\chi(-2^{-1}pq)\,Z^pX^q\;.
\end{eqnarray}
Here $ 2^{-1}$ denotes the inverse $\frac{d+1}{2}$ of 2 in $\mathbb{Z}_d$, i.e.,  the  $x$ in $\Z_d$ such that $2x \equiv 1 \mod d$.
If $d=2$, the Weyl operators are defined as 
\begin{eqnarray*}
w(p,q)=i^{-pq}\,Z^pX^q \; .
\end{eqnarray*}

The Weyl operators satisfy, when $d>2$,
\begin{eqnarray}\label{0106shi2}
w(p,q)\,w(p',q')=
 \chi(2^{-1}\inner{(p,q)}{(p',q')}_s)\,
 w(p+p',q+q') \; ,
\end{eqnarray}
and when $d=2$, 
\begin{eqnarray}\label{0106shi2_2}
w(p,q)\,w(p',q')=
i^{\inner{(p,q)}{(p',q')}_s}\,
 w(p+p',q+q') \; ,
\end{eqnarray}
where the symplectic inner product $\inner{(p,q)}{(p',q')}_s$ is defined as 
\begin{eqnarray}\label{SymplecticInnerProduct}
\inner{(p,q)}{(p',q')}_s
=pq'-qp' \; .
\end{eqnarray}
Let us denote 
\begin{eqnarray*}
V^n=\mathbb{Z}^n_d\times\mathbb{Z}^n_d \; ,
\end{eqnarray*}  
and for any $(\vec p, \vec q)\in V^n$, the Weyl operator $w(\vec{p}, \vec q)$ is defined as 
\begin{eqnarray*}
w(\vec p, \vec q)
=w(p_1, q_1)\ot...\ot w(p_n, q_n) \; ,
\end{eqnarray*}
with $\vec p=(p_1, p_2,..., p_n)\in \mathbb{Z}^n_d,\ \vec q=(q_1,..., q_n)\in \mathbb{Z}^n_d $.

\begin{Def}[\bf Weyl group]
Let the local dimension $d$ be an odd prime. The Weyl group is 
\begin{eqnarray}
 \set{\xi^k_d:k\in \mathbb{Z}_d} \times  \set{w(\vec p, \vec q): \vec p, \vec q\in \mathbb{Z}^n_d}.
\end{eqnarray}
If the local dimension $d=2$, 
the Weyl group is 
\begin{eqnarray}
 \set{\pm i, \pm 1} \times  \set{w(\vec p, \vec q): \vec p, \vec q\in \mathbb{Z}^n_2}.
\end{eqnarray}
\end{Def}
The set of the Weyl operators $\set{w(\vec{p},\vec q)}_{(\vec p, \vec q)\in V^n}$ forms an orthonormal basis 
in $\mathcal B(\mathcal{H}^{\ot n})$ with respect to the inner product 
\be
\inner{A}{B}=\frac{1}{d^n}\Tr{A^\dag B}\;.
\ee

\begin{Def}[\bf Characteristic function]\label{def:charfn}
For any $n$-qudit state $\rho\in \mathcal{D}(\mathcal{H}^{\ot n})$, the characteristic function $\Xi_{\rho}:V^n\to\complex$ is defined
as 
\begin{eqnarray}\label{eq:charfn}
\Xi_{\rho}(\vec{p},\vec q)
=\Tr{w(-\vec{p},-\vec q)\,\rho } \; .
\end{eqnarray}
\end{Def}
Hence, the state $\rho$ can be written as a linear combination of the Weyl operators with characteristic function
\begin{eqnarray}\label{0109shi5}
\rho=\frac{1}{d^n}
\sum_{(\vec{p},\vec q)\in V^n}
\Xi_{\rho}(\vec{p},\vec q)\,w(\vec{p},\vec q) \; .
\end{eqnarray}

We define the expansion into characteristic functions as the quantum Fourier transform that we study. 
The characteristic function has been used to study quantum Boolean functions~\cite{montanaro2010quantum}, and 
was later  applied to study quantum circuit complexity~\cite{Bucomplexity22}, and quantum scrambling~\cite{GBJPNAS23}.
(See also a more general framework of quantum Fourier analysis through a picture Fourier transform, which intertwines matrix units with Weyl operators~\cite{JaffePNAS20}.)

\begin{Def}[\bf Stabilizer state ~\cite{Gottesman96,Gottesman97}]
A pure stabilizer vector for an $n$-qudit system  is a common unit eigenvector for an abelian subgroup of the Weyl group of size $d^n$. (The corresponding state is the projection onto the eigenvector, and sometimes we use this nomenclature interchangeably, calling the eigenvector a pure state.) A general stabilizer state $\rho$ is a  convex linear combination of pure stabilizer states.
\end{Def}

In other words, $\ket{\psi}$ is a pure stabilizer vector if there exists an abelian subgroup $S$ of the 
Weyl operators with $n$ generators $\set{w(\vec{p}_i, \vec q_i)}_{i\in[n]}$ such that $w(\vec{p}_i, \vec q_i)\ket{\psi}=\chi(x_i)\ket{\psi}$ with 
$x_i\in \mathbb{Z}_d$ for every $i\in[n]$. 
In general, every abelian subgroup $S$ of the Weyl operators has size $d^r$ with $0\leq r\leq n$.
The operators in $S$ generate an abelian $C^*$-algebra $C^*(S)$.  The projections in $C^*(S)$ are called the \emph{stabilizer projections} associated with $S$.

\begin{Def}[\bf Minimal stabilizer-projection state]\label{def:MSPS}
Let $S$ be an abelian subgroup of Weyl operators. 
A minimal projection in $C^*(S)$ is called a \emph{minimal stabilizer projection} associated with $S$.
A minimal stabilizer-projection state (MSPS) is a minimal stabilizer projection normalized by  dividing  by its rank. 
\end{Def}

It is clear that if $P$ is a stabilizer projection associated with a subgroup $S$ of an abelian group $S'$,
then $P$ is also associated with $S'$.
In addition, when some stabilizer projection $P$ is given,
there is a unique minimal abelian subgroup $S$ associated with $P$, in the sense that for every  $S'$ associated with $P$, we have $S\subseteq S'$.
For example, let us consider the abelian group $S=\langle Z_1,...,Z_{n-1}\rangle$ for an $n$-qudit system, 
the states $\set{\frac{1}{d}\proj{\vec j}\ot I}_{\vec j\in\mathbb{Z}^{n-1}_d}$ are MSPS. 

\begin{Def}[\bf Clifford unitary]
An  $n$-qudit unitary $U$ is  Clifford, if conjugation by $U$ maps every Weyl operator to another Weyl operator, up to a phase.
\end{Def}

Clifford unitaries map stabilizer states to stabilizer states. In general, we have the following definition of a stabilizer channel.
\begin{Def}[\bf Stabilizer channel]
A  quantum  channel is a stabilizer channel if it maps stabilizer states 
to stabilizer states.
\end{Def}

The Gottesman-Knill theorem is an important result about stabilizer states~\cite{gottesman1998heisenberg}. It states that circuits comprised of products of Clifford unitaries acting on stabilizer vectors  
can be efficiently simulated  on a classical computer. 
This means that in order for a quantum processor to achieve an advantage over classical  computation, non-stabilizer states are necessary.  In the recent literature, the property of being a non-stabilizer state has been called ``magic''~\cite{bravyi2005universal}.

In addition to the  Weyl operators  $w(\vec p, \vec q)$,
we introduce phase-space point operators  $\set{T(\vec{p}, \vec q)}$, which are  the symplectic Fourier transform of the Weyl operators,
\begin{eqnarray}\label{eq:phase_poin}
T(\vec{p},\vec q)
=\frac{1}{d^n}
\sum_{(\vec{u}, \vec v)\in V^n}
\chi(\inner{(\vec p, \vec q)}{(\vec u, \vec v)}_s)\,
w(\vec{u}, \vec v) \; .
\end{eqnarray}
\begin{Rem}\label{rem:phas_op}
The set of phase space point operators $\set{T(\vec{p}, \vec q)}_{(\vec{p}, \vec q)\in V^n}$ satisfies three properties when the local dimension $d$ is an odd prime \cite{Gross06}:
\begin{enumerate}[(1)]
\item{}
 $\set{T(\vec{p}, \vec q)}_{(\vec{p}, \vec q)\in V^n}$ forms a Hermitian, orthonormal basis with respect to the inner product defined by 
$\inner{A}{B}=\frac{1}{d^n}\Tr{A^\dag B}$.
\item{}
$T(\vec{0},\vec 0)=\sum_{\vec{j}}\ket{-\vec{j}}\bra{\vec{j}}$ in the Pauli $Z$ basis.
\item{}
$T(\vec{p}, \vec q)=w(\vec{p}, \vec q)T(\vec{0}, \vec 0)w(\vec{p}, \vec q)^\dag$.
\end{enumerate}
\end{Rem}

\begin{Def}[\bf Discrete Wigner function] \label{Def:DiscreteWignerFunction}
For any $n$-qudit state $\rho\in \mathcal{D}(\mathcal{H}^{\ot n})$, the discrete Wigner function $W_{\rho}:V^n\to \real$ is defined as
\begin{eqnarray*}
W_{\rho}(\vec{p}, \vec q)
= \langle \rho, T(\vec p, \vec q)  \rangle
=\frac{1}{d^n}\Tr{\rho T(\vec{p}, \vec q)} \; .
\end{eqnarray*}
\end{Def}
Hence, the quantum state $\rho$ can be written as $\rho=\sum_{\vec{p}, \vec q}W_{\rho}(\vec{p}, \vec q)T(\vec{p}, \vec q)$.
One important result about the discrete Wigner function is the Discrete Hudson Theorem.

\begin{lem}[\bf Discrete Hudson Theorem \cite{Gross06}]
Given an $n$-qudit system with the local dimension being an odd prime number, a pure state $\psi$ is a stabilizer state, if and only if  the Wigner function $W_{\psi}$ is non-negative. 
\end{lem}
This result is a discrete version of Hudson's theorem 
\cite{Hudson74,Soto83}.
It remains challenging to generalize the discrete Wigner function to any  local dimension which is not an odd prime number~\cite{RaussendorfPRX15,RaussendorfPRA17,Love17,RaussendorfPRA20}.

\section{Mean state}\label{Sect:MeanState}
Here we define a \textit{mean state} (MS) for a given quantum state, and we show that the mean state is an MSPS in the sense of Definition \ref{def:MSPS}.  We call $\CMM(\rho)$ the mean state, because we use it to define
the mean-value vector of the state $\rho$ in  \eqref{eq:mean_value}, and the zero-mean state in Definition \ref{Def:Zero_mean}.

\begin{Def}[\bf Mean state]\label{def:mean_state}
Given an $n$-qudit state $\rho$,  the mean state $\CMM(\rho)$ is the 
operator with the characteristic function: 
\begin{align}\label{0109shi6}
\Xi_{\CMM(\rho)}(\vec p, \vec q) =
\left\{
\begin{aligned}
&\Xi_\rho ( \vec p, \vec q) \; , && |\Xi_\rho ( \vec p, \vec q)|=1 \; ,\\
& 0 \; , && |\Xi_\rho (  \vec p, \vec q)|<1 \; .
\end{aligned}
\right.
\end{align}
\end{Def}

\begin{prop}\label{0109lem1}
For any $n$-qudit quantum state $\rho$,
$\CMM(\rho)$ is an MSPS.
\end{prop}
\begin{proof}
Let
\begin{align}\label{0109shi7}
S = \{\vec x\in V^n: |\Xi_\rho ( \vec x)|=1\} \; .
\end{align}
Based on Lemmas \ref{0106lem1}, \ref{0106lem2} and \ref{0106lem3}, we have three properties:
(a) First, $\forall \vec x,\vec y\in S$, $\ep{\vec x,\vec y}_s=0$, i.e., $[w(\vec x), w(\vec y)]=0$;
(b) second, $\forall \vec x\in S$, $t\vec x\in S$ for every $t\in \Z_d$;
(c) and third, $\forall \vec x_1,...,\vec x_r\in S$, $\vec x_1+\cdots +\vec x_r \in S$.
    Hence, the subgroup $S$ has $r$ generators, i.e.,  there exists $  \vec x_1,...,\vec x_r $ in $S$ with $|S| = d^r$ such that 
\begin{align*}
S = \left\{\sum_{j=1}^r t_j \vec x_j : t_1,...,t_r\in \Z_d\right\} \; ,
\end{align*}
and for some $k_1,...,k_r\in \mathbb{Z}_d$, one has
\begin{eqnarray*}
\Xi_{\rho}\left(\sum_jt_j\vec{x}_j\right)=
\chi\left(\sum_jt_jk_j\right) \; .
\end{eqnarray*}
There exists some Clifford unitary $U$ such that
$Uw(\vec x_j) U^\dag= Z_j$, where $Z_j$ is the Pauli $Z$ operator acting on the $j$-th qudit. 
Hence, 
\begin{eqnarray*}
U\CMM(\rho)U^\dag=
\frac{1}{d^n}
\sum_{t_1,..,t_r}
\chi\left( \sum_jt_jk_j\right)
\left(\prod_j Z^{t_j}_j\right)
=\frac{1}{d^{n-r}}
\Pi^r_{j=1}
\mathbb{E}_{t_j}
[\chi(k_j)Z_j]^{t_j} \; ,
\end{eqnarray*}
is a projection multiplied by a factor $1/d^{n-r}$, where $\mathbb{E}_{t_j}=\frac{1}{d}\sum_{t_j\in\mathbb{Z}_d}$.
\end{proof}

\begin{lem}\label{0106lem1}
Let $\rho$ be an $n$-qudit state,  
such that there exists some  $(\vec p, \vec q)\in V^n$ with
$
|\trace[\rho w(\vec p, \vec q)]|=1
$.
 Then $\rho$ commutes with $w(\vec p, \vec q)$ and $range(\rho)\subseteq range(P_k)$, 
 where $P_k$ is the projection onto the eigenspace of $w(\vec p, \vec q)$, i.e., 
\begin{eqnarray*}
[\rho, w(\vec{p},\vec q)]=0 \; ,
\quad
\text{and} 
\quad
P_k \rho P_k = \rho \; ,
\end{eqnarray*}
for some $1\le k\le d$, where we assume
$
w(\vec p,\vec q) = \sum_j \chi(j) P_j
$
is the spectral decomposition of $w(\vec p,\vec q)$.
Thus $w(\vec p,\vec q)\rho =\chi(k) \rho $.
Moreover,
we also have
\begin{align*}
|\trace[\rho w(t \vec p,t\vec q)]|=1 \;,\quad \forall t\in \Z_d \;.
\end{align*}
\end{lem}
\begin{proof}

Since $w(\vec p,\vec q) = \sum_j \chi(j) P_j$,
we have
\begin{align}\label{0209shi1}
\trace[\rho w(\vec p,\vec q)] = \trace[ \rho\sum_j \chi(j) P_j] =  \sum_j \chi(j) \trace[\rho P_j] \;.
\end{align}
As $\sum_j\trace[\rho P_j]=1$, we infer that the right-hand side of \eqref{0209shi1} is a convex combination of numbers $\chi(j)$ of modulus 1.
While $|\trace[\rho w(\vec p, \vec q)]|=1$, it means \eqref{0209shi1} has modulus 1. Thus $\trace[\rho P_k] = 1$ for some $1\le k\le d$. 
Since $\rho$ is a state, we have $P_k \rho P_k = \rho$.
\end{proof}

\begin{lem}\label{0106lem2}
Let $\rho$ be an $n$-qudit  quantum state  with
\[
|\trace[\rho w(\vec p_1,\vec q_1)]|=|\trace[\rho w(\vec p_2,\vec q_2)]| = 1
\]
for some $(\vec p_1,\vec q_1), (\vec p_2,\vec q_2)\in V^n$.
Then the Weyl operators commute, i.e.,
\begin{align*}
[w(\vec p_1,\vec q_1),w(\vec p_2,\vec q_2) ]=0 \;.
\end{align*}
\end{lem}
\begin{proof}
By Lemma \ref{0106lem1},
$\rho$ commutes with both $w(\vec p_1,\vec q_1)$ and $w(\vec p_2,\vec q_2) $,
and
\begin{align*}
w(\vec p_1,\vec q_1)\rho  = \chi(k_1) \rho \quad \text{ and }\quad w(\vec p_2,\vec q_2)\rho  = \chi(k_2)\rho
\end{align*}
for some $k_1,k_2\in \Z_d$.
Therefore
\begin{align*}
\chi(k_1) \chi(k_2)\rho =&  w(\vec p_1,\vec q_1)  w(\vec p_2,\vec q_2) \rho \\
=& \chi( \ep{ (\vec p_1,\vec q_1),(\vec p_2,\vec q_2) }_s  )   w(\vec p_2,\vec q_2)  w(\vec p_1,\vec q_1) \rho\\
=& \chi( \ep{ (\vec p_1,\vec q_1), (\vec p_2,\vec q_2)}_s  ) \chi(k_1) \chi(k_2) \rho.
\end{align*}
where the second equality comes from
\eqref{0106shi2}, or from \eqref{0106shi2_2} when $d=2$.
Therefore
$\ep{ (\vec p_1,\vec q_1), (\vec p_2,\vec q_2)}_s=0$,
i.e.,
\begin{eqnarray}
    [w(\vec p_1,\vec q_1) ,w(\vec p_2,\vec q_2)]=0.
\end{eqnarray}

\end{proof}

\begin{lem}\label{0109lem2}
Let $\rho$ be an $n$-qudit state with   
\be
[w(\vec p_1,\vec q_1),w(\vec p_2,\vec q_2)]\neq 0\;,
\quad\text{and}\quad
|\Tr{\rho w(\vec p_1,\vec q_1)}| = 1\;,
\ee
for some $(\vec p_1,\vec q_1), (\vec p_2,\vec q_2) \in  V^n$.
Then  $  \Tr{\rho w(\vec p_2,\vec q_2)} =0$.

\end{lem}
\begin{proof}
By Lemma \ref{0106lem1},
$\rho$ commutes with $w(\vec p_1,\vec q_1)$ 
and $ w(\vec p_1,\vec q_1)\rho = \chi(k) \rho  $ for some $1\le k\le d$.
Thus
\begin{align*}
\Tr{\rho w(\vec p_2,\vec q_2)} 
= & \trace[ \rho w(\vec p_1,\vec q_1) w(\vec p_2,\vec q_2) w(\vec p_1,\vec q_1)^\dag]\\
&= \chi(  \ep{(\vec p_1,\vec q_1),(\vec p_2,\vec q_2)}_s )\Tr{\rho w(\vec p_2,\vec q_2) } \;.
\end{align*}
As $ [w(\vec p_1,\vec q_1),w(\vec p_2,\vec q_2)]\neq 0$, so  
$\chi(  \ep{(\vec p_1,\vec q_1), (\vec p_2,\vec q_2)}_s )\neq 1$,   and $\Tr{ \rho w(\vec p_2,\vec q_2) }=0$.
\end{proof}

\begin{lem}\label{0106lem3}
Let $\rho$ be an $n$-qudit state with $|\trace[\rho w(\vec p_1,\vec q_1)]|=\cdots =|\trace[\rho w(\vec p_r,\vec q_r)]| = 1 $
for some $(\vec p_1,\vec q_1),..., (\vec p_r,\vec q_r) \in V^n$.
Then
\begin{align*}
\Tr{\rho w(\vec p_1+\cdots + \vec p_r, \vec q_1+\cdots + \vec q_r)} = \prod_{ j=1}^r \Tr{\rho w(\vec p_j,\vec q_j)} \;.
\end{align*}

\end{lem}
\begin{proof}
By Lemma \ref{0106lem1}, there are $ k_1,...,k_r\in \Z_d$ such that
\[  \rho w( \vec p_j,\vec q_j) = \chi(k_j) \rho \;,\quad j=1,...,r \;.\]
By Lemma \ref{0106lem2},
$w(\vec p_1,\vec q_1),..., w(\vec p_r, \vec q_r)$ commute with each other.
Thus
\begin{align*}
&\Tr{\rho w(\vec p_1+\cdots + \vec p_r, \vec q_1+\cdots + \vec q_r)} \\
=& \Tr{\rho w(\vec p_1, \vec q_1)\cdots w(\vec p_r, \vec q_r) }\\
=& \chi(k_1)\cdots\chi(k_r) \trace[\rho] \\
=&\prod_{ j=1}^r \Tr{\rho w(\vec p_j , \vec q_j)} \;,
\end{align*}
and the proof is complete.
\end{proof}
 We also find that the 
map $\CMM: \mathcal{D}(\mathcal{H}^{\ot n})\to MSPS$ commutes with all Clifford unitaries.
\begin{lem}\label{lem:comm_mean}
For any $n$-qudit state $\rho$ and any Clifford unitary $U$, 
\begin{eqnarray*}
\CMM(U\rho U^\dag)
=U\CMM(\rho) U^\dag \;.
\end{eqnarray*}
\end{lem}
\begin{proof}
This result comes directly from the fact that Clifford unitary $U$ always maps Weyl operators to Weyl operators.
\end{proof}

\subsection{Mean-value vector and zero-mean state }
For a classical multi-variable random variable $\vec{X}\in \real^r$, 
its characteristic function is
\begin{eqnarray*}
\phi_{\vec X}(\vec t)=\mathbb{E}_{\vec{X}}\exp(i\vec{t}\cdot \vec{X}) \;,\quad \forall \vec{t}=(t_1,..,t_r)\in \real^r \;,
\end{eqnarray*}
and its  mean-value vector $\vec{\mu}$ is equal to the  gradient of $\phi_{X}(\vec t)$ at $\vec t= \vec 0$, i.e., 
\begin{eqnarray*}
\vec{\mu}=\left(-i\left.\frac{d}{d t_j}\phi_{\vec X}(\vec t)\right|_{\vec t=\vec 0}\right)^r_{j=1} \;.
\end{eqnarray*}
If $\vec X$ is zero-mean, then $\vec{\mu}=(0,...,0)$.

For the quantum case, we also need to define the 
zero-mean states for the quantum central limit theorem. 
Given an $n$-qudit state $\rho$,  its MS $\mathcal{M}(\rho)$ has
characteristic function 
\begin{eqnarray}\label{0212shi1}
\Xi_{\mathcal{M}(\rho)}\left(\sum^r_{i=1}t_i(\vec{p}_i, \vec{q}_i)\right)
=\Pi^r_{i=1}\chi(t_ik_i) \;,
\end{eqnarray}
where we assume the abelian group of $\CMM(\rho)$ is generated by  $\set{w(\vec{p}_i, \vec{q}_i)}_{i\in[r]}$,
and 
$
\{\chi(k_i)\}$. 
Similar to the classical case, 
we define the mean-value vector of a state $\rho$ w.r.t. the generators $\set{w(\vec{p}_i, \vec{q}_i)}_{i\in[r]}$ as 
\begin{eqnarray}\label{eq:mean_value}
\vec{\mu}_{\CMM(\rho)}
:=(k_1,...,k_r) \mod d  \;.
\end{eqnarray}
Therefore, 
we have the following definition of zero-mean states.

\begin{Def}[\bf Zero-mean state]\label{Def:Zero_mean}
An $n$-qudit state $\rho$ is called a zero-mean state if its MS $\mathcal{M}(\rho)$ has mean-value 
vector $\vec{\mu}_{\CMM(\rho)}=(0,..., 0)\mod d$, or equivalently, 
the characteristic function of $\mathcal{M}(\rho)$ takes values in $\{0,1\}$. 
\end{Def}
When $\rho$ is not a zero-mean state, there exists some Weyl operator $w(\vec p, \vec q)$ such that $w(\vec p, \vec q)\rho w(\vec p, \vec q)^\dag$
is a zero-mean state.

\begin{lem}\label{lem:zero_mean_pau} 
Given any quantum  state $\rho$, there exists a Weyl operator $w(\vec{p}, \vec q)$ such that $w(\vec{p} ,\vec q)\rho w(\vec{p}, \vec q)^\dag$ has zero mean.
\end{lem}
\begin{proof}
We follow the notation in \eqref{0212shi1}.
Let $U$ be a Clifford unitary such that 
$UZ_iU^\dag=w(\vec{p}_i, \vec{q}_i)$. Since
$\Tr{\rho w(\vec{p}_i, \vec q_i)}=\Tr{\mathcal{M}(\rho) w(\vec{p}_i, \vec{q}_i)}=\chi(k_i)$, 
we have $\Tr{U^\dag \rho UZ_i}=\chi(k_i)$, i.e., $\Tr{\mathcal{M}(U^\dag \rho U)Z_i}=\chi(k_i)$ . Hence
\begin{eqnarray*}
\mathcal{M}(U^\dag \rho U)
=\frac{1}{d^{n-r}} \left(\bigotimes^r_{i=1} \proj{k_i} \right)\ot I_{n-r} \;.
\end{eqnarray*}
Take $w(\vec{p}, \vec q)=U \left(\ot^r_{i=1}X^{-k_i} \ot I_{n-r}\right) U^\dag$ (up to some $d$-th root of unity), 
then we have 
\begin{eqnarray*}
\Xi_{\mathcal{M}(w(\vec{p}, \vec q)\rho w(-\vec{p}, -\vec q))}\left(\sum^r_{i=1}t_i(\vec{p}_i, \vec q_i)\right)
=\Xi_{\mathcal{M}(\rho) }\left(\sum^r_{i=1}t_i(\vec{p}_i, \vec q_i)\right)\Pi^r_{i=1}\chi(-t_ik_i) =1 \;.
\end{eqnarray*}
Therefore $w(\vec{p} ,\vec q)\rho w(\vec{p}, \vec q)^\dag$ has zero mean.
\end{proof}

\subsection{Extremality of MSPS}
In order to discuss the 
extremality of the MSPS, we introduce
an important family of measures, the quantum R\'enyi relative entropy. This allows us to quantify the difference between quantum states and to analyze extremality.

The quantum R\'enyi entropy is a family of entropy measures to study quantum states, which is defined as follows,
\begin{eqnarray}\label{eq:RE}
H_{\alpha}(\rho)
:=\frac{1}{1-\alpha}\log \Tr{\rho^\alpha} \;, 
\end{eqnarray}
for any $\alpha\in[0,+\infty]$.
In particular, 
$
\lim_{\alpha\rightarrow 1} H_\alpha(\rho) = H(\rho),
$
is the von Neumann entropy of $\rho$, 
$
H_{+\infty} (\rho) = \lim_{\alpha\rightarrow +\infty} H_\alpha(\rho) = -\log \lambda_{\max} (\rho)$,
where $\lambda_{\max}$ is the maximal eigenvalue of $\rho$.

\begin{Def}[\bf Quantum  R\'enyi relative entropy \cite{Hiai11,Martin13}]
Given two quantum states $\rho$ and $\sigma$, the quantum R\'enyi relative entropy of $\rho$ with respect to $\sigma$ is 
\begin{eqnarray*}
    D_{\alpha}(\rho||\sigma)=\frac{1}{\alpha-1}\log\Tr{\left(\sigma^{\frac{1-\alpha}{2\alpha}}\rho\sigma^{\frac{1-\alpha}{2\alpha}}\right)^{\alpha}} \;,
\end{eqnarray*}
where $\alpha\in[0,+\infty]$. In the cases $\alpha=0,1,\infty$, the R\'enyi relative entropy is defined as a limit. 
\end{Def}
The  family of quantum R\'enyi relative entropies includes some well-known entropy measures. 
For example, for $\alpha=1$, 
\begin{align*}
    \lim_{\alpha\to 1}D_{\alpha}(\rho||\sigma)=D(\rho||\sigma)=\Tr{\rho\log \rho}-\Tr{\rho\log \sigma}\;.
\end{align*} 
For $\alpha=0$, 
\begin{align*}
    \lim_{\alpha\to 0}D_{\alpha}(\rho||\sigma)=D_0(\rho||\sigma)=-\log\Tr{\Pi_{\rho}\sigma}\;,
\end{align*}
where $\Pi_{\rho}$ is the projector to the support of $\rho$. For $\alpha=\infty$,
\begin{align*}
    \lim_{\alpha\to\infty}D_{\alpha}(\rho||\sigma)=D_{\infty}(\rho||\sigma)=\min\set{\lambda:\rho\leq 2^{\lambda}\sigma}\;.
\end{align*}
The quantum R\'enyi relative entropy $D_{\alpha}\geq 0$, and equality holds iff $\rho=\sigma$ for $\alpha>0$, so R\'enyi relative entropy is useful to quantify the distance of quantum states. The quantum R\'enyi relative entropy $D_{\alpha}$ is additive under tensor products and monotone under quantum channels for $\alpha\geq 1/2$ \cite{Tomamichel2015quantum}. The property of monotonicity will be a useful property for the later part in this work.

Here we find that the MS is the closest MSPS in quantum R\'enyi relative entropy.
\begin{thm}[\bf Extremality of MSPS]\label{thm:main1}
Given an $n$-qudit state $\rho$, we have 
\begin{align*}
\min_{\sigma\in MSPS}D_{\alpha}(\rho||\sigma)
=D_{\alpha}(\rho|| \CMM(\rho)) =H_{\alpha}(\CMM(\rho))-H_{\alpha}(\rho) \;,
\end{align*}
for any $\alpha\geq 1$.
Moreover, $\CMM(\rho)$ is the unique minimizer, that is,  for any other state $\sigma\in MSPS$ with $\sigma\neq \CMM(\rho)$,  we have 
\begin{align*}
D_{\alpha}(\rho||\sigma)>D_{\alpha}(\rho||\CMM(\rho)) \;.
\end{align*}
\end{thm}
\begin{proof}

(1) We first prove the second equality $D_{\alpha}(\rho||\CMM(\rho))=H_{\alpha}(\CMM(\rho))-H_{\alpha}(\rho)$.
We observe that $\mathrm{range}(\rho)\subseteq \mathrm{range} (\CMM(\rho))$, i.e.,
\begin{eqnarray*}
\rho\leq \mathrm{rank}(\CMM(\rho))\CMM(\rho) \;,
\end{eqnarray*}
where $ \CMM(\rho)$ is an MSPS.

Let $S$ be the abelian Weyl group associated with $\mathcal{M}(\rho)$ and assume $|S|=d^r$. 
Then $\mathrm{rank}(\CMM(\rho))=d^{n-r}$.
Without loss of generality, assume $S$ is generated by the Weyl operators $\set{w(\vec p_i, \vec q_i)}^r_{i=1}$. 
Let $U$ be a Clifford unitary that  maps each $w(\vec p_i, \vec q_i)$ to $Z_i$, where $Z_i$ is the Pauli $Z$ operator 
on the $i$-th position. Hence, the abelian Weyl group $S_1$ associated with $\CMM(U\rho U^\dag)$
is generated by $\set{Z_1,...,Z_r}$.
Due to Lemma \ref{0109lem2},
when $w(\vec p, \vec q)\notin S_1$,
$ \Xi_{ U\rho U^\dag} (\vec p,\vec q) \neq 0$  implies that $w(\vec p,\vec q ) $ commutes with the group $S_1$, i.e., $ q = (0,...,0, q_{r+1},...,q_n)$.
Let us denote
\begin{align*}
S_2 = \{w(\vec p,\vec q):  \vec p = (0,...,0, x_{r+1},...,x_n) \text{ and } \vec q = (0,...,0, y_{r+1},...,y_n)\} \;,
\end{align*}
then 
\begin{align}\label{0202shi1}
\CMM(U\rho U^\dag) = \mathbb{E}_{w(\vec p,\vec q)\in S_2} w(\vec p,\vec q) U\rho U^\dag w(\vec p,\vec q)^\dag \;,
\end{align}
where $\mathbb{E}_{w(\vec p,\vec q)\in S_2}=\frac{1}{|S_2|}\sum_{w(\vec p,\vec q)\in S_2}$.
Hence, we have 
\begin{eqnarray*}
|S_2|\CMM(U\rho U^\dag) 
\geq U\rho U^\dag \;.
\end{eqnarray*}
By Lemma \ref{lem:comm_mean}, we have  $|S_2|\CMM(\rho)\geq \rho$, 
i.e., $d^{n-r}\CMM(\rho)\geq \rho$. Thus 
we have 
\begin{eqnarray*}
\rho\leq \mathrm{rank}(\CMM(\rho))\CMM(\rho) \;,
\end{eqnarray*}
as $\mathrm{rank}(\CMM(\rho))=d^{n-r}$.

Since $\mathrm{range}(\rho) \subseteq  \mathrm{range} (\CMM(\rho))$, 
we have $D_{\alpha}(\rho||\CMM(\rho))<+\infty$.
Denote the projection $P:=\mathrm{rank}(\CMM(\rho))\CMM(\rho)$, i.e., 
$ \CMM(\rho)=\frac{1}{d^{n-r}}P$.  Then we have 
\begin{eqnarray*}
D_{\alpha}(\rho||\CMM(\rho))
=\log d^{n-r}
+\frac{1}{\alpha-1}\log\Tr{(P\rho P)^\alpha}
=H_{\alpha}(\CMM(\rho))
-H_{\alpha}(\rho) \;,
\end{eqnarray*}
where the second equality comes from the fact that $P \rho P = \rho $.

(2) If $Q$ is a minimal stabilizer projection associated with some abelian Weyl group $S_1$ such that
\begin{align}\label{0203shi1}
D_{\alpha} \left(\rho \left\| \frac{Q}{\Tr{Q}}\right.\right)<\infty \;,
\end{align}
then we show that $S_1\subseteq S$.
Moreover, we show  $ Q\CMM( \rho) = \CMM( \rho)$,
i.e.,
the range projection of $\CMM( \rho) \le Q$.

The finiteness of the relative entropy \eqref{0203shi1} implies that $\mathrm{range}(\rho)\subseteq  \mathrm{range}(Q)$,
or equivalently,
$  \rho Q = \rho$.
Therefore $\Tr{\rho Q}=\Tr{\rho}= 1$.
Since  $Q$ is a minimal stabilizer
projection of $S_1$,
\begin{align*}
Q= \frac{1}{|S_1|}\sum_{ w(\vec p, \vec q) \in S_1 } e^{\theta(\vec p, \vec q)}
 w(\vec p, \vec q) \;,
\end{align*}
where each $ |e^{\theta(\vec p, \vec q)}| = 1$.
Then
\begin{align*}
1=\Tr{\rho Q} =& \Tr{\rho\frac{1}{|S_1|} \sum_{w(\vec p, \vec q)\in S_1}  e^{\theta(\vec p, \vec q)} w(\vec p, \vec q) }\\
=&\frac{1}{|S_1|} \sum_{w(\vec p, \vec q)\in S_1} e^{\theta(\vec p, \vec q)} \Tr{\rho w(\vec p, \vec q) }\\
=& \frac{1}{|S_1|} \sum_{w(\vec p, \vec q)\in S_1} e^{\theta(\vec p, \vec q)}
 \Xi_{\rho}(-\vec p, -\vec q) \;,
\end{align*}
which implies that  $|\Xi_{\rho}(-\vec p, -\vec q)|=1$ for every $w(\vec p, \vec q)\in S_1$, 
and hence  $S_1\subseteq S$.
Moreover,
it also implies that $e^{\theta(\vec p, \vec q)} = (\Xi_{\rho}(-\vec p,-\vec q))^{-1} = \Xi_{\rho}(\vec p, \vec q)$ for every $w(\vec p, \vec q)\in S_1$.
Hence
\begin{align*}
\Tr{ \CMM(\rho) Q} = & \Tr{\CMM(\rho)  \frac{1}{|S_1|}\sum_{ w(\vec p, \vec q) \in S_1 } \Xi_{\rho}(\vec p, \vec q)
 w(\vec p, \vec q)}\\
=& \frac{1}{|S_1|} \sum_{ w(\vec p, \vec q) \in S_1 }\Xi_{\rho}(\vec p, \vec q) \Tr{ \CMM(\rho)  w(\vec p, \vec q) } \\
=& \frac{1}{|S_1|} \sum_{ w(\vec p, \vec q) \in S_1 } \Xi_{\rho}(\vec p, \vec q) \Xi_{\rho}(-\vec p, -\vec q) = 1
=\Tr{\CMM( \rho)} \;.
\end{align*}
From this we infer that  $\CMM( \rho)Q = \CMM( \rho)$.

Therefore, for any MSPS $\tau$ with $\tau\neq \CMM(\rho)$ and $D_{\alpha}(\rho||\tau)<+\infty$, 
one has 
\begin{align*}
    \mathrm{range}(\CMM(\rho))
    &\subsetneq \mathrm{range}(\tau)\;,
    \quad\text{and thus}\quad
    \mathrm{rank}(\CMM(\rho))<\mathrm{rank}(\tau)\;.
\end{align*}
Thus
\begin{align*}
D_{\alpha} \left(\rho 
 \left\| \tau\right.\right) 
= \log (\mathrm{rank}(\tau))-H_{\alpha}(\rho)
>  \log  (\mathrm{rank}(\CMM(\rho) )) -H_{\alpha}(\rho)
= D_{\alpha}(\rho||\CMM(\rho)) \;,
\end{align*}
and the proof is complete.
\end{proof}

Based on the above result, we can rewrite the quantum R\'enyi   entropy as follows
\begin{eqnarray*}
H_{\alpha}(\CMM(\rho))=H_{\alpha}(\rho)+D_{\alpha}(\rho||\CMM(\rho)) \;.
\end{eqnarray*}
This equation shows the extremality of MSPS with respect to quantum R\'enyi  entropy:
among all quantum states having the same  MS up to Clifford 
conjugate, the MSPS $\CMM(\rho)$ attains the maximal value for  quantum R\'enyi  entropy, 
which we call "maximal entropy principle in DV systems."
Recall the extremality of Gaussian states in CV systems, i.e.,  among all states
having a given covariance matrix,  Gaussian states attain the maximum von Neumann entropy  \cite{Holevo99,WolfPRL06}.
Hence the above theorem is the discrete version of the extremality of Gaussian states with the same covariance matrix in CV systems. 
In fact, the extremality of MSPSs holds for any Schur concave function. To introduce the Schur concave function, let us first introduce the concept of 
majorization, which we elaborate on in Appendix \ref{0107app1}.

\begin{Def}[\bf Majorization \cite{MOA79}]\label{def:major}
Given two probability vectors $\vec p=\set{p_i}_{i\in[n]}$ and $\vec q=\set{q_i}_{i\in[n]}$, $\vec p$ is said to
be majorized by $\vec q$, denoted as  $\vec p\prec \vec q$,  if 
\begin{eqnarray*}
\sum^k_{i=1}p^{\downarrow}_i
&\leq& \sum^k_{i=1} q^{\downarrow}_i \;, ~~\forall 1\leq k\leq n-1 \;,\\
\sum^n_{i=1}p_i&=&\sum^n_{i=1} q_i=1 \;,
\end{eqnarray*}
where $p^{\downarrow}$ is the vector obtained by permuting the components of $p$ to be in decreasing order, i.e., 
$p^{\downarrow}=(p^{\downarrow}_1,p^{\downarrow}_2,...,p^{\downarrow}_n)$ and $p^{\downarrow}_1\geq p^{\downarrow}_2\geq...\geq p^{\downarrow}_n$.
\end{Def}

\begin{Def}[\bf Schur concavity \cite{MOA79}]
A function $f:\real^n_+\to \real$ is said to be Schur concave if $\vec p\prec \vec q$ implies that $f(\vec p)\geq f(\vec q)$.
A function is called strictly Schur concave if $\vec p\prec \vec q$ implies that $f(\vec p)> f(\vec q)$ except for $\vec p= \vec q$.
\end{Def}

Let us consider two well-known examples of Schur-concave functions: the subentropy and the generalized quantum R\'enyi entropy.  Both of these functions play an important role in quantum information theory.

\begin{Def}[\bf Subentropy \cite{datta2014suben}]\label{def:suben}
Given a state $\rho$ with eigenvalues $\set{\lambda_i}^d_{i=1}$, the subentropy of
$\rho$ is 
\begin{eqnarray}
Q(\rho):=-\sum^d_{i=1}\frac{\lambda^d_i}{\Pi_{j\neq i}(\lambda_i-\lambda_j)}\log\lambda_i.
\end{eqnarray}
\end{Def}

The  generalized quantum R\'{e}nyi entropy is an extension of quantum R\'{e}nyi entropy~\eqref{eq:RE}.

\begin{Def}[\bf Generalized quantum R\'{e}nyi entropy \cite{BHOW15}]\label{def:gen_ren_en}
For  $\alpha\in[-\infty,+\infty]$,
 the generalized R\'enyi entropy $H_\alpha(\rho)$ for a quantum state $\rho$ with eigenvalues $\set{\lambda_i}_i$ is 
\begin{eqnarray*}
H_{\alpha}(\rho)
=\frac{sgn(\alpha)}{1-\alpha}\log {\sum_i\lambda^{\alpha}_i} \;,
\end{eqnarray*}
with
 $sgn(\alpha)=1$ for $\alpha\geq 0$, $sgn(\alpha)=-1$ for $\alpha<0$.
 Here, $\alpha=0,+\infty,-\infty$ are defined by limits.
\end{Def}

\begin{thm}[\bf Generalized extremality of MSPS]\label{thm:stab_proj}
Let   $\rho $ be an $n$-qudit state, we have 
\begin{eqnarray*}
\vec{\lambda}_{\mathcal{M}(\rho)}
\prec \vec{\lambda}_{\rho}\;,
\end{eqnarray*}
where $\vec{\lambda}_{\rho}$ is the vector of eigenvalues of $\rho$, and $\CMM(\rho)$ is the MS of $\rho$.
For any Schur-concave function $f$,  
\begin{eqnarray*}
f(\vec{\lambda}_{\rho})\leq f(\vec{\lambda}_{\mathcal{M}(\rho)}) \;.
\end{eqnarray*}
\end{thm}
\begin{proof}
Since $f$ is Schur concave, $f(\vec{\lambda}_{\rho})\leq f(\vec{\lambda}_{\mathcal{M}(\rho)})$ follows directly from $\vec{\lambda}_{\mathcal{M}(\rho)}
\prec \vec{\lambda}_{\rho}$, which we now prove.
Denote by $S $ the abelian Weyl group associated with $\CMM(\rho)$,
and assume $|S| = d^r$.
Then $S$ is generated by $r$ Weyl operators,
say $\set{w(\vec p_i, \vec q_i)}^r_{i=1}$. 
Let $U$ be a Clifford unitary that maps each  $w(\vec p_i, \vec q_i)$ to $Z_i$, where $Z_i$ is the Pauli $Z$ operator 
on the $i$-th qudit. Hence, the abelian Weyl group $S_1$ associated with $\CMM(U\rho U^\dag)$
is generated by $\set{Z_1,...,Z_r}$.

Denote
\begin{align*}
S_2 = \{w(\vec p,\vec q):  \vec p = (0,...,0, p_{r+1},...,p_n) \text{  and } \vec q = (0,...,0, q_{r+1},...,q_n)\} \;.
\end{align*}
As a consequence of Lemma \ref{0109lem2},
the characteristic function $ \Xi_{ U\rho U^\dag} (\vec p,\vec q) \neq 0$ for some $(\vec p,\vec q )$ implies that $w(\vec p,\vec q ) $ commutes with all the elements in  group $S_1$. 
Then 
\begin{align}
\CMM(U\rho U^\dag) = \mathbb{E}_{w(\vec p,\vec q)\in S_2} w(\vec p,\vec q) U\rho U^\dag w(\vec p,\vec q)^\dag \;,
\end{align}
where $ \mathbb{E}_{w(\vec p,\vec q)\in S_2}=\frac{1}{|S_2|}\sum_{w(\vec p,\vec q)\in S_2}$.

Let 
\begin{align*}
U\rho U^\dag= \sum_i \lambda_i | \varphi_i \rangle \langle\varphi_i| \;,\\
\CMM( U\rho U^\dag) = \sum_i \mu_i | \psi_i\rangle \langle \psi_i| \;,
\end{align*}
be the spectral decompositions. 
Since $\CMM(U\rho U^\dag)=U\CMM(\rho)U^\dag$ for any Clifford unitary,  $\vec{\lambda}$ and $\vec{\mu}$ are the 
spectral vectors of $\rho$ and $\mathcal{M}(\rho)$, respectively.

Let $M$ be the $d^n\times d^n$ matrix with 
\begin{align*}
M_{ij} = \mathbb{E}_{w(\vec p,\vec q)\in S_2} \langle \psi_j| w(\vec p,\vec q) | \varphi_i \rangle \langle\varphi_i| w(\vec p,\vec q)^\dag | \psi_j\rangle \;.
\end{align*}
Since
\begin{align*}
\sum_i M_{ij} &
=\sum_i\mathbb{E}_{w(\vec p,\vec q)\in S_2} \langle \psi_j| w(\vec p,\vec q) | \varphi_i \rangle \langle\varphi_i| w(\vec p,\vec q)^\dag | \psi_j\rangle\\
&=
\mathbb{E}_{w(\vec p,\vec q)\in S_2} \langle \psi_j| w(\vec p,\vec q) I w(\vec p,\vec q)^\dag | \psi_j\rangle =1 \;,\\
\sum_j M_{ij} &=\sum_j\mathbb{E}_{w(\vec p,\vec q)\in S_2} \langle \psi_j| w(\vec p,\vec q) | \varphi_i \rangle \langle\varphi_i| w(\vec p,\vec q)^\dag | \psi_j\rangle\\
&=
\mathbb{E}_{w(\vec p,\vec q)\in S_2} \trace\left[ w(\vec p,\vec q) | \varphi_i \rangle \langle\varphi_i| w(\vec p,\vec q)^\dag\right] = 1 \;,
\end{align*}
then the matrix $M = (M_{ij})$  is a doubly stochastic matrix.
Since
\begin{align*}
\mu_j = & \langle \psi_j|  \CMM( U\rho U^\dag)  | \psi_j\rangle  
= \mathbb{E}_{w(\vec p,\vec q)\in S_2} \langle \psi_j| w(\vec p,\vec q) U\rho U^\dag w(\vec p,\vec q)^\dag | \psi_j\rangle \\
=& \sum_i  \lambda_i \mathbb{E}_{w(\vec p,\vec q)\in S_2} \langle \psi_j| w(\vec p,\vec q) | \varphi_i \rangle \langle\varphi_i| w(\vec p,\vec q)^\dag | \psi_j\rangle 
= \sum_i M_{ij} \lambda_i \;,
\end{align*}
where the second equality comes from \eqref{0202shi1}. Thus, 
we have
\begin{align*}
(\mu_1,...,\mu_{d^n})^T = M^T (\lambda_1,...,\lambda_{d^n}) ^T \;.
\end{align*}
By Chapter 1, \S A.3  in \cite{MOA79}, $
(\mu_1,...,\mu_{d^n}) \prec(\lambda_1,...,\lambda_{d^n})$, completing the proof.
\end{proof}

Since subentropy and generalized quantum R\'{e}nyi entropy  are both  Schur concave functions, we have the following 
corollary directly.
\begin{cor}
Let   $\rho $ be an $n$-qudit state, we have 
\begin{eqnarray}
Q(\rho)\leq Q(\CMM(\rho)),\quad
H_{\alpha}(\rho)\leq H_{\alpha}(\CMM(\rho)),
\end{eqnarray}
where $\alpha\in [-\infty,+\infty]$.
\end{cor}

\subsection{The map \texorpdfstring{$\mathcal{M}$}{M} in Definition \ref{def:mean_state}  destroys magic}
Several measures have been proposed to quantify and to study magic~\cite{Veitch12mag,LeonePRL22,HowardPRL17,BeverlandQST20,SeddonPRXQ21,BravyiPRL16,BravyiPRX16,bravyi2019simulation, Bu19,Bucomplexity22,BuPRA19_stat}. They have been used to quantify the complexity of the classical simulation 
of quantum circuits~\cite{SeddonPRXQ21,BravyiPRL16,BravyiPRX16,bravyi2019simulation,Bu19,Bucomplexity22,bu2022classical} and 
also bound the number of resource gates in unitary  synthesis \cite{HowardPRL17,BeverlandQST20}. To achieve quantum advantage in a DV quantum system, several sampling tasks 
have been proposed \cite{jozsa2014classical,koh2015further,bouland2018complexity, boixo2018characterizing,bouland2018onthecomplexity,bremner2010classical,Yoganathan19}.  In several instances
a computational advantage over classical supercomputers  has been claimed in experimental work~\cite{Google19,zucongzi1,Pansci22}.

Quantum resource theory provides a framework for studying magic resources~\cite{HowardPRL17,Chitambar19}. One important concept in this framework is the resource-destroying map.

\begin{Def}[ \cite{LiuZiwenPRL17}]\label{Def:RD}
Given a chosen subset $\mathcal{F}\subseteq \mathcal{D}(\mathcal{H}^{\ot n})$ , 
a map $\lambda: \mathcal{D}(\mathcal{H}^{\ot n})\to \mathcal{D}(\mathcal{H}^{\ot n})$ is a resource-destroying map for 
$\mathcal{F}$
if its range is $\mathcal{F}$ and  $\lambda^2=\lambda$.

\end{Def}

Natural resource-destroying maps are known for several resource theories, including coherence, asymmetry, and non-Gaussianity (see Table \ref{tab:example}). However, a nontrivial resource-destroying map for magic has remained unknown. We present such a map here.
\begin{prop}\label{prop:RD}
The map $\CMM: \mathcal{D}(\mathcal{H}^{\ot n})\to MSPS$ 
 given by $\rho\to\CMM(\rho)$ is a resource-destroying map for magic. This map   
 satisfies  
$\min_{\sigma\in MSPS}D_{\alpha}(\rho||\sigma)=D_{\alpha}(\rho||\mathcal{M}(\rho))$. 
\end{prop}
\begin{proof}
It is straightforward to verify that the map $\CMM$ satisfies the conditions in Definition \ref{Def:RD}. The fact that this map satisfies
$\min_{\sigma\in MSPS}D_{\alpha}(\rho||\sigma)=D_{\alpha}(\rho||\mathcal{M}(\rho))$ follows directly from Theorem \ref{thm:main1}.
\end{proof}

\begin{table*}[!htbp]
\begin{tabularx}{\textwidth}{lp{0.75\textwidth}}
Theory & resource destroying map\\  \hline\hline
Coherence &  $\Delta(\rho) = \sum_i\bra{i}\rho\ket{i}\proj{i}$, where $\Delta$  is the complete dephasing channel:$\{\ket{i}\}$ w.r.t. the  reference basis \cite{BaumgratzPRL14,Plenio17}.\\
Asymmetry & $\mathcal{G}(\rho) = \int_{G} d\mu(U) U\rho U^\dag$, where the integral is taken over the Haar measure on $G$  \cite{Gour09}. \\
Non-Gaussianity& 
$\lambda(\rho) = \rho_G$, where $\rho_G$ is the Gaussian state with the same mean displacement and covariance matrix as $\rho$ \cite{Marian13}\\
Magic &  $\mathcal{M}(\rho)$, the closest MSPS  (Theorem \ref{thm:main1} and Proposition~\ref{prop:RD} in this work).\\ \hline\hline
\end{tabularx}
\vskip 5pt
\caption{\label{tab:example}Resource theories with a nontrivial resource destroying map.}

\end{table*}

\subsection{The magic gap}
In Definition \ref{def:charfn} we gave the characteristic function of a state. 
Consider the gap between the largest absolute value,
which is 1,  
and the second-largest absolute value in the support of the characteristic function.
We call this the magic gap (or non-stabilizer gap).  It plays a role similar to the spectral gap of a Laplacian between the two smallest eigenvalues.

\begin{Def}[\bf Magic gap]\label{def:ma_gap}
Given an $n$-qudit state $\rho$,  the  magic gap  of $\rho$ is 
\begin{eqnarray*}
MG(\rho)=1-\max_{(\vec{p}, \vec q)\in  \text{Supp}(\Xi_{\rho}): |\Xi_{\rho}(\vec{p},\vec q)|\neq 1}|\Xi_{\rho}(\vec{p},\vec q)| \;.
\end{eqnarray*}
If $\set{(\vec{p},\vec q)\in  \text{Supp}(\Xi_{\rho}): |\Xi_{\rho}(\vec{p}, \vec q)|\neq 1}=\emptyset$, 
then $MG(\rho):=0$, i.e., there is no gap on the support.

\end{Def}

\begin{Def}[ \cite{Bu19}]\label{def:PauRank}
Given an $n$-qudit state $\rho$, the Pauli rank $R_P(\rho)$ of $\rho$ is 
\begin{eqnarray*}
R_P(\rho)=|\text{Supp}(\Xi_\rho)|=|\set{(\vec{p}, \vec q)\in V^n:\Xi_{\rho}(\vec{p}, \vec q)\neq 0}| \;.
\end{eqnarray*}
\end{Def}

\begin{lem}[\cite{Bu19}]\label{lem:pauR}
For a pure state $\psi$, the Pauli rank $R_P(\psi)$ satisfies the following properties : 

(1)  $d^n\leq R_P(\psi)\leq d^{2n}$, and $ R_P(\psi)=d^n$ holds if and only if $\psi$ is a stabilizer state; 

(2) The Pauli rank $ R_P$ is invariant under Clifford unitary conjugation.
\end{lem}

\begin{prop}[\bf Properties of magic gap]\label{prop:magap}
The magic gap ($MG$) satisfies the following properties:

(1) The magic gap of $\rho$ is 0, iff $\rho$ is an MSPS. In addition, for any state $\rho$ that is not an MSPS, we have
\begin{align}
  0< MG(\rho)\leq 1-\sqrt{ \frac{d^n\Tr{\rho^2}-d^k}{R_P(\rho)-d^k}}\;,  
\end{align}
where the Pauli rank $R_P(\rho)=|\text{Supp}(\Xi_{\rho})|$ is given in Definition \ref{def:PauRank}, and 
$d^k$ is the size of the set  $S=\set{(\vec{p}, \vec q):|\Xi_{\rho}(\vec{p}, \vec q)|=1}$.

(2) The $MG$ is invariant under a Clifford unitary acting on $\rho$. 

(3) The magic gap $MG(\rho_1\ot\rho_2)= \min\set{MG(\rho_1),MG(\rho_2)}$ if $\rho_1$ and $\rho_2$ are not MSPS.
\end{prop}

\begin{proof}
Properties (2) and (3) follow directly from the definition and factorization of characteristic function in Definition \ref{def:charfn}, so we only need to prove (1). 
First, $MG(\rho)\geq 0$ by definition. Note
\begin{eqnarray*}
\rho=\mathcal{M}(\rho)+\frac{1}{d^n}\sum_{(\vec{p}, \vec q):0<|\Xi_{\rho}(\vec{p}, \vec q)|<1}\Xi_{\rho}(\vec{p}, \vec q)w(\vec{p}, \vec q) \;,
\end{eqnarray*}
so using Proposition~\ref{0109lem1}, $\mathcal{M}(\rho)$ is an MSPS. Hence $MG(\rho)=0$ iff $\rho$ is an MSPS. 
In addition, the set $S=\set{(\vec{p}, \vec q):|\Xi_{\rho}(\vec{p}, \vec q)|=1}$ has size $d^k$ for some  $0\leq k\leq n$. 
If $k<n$, then
\begin{eqnarray*}
\sum_{(\vec{p}, \vec q):0<|\Xi_{\rho}(\vec{p}, \vec q)|<1}|\Xi_{\rho}(\vec{p}, \vec q)|^2
+d^k=d^n\Tr{\rho^2} \;,
\end{eqnarray*}
and thus
\begin{eqnarray*}
\max_{(\vec{p}, \vec q)\in  \text{Supp}(\Xi_{\rho}): |\Xi_{\rho}(\vec{p}, \vec q)|\neq 1}| \Xi_{\rho}(\vec{p}, \vec q)|^2
\geq \frac{d^n \Tr{\rho^2} -d^k}{R_P(\rho)-d^k} \;.
\end{eqnarray*}
\end{proof}

Since $-\log(1-x)=x+O(x^2)$, let us consider the logarithmic magic gap. 

\begin{Def}[\bf Logarithmic magic gap (LMG)]
Given  an $n$-qudit state  $\rho$,  the logarithmic magic gap (LMG) of  $\rho$ is
\begin{eqnarray*}
LMG(\rho)=-\log\max_{(\vec{p}, \vec q)\in  \text{Supp}(\Xi_{\rho}): |\Xi_{\rho}(\vec{p}, \vec q)|\neq 1}|\Xi_{\rho}(\vec{p}, \vec q)| \;.
\end{eqnarray*}
If $\set{(\vec{p},\vec q)\in  \text{Supp}(\Xi_{\rho}): |\Xi_{\rho}(\vec{p}, \vec q)|\neq 1}=\emptyset$,  define $\max_{(\vec{p},\vec q)\in  \text{Supp}(\Xi_{\rho}): |\Xi_{\rho}(\vec{p}, \vec q)|\neq 1}|\Xi_{\rho}(\vec{p}, \vec q)|=1$, and then 
$LMG(\rho)=0$.
\end{Def}

\begin{prop}[\bf Properties of the LMG]\label{prop:lmagap}
The  $LMG$ satisfies: 

(1)  $LMG(\rho)=0$ iff $\rho$ is an MSPS, i.e., $\rho=\CMM(\rho)$.
In addition, if $\rho$ is not an MSPS, we have  $0<LMG(\rho)\leq \frac{1}{2}\log\left[\frac{R_P(\rho)-d^k}{d^n\Tr{\rho^2}-d^k}\right]$;

(2) $LMG$ is invariant under the Clifford unitaries;

(3) $LMG(\rho_1\ot\rho_2)= \min\set{LMG(\rho_1),LMG(\rho_2)}$ if $\rho_1$ and $\rho_2$ are not MSPS.

\end{prop}

\begin{proof}
The proof is similar to that of Proposition \ref{prop:magap}.
\end{proof}

We now consider the application of the magic gap in the unitary synthesis. 
 In an $n$-qubit system,  universal quantum 
circuits usually   consist of Clifford unitary gates and $T$ gates, where
\begin{equation}
    T=\left[\begin{array}{cc}
       1  & 0 \\
       0  & e^{i\pi/4}
    \end{array}\right].
\end{equation}
Since the Clifford unitaries can be efficiently simulated  on 
a classical computer, see the Gottesman-Knill theorem~\cite{gottesman1998heisenberg}, the $T$ gates are a resource for quantum computational advantage. 
Hence, it is important to determine how many $T$ gates are necessary to generate the target unitary. 
We find that the LMG  provides a lower bound on the number of $T$ gates.

\begin{prop}\label{prop:gap_syn}
Given an $n$-qubit  input state $\rho$, after applying quantum circuit $V_N$ which consists of Clifford unitaries and $N$ magic $T$ gates, 
the LMG of output state $V_N\rho V^\dag_N$ and of input state $\rho$ have the following relation,
\begin{eqnarray*}
LMG(V_N\rho V^\dag_N)\leq LMG(\rho)+\frac{N}{2} \;.
\end{eqnarray*}
\end{prop}

\begin{proof}
Without loss of generality, we may assume the single-qubit $T$ gate acts on the first qubit, i.e., $T_1$.
We only need to show
\begin{eqnarray*}
\max_{(\vec{p},\vec q)\in  \text{Supp}(\Xi_{\rho}): |\Xi_{T_1\rho T^\dag_1}(\vec{p}, \vec q)|\neq 1}|\Xi_{T_1\rho T^\dag_1}(\vec{p},\vec q)|
\geq \frac{1}{\sqrt{2}}\max_{(\vec{p},\vec q)\in  \text{Supp}(\Xi_{\rho}): |\Xi_{\rho}(\vec{p}, \vec q)|\neq 1}|\Xi_{\rho}(\vec{p}, \vec q)| \;.
\end{eqnarray*}
Denote $S=\set{(\vec{p},\vec q):|\Xi_{\rho}(\vec{p}, \vec q)|=1}$ and $S^C=\text{Supp}[\Xi_{\rho}]\setminus S$, then 
\begin{eqnarray*}
\rho=\frac{1}{d^n}\sum_{(\vec{p},\vec q)\in S}\Xi_{\rho}(\vec{p}, \vec q)w(\vec{p}, \vec q)
+\frac{1}{d^n}\sum_{(\vec{p},\vec q)\in S^C }\Xi_{\rho}(\vec{p}, \vec q)w(\vec{p}, \vec q) \;.
\end{eqnarray*}
Consider the following cases:
(1)
If $S^C=\emptyset$, i.e., $\rho$ is an MSPS, 
then \begin{eqnarray*}
\max_{(\vec{p},\vec q)\in  \text{Supp}(\Xi_{\rho}): |\Xi_{T_1\rho T^\dag_1}(\vec{p},\vec q)|\neq 1}|\Xi_{T_1\rho T^\dag_1}(\vec{p},\vec q)|
\geq \frac{1}{\sqrt{2}}\max_{(\vec{p}, \vec q)\in  \text{Supp}(\Xi_{\rho}): |\Xi_{\rho}(\vec{p}, \vec q)|\neq 1}|\Xi_{\rho}(\vec{p}, \vec q)|
\end{eqnarray*}
comes directly from the fact that $TXT^\dag=\frac{1}{\sqrt{2}}X+\frac{1}{\sqrt{2}}Y$ and 
$TYT^\dag=\frac{1}{\sqrt{2}}X-\frac{1}{\sqrt{2}}Y$, where $Y=iXZ$.

(2) 
If $S^C\neq\emptyset$,
let
\begin{eqnarray*}
S^C_I&=&\set{(\vec{p}, \vec q)\in V^{n-1}:(0, \vec p, 0, \vec q)\in S^C } \;,\\
S^C_X&=&\set{(\vec{p}, \vec q)\in V^{n-1}:(0, \vec p, 1, \vec q)\in S^C} \;,\\
S^C_Y&=&\set{(\vec{p}, \vec q)\in V^{n-1}:(1, \vec p, 1, \vec q)\in S^C } \;,\\
S^C_Z&=&\set{(\vec{p}, \vec q)\in V^{n-1}:(1, \vec p, 0, \vec q)\in S^C} \;.
\end{eqnarray*}

(2.1)
If the second-largest absolute value of the characteristic function $\Xi_{\rho}$ is taken from either $S^C_Z$ or $S^C_I$, i.e., there exists $(\vec{p}_*, \vec q_*)\in S^C_Z$ or $S^C_I$,
\begin{eqnarray*}
\max_{(\vec{p}, \vec q)\in  \text{Supp}(\Xi_{\rho}): |\Xi_{\rho}(\vec{p}, \vec q)|\neq 1}|\Xi_{\rho}(\vec{p}, \vec q)|
=|\Xi_{\rho}(0,\vec p_*, 0, \vec q_*)| \;,
\end{eqnarray*}
or 
\begin{eqnarray*}
\max_{(\vec{p}, \vec q)\in  \text{Supp}(\Xi_{\rho}): |\Xi_{\rho}(\vec{p}, \vec q)|\neq 1}|\Xi_{\rho}(\vec{p}, \vec q)|
=|\Xi_{\rho}(1, \vec p_*, 0, \vec q_*)| \;.
\end{eqnarray*}
Then 
\begin{eqnarray*}
\max_{(\vec{p}, \vec q)\in  \text{Supp}(\Xi_{\rho}): |\Xi_{T_1\rho T^\dag_1}(\vec{p}, \vec q)|\neq 1}|\Xi_{T_1\rho T^\dag_1}(\vec{p}, \vec q)|
\geq \max_{(\vec{p}, \vec q)\in  \text{Supp}(\Xi_{\rho}): |\Xi_{\rho}(\vec{p}, \vec q)|\neq 1}|\Xi_{\rho}(\vec{p}, \vec q)| \;.
\end{eqnarray*}

(2.2) If the second-largest absolute value of the characteristic function $\Xi_{\rho}$ is taken from $S^C_X$, 
i.e., there exists $(\vec{p}_*, \vec q_*)\in S^C_X$ such that
\begin{eqnarray*}
\max_{(\vec{p}, \vec q)\in  \text{Supp}(\Xi_{\rho}): |\Xi_{\rho}(\vec{p}, \vec q)|\neq 1}|\Xi_{\rho}(\vec{p}, \vec q)|
=|\Xi_{\rho}(0, \vec p_*, 1, \vec q_*)| \;.
\end{eqnarray*}
Let 
\begin{eqnarray*}
S_X&=&\set{(\vec{p}, \vec q) \in V^{n-1}:(0, \vec p, 1, \vec q)\in S} \;,\\
S_Y&=&\set{(\vec{p}, \vec q) \in V^{n-1}:(1, \vec p, 1, \vec q)\in S} \;.
\end{eqnarray*}
Then 
\begin{eqnarray*}
S_X\cap S_Y=\emptyset \;,\quad
S_X\cap S^C_X=\emptyset \;,\quad
S_Y\cap S^C_Y=\emptyset \;.
\end{eqnarray*}
By Lemma \ref{0109lem2}, we have 
\begin{eqnarray*}
S_X\cap S^C_Y=\emptyset \;,\quad
S_Y\cap S^C_X=\emptyset \;.
\end{eqnarray*}
Hence $(\vec{p}_*, \vec q_*)\notin S_X$ or $S_Y$.
Then 
\begin{eqnarray*}
\rho&=&\frac{1}{d^n}
\sum_{(\vec{p}, \vec q)\in S_X}
\Xi_{\rho}(0, \vec p, 1, \vec q)X\ot w(\vec{p}, \vec q)
+\frac{1}{d^n}
\sum_{(\vec{p}, \vec q)\in S_Y}
\Xi_{\rho}(1,\vec p, 1, \vec q)Y\ot w(\vec{p}, \vec q)\\
&+&\frac{1}{d^n}
\sum_{(\vec{p}, \vec q)\in S^C_X}
\Xi_{\rho}(0,\vec p, 1, \vec q)X\ot w(\vec{p}, \vec q)
+\frac{1}{d^n}
\sum_{(\vec{p}, \vec q)\in S^C_Y}
\Xi_{\rho}(1, \vec p, 1, \vec q)Y\ot w(\vec{p}, \vec q)\\
&+&\Delta_{Z_1}(\rho) \;,
\end{eqnarray*}
and 
\begin{eqnarray*}
T_1\rho T^\dag_1&=&\frac{1}{d^n}
\sum_{(\vec{p}, \vec q)\in S_X}
\Xi_{\rho}(0, \vec p, 1, \vec q)\frac{1}{\sqrt{2}}(X+Y)\ot w(\vec{p}, \vec q)\\
&+&\frac{1}{d^n}
\sum_{(\vec{p}, \vec q)\in S_Y}
\Xi_{\rho}(1,\vec p, 1, \vec q)\frac{1}{\sqrt{2}}(X-Y)\ot w(\vec{p}, \vec q)\\
&+&\frac{1}{d^n}
\sum_{(\vec{p}, \vec q)\in S^C_X}
\Xi_{\rho}(0, \vec p, 1, \vec q)\frac{1}{\sqrt{2}}(X+Y)\ot w(\vec{p}, \vec q)\\
&+&\frac{1}{d^n}
\sum_{(\vec{p}, \vec q)\in S^C_Y}
\Xi_{\rho}(1, \vec p, 1, \vec q) \frac{1}{\sqrt{2}}(X-Y)\ot w(\vec{p}, \vec q)\\
&+&\Delta_{Z_1}(\rho) \;.
\end{eqnarray*}
Here, $\Delta_{Z_1}$ denotes the completely dephasing channel 
$\Delta_{Z_1}(\rho)=\sum_j\bra{j}\rho\ket{j}\proj{j}$, with $\set{\ket{j}}$ being the eigenbasis for 
$Z_1$. (This definition of a completely dephasing channel is generalized to arbitrary Hermitian operators in \eqref{Def:CDC}.)

Hence  there are only two cases:
(a) $(\vec{p}_*, \vec q_*)\in S^C_Y$;
(b) $(\vec{p}_*, \vec q_*)\notin S^C_Y$.

(2.2.1)
If $(\vec{p}_*, \vec q_*)\in S^C_Y$, then 
\begin{eqnarray*}
\Xi_{T_1\rho T^\dag_1}(0,\vec p_*, 1, \vec q_*)=
\frac{1}{\sqrt{2}}[\Xi_{\rho}(1, \vec p_*,1, \vec q_*)+\Xi_{\rho}(0, \vec p_*, 1, \vec q_*)] \;,\\
\Xi_{T_1\rho T^\dag_1}(1, \vec p_*, 1, \vec q_*)=
\frac{1}{\sqrt{2}}[\Xi_{\rho}(1, \vec p_*,1, \vec q_*)-\Xi_{\rho}(0, \vec p_*,1, \vec q_*)] \;.
\end{eqnarray*}
Hence, we have 
\begin{align*}
&\max\set{|\Xi_{T_1\rho T^\dag_1}(0, \vec p_*, 1, \vec q_*)|,|\Xi_{T_1\rho T^\dag_1}(1, \vec p_*, 1, \vec q_*)|}\\
&\qquad\qquad\geq \frac{1}{\sqrt{2}}\sqrt{|\Xi_{\rho}(1, \vec p_*, 1, \vec q_*)|^2+|\Xi_{\rho}(0, \vec p_*, 1, \vec q_*)|^2}\\
&\qquad\qquad\geq \frac{1}{\sqrt{2}}|\Xi_{\rho}(0, \vec p_*, 1, \vec q_*)| \;.
\end{align*}

(2.2.2)
If $(\vec{p}_*, \vec q_*)\notin S^C_Y$, 
then 
\begin{eqnarray*}
\Xi_{T_1\rho T^\dag_1}(0, \vec p_*, 1, \vec q_*)=
\frac{1}{\sqrt{2}}\Xi_{\rho}(0,\vec p_*, 1, \vec q_*) \;,\\
\Xi_{T_1\rho T^\dag_1}(1, \vec p_*, 1, \vec q_*)=
\frac{1}{\sqrt{2}}\Xi_{\rho}(0, \vec p_*, 1, \vec q_*) \;.
\end{eqnarray*}
Hence
\begin{eqnarray*}
|\Xi_{T_1\rho T^\dag_1}(0, \vec p_*, 1, \vec q_*)|=
\frac{1}{\sqrt{2}}|\Xi_{\rho}(0, \vec p_*, 1, \vec q_*)| \;.
\end{eqnarray*}

(2.3) If the second largest absolute value of the characteristic function $\Xi_{\rho}$  is taken from $S^C_Y$, i.e., there exists $(\vec{p}_*, \vec q_*)\in S^C_Y$ such that
\begin{eqnarray*}
\max_{(\vec{p}, \vec q)\in  \text{Supp}(\Xi_{\rho}): |\Xi_{\rho}(\vec{p}, \vec q)|\neq 1}|\Xi_{\rho}(\vec{p}, \vec q)|
=|\Xi_{\rho}(1, \vec p_*,1, \vec q_*)| \;.
\end{eqnarray*}
We can prove this case similarly as the above case (2.2).
\end{proof}

\section{A framework for convolution of quantum states}

We introduce the convolution of two $n$-qudit systems,  denoted by  
$\mathcal{H}_A=\mathcal{H}^{\ot n}$ and $\mathcal{H}_B=\mathcal{H}^{\ot n}$  where
$\dim \mathcal{H}=d$. 

\subsection{Definition of the discrete convolution}
Given a prime number $d$,
consider the $2\times 2 $ invertible  matrix of parameters, 
\begin{equation}\label{0125shi1}
G=\left[
\begin{array}{cc}
g_{00}&g_{01}\\
g_{10}&g_{11}
\end{array}
\right] \;,
\end{equation}
with entries in $\mathbb{Z}_d$, 
satisfying $\det \,G=g_{00}g_{11}-g_{01}g_{10} \not\equiv 0\mod d$. The inverse of $G$ mod $d$ is  
\begin{equation*}
G^{-1}=N\left[
\begin{array}{cc}
\phantom{-}g_{11}&-g_{01}\\
-g_{10}&\phantom{-}g_{00}
\end{array}
\right] \;, \quad
\text{where } N=(\det \,G)^{-1}\;.
\end{equation*}

 \begin{Def}[\bf Classes of G]\label{NontrivialMatrix }
 \begin{enumerate}[i.]
\item{} The invertible matrix $G$ in \eqref{0125shi1} is called \emph{nontrivial} if at most one of $ g_{11}, g_{10}, g_{01}, g_{00}$ is 0 mod $d$. 

\item{} 
A nontrivial matrix $G$ is called \emph{odd-parity positive} if neither  $g_{01}$ nor $g_{10} $ is  $0$ mod  $d$. 

\item{}
A nontrivial matrix $G$ is called \emph{even-parity positive} if  neither  $g_{00}$ nor $g_{11} $ is 0 mod $  d$.

\item{}
A nontrivial matrix $G$ is called \emph{positive} if it is both odd-parity and even-parity positive.
\end{enumerate}
 \end{Def}
 
 \begin{lem}
For any odd prime number $d$, there always exists a positive and invertible matrix $G$ in $\mathbb{Z}_d$. However,
for $d=2$, there is no positive and  invertible matrix $G$. 
\end{lem}
\begin{proof}
For any odd prime number $d$, let  
\begin{equation*}
G=\left[
\begin{array}{cc}
1&1\\
1&d-1
\end{array}
\right]\;.
\end{equation*}
It is clearly positive, and it is invertible as $\det\,G=d-2\not\equiv 0\mod d$. 
For $d=2$, the only positive matrix in $\mathbb{Z}_2$
is 
\begin{equation*}
G=\left[
\begin{array}{cc}
1&1\\
1&1
\end{array}
\right] \;,
\end{equation*}
which is not invertible as the $\det\,G\equiv 0\mod 2$.
\end{proof}

\begin{Def}[\bf Key Unitary]\label{def:Cli_unitary}
Given a nontrivial and invertible matrix $G$.
The key unitary $U_G$ on $\CHH_A\ot \CHH_B$ is
\begin{eqnarray}\label{eq:Cli_unitary}
U_G=\sum_{\vec{i},\vec j}\ket{\vec{i}'}\bra{\vec i}\ot \ket{\vec{j}'}\bra{\vec j} \;,
\end{eqnarray}
 where the vectors $| \vec i \rangle = |  i_1 \rangle \otimes \cdots \otimes |  i_n \rangle \in \CHH_A$ and $| \vec j \rangle = |  j_1 \rangle \otimes \cdots \otimes |  j_n \rangle \in \CHH_B $,
 and $\left[\begin{array}{c}
 i'_k\\
 j'_k
 \end{array}\right]=(G^{-1})^T
 \left[\begin{array}{c}
 i_k\\
 j_k
 \end{array}
 \right]
 $ for any $k\in [n]$.
That is, $U_G$ maps the vector $\ket{\vec i} \ot\ket{ \vec j} $ to the vector $ \ket{Ng_{11}\vec{i}-Ng_{10}\vec{j}}\ot\ket{ -Ng_{01}\vec{i}+Ng_{00}\vec{j}}$,  
where $N=(det\,G)^{-1}=(g_{00}g_{11}-g_{01}g_{10})^{-1}$.
\end{Def}

By checking the conjugate action of the key unitary $U_G$ on Weyl operators, we show that 
$U_G$ is a Clifford unitary.

\begin{prop}[\bf The key unitary is Clifford]\label{prop:U_wely}
The key unitary $U_G$ satisfies:  
\begin{align*}
&U_G w(\vec{p}_1,\vec{q}_1)\ot w(\vec{p}_2,\vec{q}_2) U^\dag_G\\
&=w(g_{00}\vec{p}_1+g_{01}\vec{p}_2, Ng_{11}\vec{q}_1-Ng_{10}\vec{q}_2)\ot w(g_{10}\vec{p}_1+g_{11}\vec{p}_2,-Ng_{01}\vec{q}_1+Ng_{00}\vec{q}_2) \;,
\end{align*}
for $(\vec{p}_1, \vec{q}_1), (\vec{p}_2, \vec{q}_2)\in V^n$.
\end{prop}

\begin{proof}
It is sufficient to verify the statement for single-qudit Weyl operators. We verify the case  $n=1$ in the following lemma.  The general case is similar.
\end{proof}

\begin{lem} The following identities hold:
\begin{align}
U_G(X \otimes I) U^\dag_G
&= X^{Ng_{11}}\otimes X^{-Ng_{01}}\;,
&U_G(I\otimes X) U^\dag_G
&= X^{-Ng_{10}}\otimes X^{Ng_{00}}\;,\\
U_G (Z\otimes I) U^\dag_G
&= Z^{g_{00}}\otimes Z^{g_{10}}\;,
&
U_G(I\otimes Z) U^\dag_G
&=    Z^{g_{01}}\otimes Z^{g_{11}}\;,
\end{align}
where the key unitary $U_G$ is  given in \eqref{eq:Cli_unitary}.
\end{lem}

\begin{proof}
We carry out the computation.  
\begin{align*}
U_G(X \otimes I) U^\dag_G
&= \sum_{i,j} U_G\big(|i+1 \rangle \langle i| \otimes |j \rangle \langle j|\big)U^\dag_G \\
&=   \sum_{i,j}     |Ng_{11}(i+1)-Ng_{10}j \rangle 
\langle Ng_{11}i-Ng_{10}j| \\
&\qquad \qquad\otimes |-Ng_{01}(i+1)+Ng_{00}j \rangle \langle -Ng_{01}i+Ng_{00}j|  \\
&=  \sum_{i',j'} |i'+Ng_{11} \rangle \langle i'| \otimes |j'-Ng_{01} \rangle \langle j'|\\
&=  X^{Ng_{11}}\otimes X^{-Ng_{01}}\;,
\end{align*}
where in the third equality we used the fact that,
for every pair $(i',j')$ the following equation
\begin{equation}
\left[
\begin{array}{c}
i'\\
j'
\end{array}
\right]
=(G^{-1})^T
\left[
\begin{array}{c}
i\\
j
\end{array}
\right]
\end{equation}
always has a solution $(i,j)$ as $G$ is invertible.  Likewise, 
\begin{align*}
U_G(I\otimes X) U^\dag_G
&= \sum_{i,j} U_G \big(|i  \rangle \langle i| \otimes |j+1 \rangle \langle j|\big) U^\dag_G\\
&=  \sum_{i,j}      |Ng_{11}i-Ng_{10}(j+1) \rangle \langle Ng_{11}i-Ng_{10}j  | \\
&\qquad \qquad\otimes |-Ng_{01}i+Ng_{00}(j+1) \rangle \langle -Ng_{01}i+Ng_{00}j|  \\
&= \sum_{i',j'} |i'-Ng_{10} \rangle \langle i'| \otimes |j'+Ng_{00} \rangle \langle j'|\\
&= X^{-Ng_{10}}\otimes X^{Ng_{00}}\;.
\end{align*}
\begin{align*}
U_G(Z\otimes I) U^\dag_G
&= \sum_{i,j}\chi(i)U_G \big(|i  \rangle \langle i| \otimes |j \rangle \langle j|\big)U^\dag_G\\
&=  \sum_{i,j}   \chi(i)   |Ng_{11}i-Ng_{10}j \rangle \langle Ng_{11}i-Ng_{10}j| \\
&\qquad \qquad\otimes |-Ng_{01}i+Ng_{00}j \rangle \langle -Ng_{01}i+Ng_{00}j|  \\
&= \sum_{i',j'} \chi( g_{00}i'+g_{10}j')|i' \rangle \langle i'| \otimes |j' \rangle \langle j'|\\
&= Z^{g_{00}}\otimes Z^{g_{10}}\;.
\end{align*}
\begin{align*}
U_G(I\otimes Z) U^\dag_G
&= \sum_{i,j}\chi(j) U_G\big(|i  \rangle \langle i| \otimes |j \rangle \langle j|\big) U^\dag_G\\
&=  \sum_{i,j}   \chi(j)   |Ng_{11}i-Ng_{10}j \rangle \langle Ng_{11}i-Ng_{10}j| \\
&\qquad \qquad\otimes |-Ng_{01}i+Ng_{00}j \rangle \langle -Ng_{01}i+Ng_{00}j| \\
&= \sum_{i',j'} \chi( g_{01}i'+g_{11}j')|i' \rangle \langle i'| \otimes |j' \rangle \langle j'|\\
&= Z^{g_{01}}\otimes Z^{g_{11}} \;.
\end{align*}
This completes the proof.
\end{proof}

\begin{Def}[\bf Convolution of states]\label{def:convo}
Given two quantum states $\rho\in \mathcal{D}(\mathcal{H}_A)$, $\sigma\in \mathcal{D}(\mathcal{H}_B)$,
the convolution  of $\rho$ and $\sigma$ is  
\begin{align}\label{eq:convo}
\rho \boxtimes \sigma = \Ptr{B}{ U_G (\rho \otimes \sigma) U^\dag_G}\;,
\end{align}
where the key unitary $U_G$ is defined in \eqref{eq:Cli_unitary}, and the partial trace is taken on the second $n$-qudit system $\mathcal{H}_B$.
\end{Def}

\begin{Def}[\bf Convolutional channel]\label{def:con_Channel}
The quantum convolutional channel $\mathcal{E}_G$ based on the key unitary $U_G$ defined in \eqref{eq:Cli_unitary} is
\begin{eqnarray}\label{eq:convo_chn}
\mathcal{E}_G(\ \cdot\ )=\Ptr{B}{ U_G(\ \cdot\ ) U^\dag_G} \;.
\end{eqnarray}
Hence, $\rho \boxtimes \sigma=\mathcal{E}_G(\rho\ot\sigma)$ is the output state of the 
convolutional channel $\mathcal{E}_G$. 

\end{Def}

\begin{lem}\label{lem:adj_con}
For the quantum convolutional channel $\mathcal{E}_G$, we have 
\begin{eqnarray*}
\mathcal{E}^\dag_G(w(\vec p,\vec q))
=w(Ng_{11}\vec p, g_{00}\vec q)\ot w(-Ng_{10}\vec p, g_{01}\vec q) \;,\quad\forall (\vec p, \vec q ) \in V^n \;.
\end{eqnarray*}
\end{lem}
\begin{proof}
This is because 
\begin{eqnarray*}
\mathcal{E}^\dag_G(w(\vec p,\vec q))
=U^\dag_G (w(\vec p,\vec q)\ot I  )U_G
=w(Ng_{11}\vec p, g_{00}\vec q)\ot w(-Ng_{10}\vec p, g_{01}\vec q) \;,
\end{eqnarray*}
where the second equality comes from 
Proposition \ref{prop:U_wely}.
\end{proof}

Based on the above  lemma, we can get the following results on the conditions for  abelian convolution.
\begin{cor}
If $g_{11}=-g_{10}$, and $g_{00}=g_{01}$, then the corresponding convolution is 
 abelian, i.e., $\rho\boxtimes \sigma=\sigma\boxtimes\rho$.
\end{cor}

\begin{prop}\label{prop:comm_wel}
For a $2n$-qudit state $\rho_{AB}$,
the quantum convolutional channel $\mathcal{E}_G$ commutes with the conjugation of Weyl operators:
 \begin{eqnarray*}
  \mathcal{E}_G[
( w(\vec{x}_1, \vec{y}_1)\ot w(\vec{x}_2, \vec{y}_2))
\,\rho_{AB}\, 
( w(\vec{x}_1, \vec{y}_1)^\dag\ot w(\vec{x}_2, \vec{y}_2)^\dag)
 ]
 =w(\vec{x},\vec{y})
 \mathcal{E}_G(\rho_{AB})
w(\vec{x}, \vec{y})^\dag \;,
 \end{eqnarray*}
 where  $\vec{x}=g_{00}\vec{x}_1+g_{01}\vec{x}_2$ and 
 $\vec{y}=Ng_{11}\vec{y}_1-Ng_{10}\vec{y}_2$.
\end{prop}
\begin{proof}
For every Weyl operator
$w(\vec p, \vec q)$, we have 
\begin{align*}
 &w(\vec{x}_1, \vec{y}_1)^\dag\ot w(\vec{x}_2, \vec{y}_2)^\dag \mathcal{E}^\dag_G (w(\vec p,\vec q)) \,w(\vec{x}_1, \vec{y}_1)\ot w(\vec{x}_2, \vec{y}_2)\\
&\quad = w(\vec{x}_1, \vec{y}_1)^\dag\ot w(\vec{x}_2, \vec{y}_2)^\dag
 (w(Ng_{11}\vec p, g_{00}\vec q)\ot w(-Ng_{10}\vec p, g_{01}\vec q))
   w(\vec{x}_1, \vec{y}_1)\ot w(\vec{x}_2, \vec{y}_2)\\
&\quad = \chi(-g_{00}\vec{q}^T\vec{x}_1+Ng_{11}\vec{p}^T\vec{y}_1-g_{01}\vec{q}^T\vec{x}_2-Ng_{10}\vec{p}^T\vec{y}_2)\\
&\qquad \qquad \times w(Ng_{11}\vec p, g_{00}\vec q) \ot w(-Ng_{10}\vec p, g_{01}\vec q) \\
& \quad =  \chi(-\vec{q}^T(g_{00}\vec{x}_1+g_{01}\vec{x}_2)-\vec{p}^T(-Ng_{11}\vec{y}_1+Ng_{10}\vec{y}_2))w(Ng_{11}\vec p, g_{00}\vec q)\\ 
& \qquad \qquad \qquad \ot w(-Ng_{10}\vec p, g_{01}\vec q) \\
&\quad =\mathcal{E}^\dag_G  (w(\vec{x}, \vec{y})^\dag w(\vec p, \vec q) \, w(\vec{x}, \vec{y})) \;,
\end{align*}
where $\vec{x}=g_{00}\vec{x}_1+g_{01}\vec{x}_2$ and 
 $\vec{y}=Ng_{11}\vec{y}_1-Ng_{10}\vec{y}_2$.
\end{proof}

\begin{prop}[\bf Convolution-multiplication duality]\label{prop:chara_conv}
Given two $n$-qudit states $\rho $ and $\sigma$, the characteristic function of $\rho\boxtimes\sigma$ satisfies
\begin{align*}
\Xi_{ \rho \boxtimes \sigma} (\vec p, \vec q) = \Xi_\rho (Ng_{11}\vec p, g_{00}\vec q) \;\Xi_\sigma (-Ng_{10}\vec p, g_{01}\vec q) \;, \quad \forall \vec p, \vec q \in \mathbb{Z}^n_d \;.
\end{align*}
\end{prop}
\begin{proof}
\begin{align*}
\Xi_{ \rho \boxtimes \sigma} (\vec p,\vec q) 
=&  \Tr{  \rho \boxtimes\sigma  w(-\vec p,-\vec{q}) }
= \Tr{ \Ptr{B}{U_G(\rho \otimes \sigma) U^\dag_G }w(-\vec p, -\vec q) } \\
=& \Tr{ U_G (\rho \otimes \sigma) U^\dag_G(w(-\vec p,-\vec{q}) \otimes I)}
= \Tr{  (\rho \otimes \sigma) U^\dag_G (w(-\vec p,-\vec q)\otimes I) U_G}\\
=& \Tr{  (\rho \otimes \sigma)   (w(-Ng_{11}\vec p, -g_{00}\vec q)\ot w(Ng_{10}\vec p, -g_{01}\vec q)) }\\
=& \Tr{  \rho w(-Ng_{11}\vec p, -g_{00}\vec q)  }\Tr{    \sigma w(Ng_{10}\vec p, -g_{01}\vec q)}\\
=& \Xi_\rho (Ng_{11}\vec p, g_{00}\vec q)\; \Xi_\sigma (-Ng_{10}\vec p, g_{01}\vec q) \;.
\end{align*} 
\end{proof}

Based on the relation between characteristic functions, we have the following result directly.
\begin{lem}\label{lem:conv_iden}
Let  $\rho$ be an $n$-qudit state.

(1) If $G$ is odd-parity positive, then 
\begin{align*}
 \rho \boxtimes \frac{I_n}{d^n} = \frac{I_n}{d^n} \;,
\end{align*}

(2) If $G$ is even-parity positive, then
\begin{align*}
    \frac{I_n}{d^n}\boxtimes \rho =\frac{I_n}{d^n} \;.
\end{align*}
\end{lem}
\begin{proof}
By Proposition \ref{prop:chara_conv},  the characteristic function of $\rho \boxtimes I_n/d^n$ is
\begin{align*}
\Xi_{ \rho \boxtimes I_n/d^n} (\vec{p}, \vec q) = \Xi_\rho (Ng_{11}\vec p, g_{00}\vec q) \,\Xi_{I_n/d^n}(-Ng_{10}\vec p, g_{01}\vec q)
= \Xi_\rho (Ng_{11}\vec p, g_{00}\vec q)
\delta_{\vec{p},\vec{0}}\delta_{\vec{q},\vec{0}} \;,
\end{align*}
where the second equality comes from the fact that  $G$ is odd-parity positive, i.e., both $g_{01}$ and $g_{10}$ are not $  0\mod d$.
Thus $\rho \boxtimes \frac{I_n}{d^n}=\frac{I_n}{d^n}$.
A similar argument can be made for the case when $G$ is even-parity positive.
\end{proof}

In classical probability theory, it is well-known that the convolution of 
two Gaussian distributions is still Gaussian. Here, we find the analogous property of MSPS.
\begin{lem}
The convolution $\rho\boxtimes\sigma$ of two MSPSs $\rho$ and $\sigma$  is an MSPS.
\end{lem}
\begin{proof}
By Proposition \ref{prop:chara_conv},  the characteristic function of $\rho \boxtimes \sigma$ is
\begin{eqnarray}
    \Xi_{ \rho \boxtimes \sigma} (\vec{p}, \vec q) = \Xi_\rho (Ng_{11}\vec p, g_{00}\vec q) \Xi_{\sigma}(-Ng_{10}\vec p, g_{01}\vec q).
\end{eqnarray}
Since $\rho,\sigma$ are MSPSs, 
$|\Xi_{\rho}|$ and $ |\Xi_{\sigma}|$ are either equal to 1 or 0. Thus, 
$|\Xi_{\rho\boxtimes\sigma}|$ is either  equal to 1 or 0. Therefore, $\rho\boxtimes \sigma$ is an MSPS.
\end{proof}
In fact, the above result also implies  convolutional stability.
\begin{prop}[\bf Convolutional stability for states]
Given two $n$-qudit stabilizer states $\rho$ and $\sigma$, their convolution $\rho\boxtimes\sigma$ is a stabilizer state.
\end{prop}
\begin{proof}
Since $\rho,\sigma $ are stabilizer states, they can be written as a convex combination of pure stabilizer states, i.e.,
$\rho=\sum_i\lambda_i\proj{\psi_i},\sigma=\sum_j\mu_j\proj{\phi_j}$. 
Based on the above Lemma, the convolution of two pure stabilizer states $\proj{\psi_i}\boxtimes\proj{\phi_j}$ is an MSPS, hence a convex combination of pure stabilizer states, then $\rho\boxtimes \sigma=\sum_{i,j}\lambda_i\mu_j\proj{\psi_i}\boxtimes\proj{\phi_j}$ 
is a stabilizer state.
\end{proof}

\begin{prop}[\bf Convolution of discrete Wigner functions]\label{prop:wig_conv}
Assume that the parameter matrix $G$ is  positive and invertible. Then the discrete Wigner function of Definition \ref{Def:DiscreteWignerFunction},
for $n$-qudit states $\rho$ and $\sigma$, satisfies
\begin{eqnarray*}
W_{\rho\boxtimes\sigma}(\vec u, \vec v)
=\sum_{\vec{u}_1,\vec{v}_1}
W_{\rho}(g^{-1}_{00}\vec{u}_1,(Ng_{11})^{-1}\vec{v}_1)\;\;
W_{\sigma}(g^{-1}_{01}(\vec{u}-\vec{u}_1), (Ng_{10})^{-1}(\vec{v}_1-\vec{v})) \;.
\end{eqnarray*}
\end{prop}
The formula in Proposition \ref{prop:wig_conv} simplifies in some 
cases, such as for the discrete beam splitter in Proposition~\ref{prop:wign_bs}, and for the discrete amplifier in Proposition~\ref{prop:wign_am}.
\begin{proof}
Since 
$\Xi_{ \rho \boxtimes \sigma} (\vec p, \vec q) = \Xi_\rho (Ng_{11}\vec p, g_{00}\vec q)\; \Xi_\sigma (-Ng_{10}\vec p, g_{01}\vec q)$, we have 
\begin{eqnarray*}
&&W_{\rho\boxtimes\sigma}(\vec u, \vec v)\\
&=&\frac{1}{d^{2n}}
\sum_{\vec{p},\vec{q}} \Xi_{ \rho \boxtimes \sigma} (\vec p, \vec q) 
\chi(\vec{p}^T\vec{v}-\vec{q}^T\vec{u})\\
&=&\frac{1}{d^{2n}}
\sum_{\vec{p},\vec{q}} 
\Xi_\rho (Ng_{11}\vec p, g_{00}\vec q) \Xi_\sigma (-Ng_{10}\vec p, g_{01}\vec q)
\chi(\vec{p}^T\vec{v}-\vec{q}^T\vec{u})\\
&=&\frac{1}{d^{2n}}
\sum_{\vec{p},\vec{q}} 
\left( \sum_{\vec{u}_1,\vec{v}_1} W_{\rho}(\vec{u}_1,\vec{v}_1)\chi(-Ng_{11}\vec{p}^T\vec{v}_1+g_{00}\vec{q}^T\vec{u}_1)\right)\\
&&\quad\times
\left( \sum_{\vec{u}_2,\vec{v}_2} W_{\sigma}(\vec{u}_2,\vec{v}_2)\chi(Ng_{10}\vec{p}^T\vec{v}_2+g_{01}\vec{q}^T\vec{u}_2)\right)
\chi(\vec{p}^T\vec{v}-\vec{q}^T\vec{u})\\
&=&\sum_{\vec{u}_1,\vec{v}_1}
\sum_{\vec{u}_2,\vec{v}_2}W_{\rho}(\vec{u}_1,\vec{v}_1)
W_{\sigma}(\vec{u}_2,\vec{v}_2)\\
&&\quad\times
\frac{1}{d^{2n}}\sum_{\vec{p},\vec{q}} \chi(-Ng_{11}\vec{p}^T\vec{v}_1+g_{00}\vec{q}^T\vec{u}_1)\chi(Ng_{10}\vec{p}^T\vec{v}_2+g_{01}\vec{q}^T\vec{u}_2)\chi(\vec{p}^T\vec{v}-\vec{q}^T\vec{u})\\
&=&\sum_{\vec{u}_1,\vec{v}_1}
\sum_{\vec{u}_2,\vec{v}_2}W_{\rho}(\vec{u}_1,\vec{v}_1)
W_{\sigma}(\vec{u}_2,\vec{v}_2)
\delta_{\vec{v}-Ng_{11}\vec{v}_1+Ng_{10}\vec{v}_2,\vec{0}}
\delta_{\vec{u}-g_{00}\vec{u}_1-g_{01}\vec{u}_2, \vec 0}\\
&=&\sum_{\vec{u}_1,\vec{v}_1}
W_{\rho}(g^{-1}_{00}\vec{u}_1,(Ng_{11})^{-1}\vec{v}_1)
W_{\sigma}(g^{-1}_{01}(\vec{u}-\vec{u}_1), (Ng_{10})^{-1}(\vec{v}_1-\vec{v})) \;.
\end{eqnarray*}
\end{proof}

\subsection{Monotonicity under convolution}
In classical probability theory, it is well-known that the distance measures are
monotone under convolution $*$, such as the $L_1$ norm (also known as total variation distance), the relative entropy, and the Wasserstein distance.

To compare with the classical case, we consider the quantum version of the monotonicity of distance measures under quantum convolution, including the $L_1$-norm, the relative entropy, and the quantum Wasserstein distance. Here, we  consider two properties of a distance measure $D:\mathcal{D}(\mathcal{H}^{\ot n})\times \mathcal{D}(\mathcal{H}^{\ot n})\to \real$: (i) monotonicity under quantum channels, i.e., $D(\Lambda(\rho),\Lambda(\sigma))\leq D(\rho,\sigma)$; 
(ii) subadditivity under tensor product, i.e., $D(\rho_1\ot\sigma_1,\rho_2\ot\sigma_2)\leq D(\rho_1, \rho_2)+D(\sigma_1,\sigma_2)$.

\begin{prop}[\bf Monotonicity under Quantum Convolution]
Let $D:\mathcal{D}(\mathcal{H}^{\ot n})\times \mathcal{D}(\mathcal{H}^{\ot n})\to \real$ be a distance measure, satisfying
(i) monotonicity under quantum channels, and 
(ii) subadditivity under tensor product.
Then
\begin{eqnarray}\label{ineq:suba_1}
D(\rho_1\boxtimes\sigma_1,\rho_2\boxtimes\sigma_2)
        \leq D(\rho_1,\rho_2)+D(\sigma_1,\sigma_2) \;.
\end{eqnarray}
Moreover, if $\tau$ is invariant under $\boxtimes$, i.e., $\tau\boxtimes\tau=\tau$, then 
 \begin{eqnarray}\label{ineq:suba_2}
        D(\rho\boxtimes\sigma,\tau)\leq D(\rho,\tau)
        +D(\sigma,\tau) \;.
    \end{eqnarray}
\end{prop}
\begin{proof}
Based on the monotonicity of $D$ under quantum channels and 
 the subadditivity of $D$ under tensor product, we have
\begin{align*}
D(\rho_1\boxtimes\sigma_1,\rho_2\boxtimes\sigma_2)
=&D(\mathcal{E}_G(\rho_1\ot\sigma_1),\mathcal{E}_G(\rho_2\ot\sigma_2))\\
\leq& D(\rho_1\ot\sigma_1,\rho_2\ot\sigma_2)\\
\leq&  D(\rho_1,\rho_2)+D(\sigma_1,\sigma_2) \;,
\end{align*}
If $\tau\boxtimes\tau=\tau$, then 
\begin{eqnarray*}
 D(\rho\boxtimes\sigma,\tau)
 =D(\rho\boxtimes\sigma,\tau\boxtimes\tau)
 \leq D(\rho,\tau)+D(\sigma,\tau) \;,
\end{eqnarray*}
to complete the proof.
\end{proof}

\begin{cor}\label{prop:mon_con_ren}
Both the $L_1$-norm $\norm{\cdot}_1$ and the quantum  R\'enyi relative entropy $D_{\alpha}$ with $\alpha\geq 1/2$
satisfy  \eqref{ineq:suba_1} and \eqref{ineq:suba_2}.
\end{cor}

\begin{Def}[\bf De Palma, Marvian, Trevisan, and Lloyd \cite{Palma21}]
Given two $n$-qudit states $\rho,\sigma\in\mathcal{D}(\mathcal{H}^{\ot n})$, the \textit{quantum Wasserstein distance of order 1} 
is 
\begin{eqnarray*}
\norm{\rho-\sigma}_{W_1}
:= 
\min
\left\{\sum_{i=1}^n |c_i|: \rho-\sigma=\sum_{i=1}^n c_i(\rho_i-\sigma_i)\;,  \rho_i,\sigma_i\in \mathcal{D}(\mathcal{H}^{\ot n}),
\text{ and }\ \right.\\
\hskip -1in
\left.\phantom{\sum_{i}^{n}} 
\Ptr{i}{\rho_i}=\Ptr{i}{\sigma_i} ,\  i=1,...,n\;. \right\}
\label{eq:def_W1}
\end{eqnarray*}
\end{Def}
The quantum Wasserstein distance, as defined above, has numerous applications in quantum information and computation, such as quantum machine learning~\cite{Kiani21,de2022limitations} and quantum circuit complexity~\cite{li2022wasserstein}. Here, we consider how the quantum Wasserstein distance acts under our quantum convolution.

\begin{prop}[\bf Wasserstein monotonicity]
Let  $\rho,\tau, \sigma$ be three $n$-qudit states.

(1) If $G$ is odd-parity positive, then 
\begin{eqnarray*}
\norm{\sigma\boxtimes\rho-\sigma\boxtimes\tau}_{W_1}
\leq \norm{\rho-\tau}_{W_1} \;.
\end{eqnarray*}

(2) If $G$ is even-parity positive, then 
\begin{eqnarray*}
\norm{\rho\boxtimes\sigma-\tau\boxtimes\sigma}_{W_1}
\leq \norm{\rho-\tau}_{W_1} \;.
\end{eqnarray*}

(3) If  $G$ is positive, then 
\begin{eqnarray*}
\norm{\rho_1\boxtimes\sigma_1-\rho_2\boxtimes\sigma_2}_{W_1}
\leq \norm{\rho_1-\rho_2}_{W_1}+ \norm{\sigma_1-\sigma_2}_{W_1} \;.
\end{eqnarray*}
\end{prop}

\begin{proof}
We begin with the case that $G$ is odd-parity positive. 
We need to prove that
\begin{eqnarray*}
\norm{\mathcal{E}_{\sigma}(\rho)-\mathcal{E}_{\sigma}(\tau)}_{W_1}
\leq \norm{\rho-\tau}_{W_1} \;,
\end{eqnarray*}
for the quantum convolutional channel 
$\mathcal{E}_{\sigma}(\rho)=\Ptr{B}{U_G\rho\ot\sigma U^\dag_G}$. 
Consider the optimal decomposition of  $ \norm{\rho-\tau}_{W_1}$. In other words,
$\norm{\rho-\tau}_{W_1}=\sum_i|c_i|$ with 
\begin{eqnarray}\label{eq:decom}
\rho-\tau=\sum_{i=1}^n c_i(\rho_i-\tau_i), \; \text{ and }\; \Ptr{i}{\rho_i}=\Ptr{i}{\tau_i}, \;   i=1,...,n \;.
\end{eqnarray}
Let $\set{w(p,q):(p,q)\in V_i}$ be the local Weyl operators  on the $i$-th qudit. Then we infer from Proposition \ref{prop:comm_wel} that 
\begin{eqnarray*}
\mathbb{E}_{(p,q)\in V_i}\,\mathcal{E}_{\sigma}(w(p ,q )\,\rho \,w(p ,q )^\dag)
=\mathbb{E}_{(p,q)\in V_i} \,
w(g_{00}p ,Ng_{11}q)\,
\mathcal{E}_{\sigma}(\rho ) \,w(g_{00}p ,Ng_{11}q )^\dag \;.
\end{eqnarray*}
Since 
$
\mathbb{E}_{(p,q)\in V_i}\, w(p ,q )\,(\ \cdot\ )\, w(p ,q )^\dag
= \trace_i[\cdot] \otimes {I_i}/{d}
$,
and $G$ is odd-parity positive, 
we have 
\begin{eqnarray}\label{eq:com_tr}
\Ptr{i}{\mathcal{E}_{\sigma}(\rho)}\ot I_i/d
=\mathcal{E}_{\sigma}(\Ptr{i}{\rho} \ot I_i/d) \;.
\end{eqnarray}
Hence, we have the following decomposition 
\begin{eqnarray*}
\mathcal{E}_{\sigma}(\rho)-\mathcal{E}_{\sigma}(\tau)=\sum_ic_i[\mathcal{E}_{\sigma}(\rho_i)-\mathcal{E}_{\sigma}(\tau_i)] \;,
\end{eqnarray*}
and
\begin{eqnarray*}
\Ptr{i}{\mathcal{E}_{\sigma}(\rho_i)}\ot I_i/d
=\mathcal{E}_{\sigma}(\Ptr{i}{\rho_i}\ot I_i/d)
=\mathcal{E}_{\sigma}(\Ptr{i}{\tau_i}\ot I_i/d)
=\Ptr{i}{\mathcal{E}_{\sigma}(\tau_i)}\ot I_i/d \;,
\end{eqnarray*}
where the first and third equalities come from \eqref{eq:com_tr}, and the second equality comes from \eqref{eq:decom}. Thus 
\begin{eqnarray*}
\Ptr{i}{\mathcal{E}_{\sigma}(\rho_i)}=\Ptr{i}{\mathcal{E}_{\sigma}(\tau_i)} \;,
\end{eqnarray*}
and therefore  $\norm{\sigma\boxtimes\rho-\sigma\boxtimes\tau}_{W_1} \le \sum_i |c_i| =\norm{\rho-\tau}_{W_1}$. This shows that case (1) holds.
A similar argument demonstrates  case (2). Finally case (3) is a corollary of cases (1) and (2). 
\end{proof}

\subsection{Convolution-Majorization}
Here we investigate how the notion of majorization in Definition \ref{def:major} is related to the convolution $\boxtimes$ of states. Recall the vector $\vec{\lambda}_\rho$ of the state $\rho$ introduced in \eqref{eq:eigen_vec}.

\begin{Def}
Assume the two states $\rho$ and $\sigma$ have vectors  $\vec{\lambda}_{\rho}, \vec{\lambda}_{\sigma}$ of eigenvalues.  
\begin{enumerate}[(1)]
    \item
    The convolution $\boxtimes$ has odd convolution-majorization if for all $\rho,\sigma$, 
\begin{eqnarray}\label{CM-1}
\vec{\lambda}_{\rho\boxtimes\sigma}
\prec\vec{\lambda}_{\sigma} \;.
\end{eqnarray}

\item
The convolution $\boxtimes$ has even convolution-majorization if for all $\rho,\sigma$,
\begin{eqnarray}\label{CM-2}
\vec{\lambda}_{\rho\boxtimes\sigma}
\prec\vec{\lambda}_{\rho} \;.
\end{eqnarray}

\item The convolution $\boxtimes$ has full convolution-majorization, if both (1) and (2) hold.   
\end{enumerate}
\end{Def}

\begin{thm}[\bf Convolution-Majorization]\label{thm:major}
(1) If $G$ is odd-parity positive, then the corresponding convolution $\boxtimes$ has  odd convolution-majorization.
For any Schur-concave function $f$, 
\begin{eqnarray}
f(\vec{\lambda}_{\rho \boxtimes \sigma})
\ge f(\vec{\lambda}_{\sigma})\;.
\end{eqnarray}

(2) If $G$ is even-parity positive, then the corresponding convolution $\boxtimes$ has even convolution-majorization.
For any Schur-concave function $f$, 
\begin{eqnarray}\label{ineq:sch_con_2}
f(\vec{\lambda}_{\rho \boxtimes \sigma})
\ge f(\vec{\lambda}_{\rho})\;.
\end{eqnarray}

(3) If $G$ is positive, then the corresponding convolution $\boxtimes$ has full convolution-majorization. For any Schur-concave function $f$, 
\begin{eqnarray}
f(\vec{\lambda}_{\rho\boxtimes \sigma})
\ge \max\set{f(\vec{\lambda}_{\rho}), f(\vec{\lambda}_{\sigma})}\;.
\end{eqnarray}

\end{thm}
\begin{proof}
Let us analyze the case where $G$ is even-parity positive; the proofs in the other cases  are similar.
Since $f$ is Schur-concave,   inequality \eqref{ineq:sch_con_2}
follows from $\vec{\lambda}_{\rho\boxtimes\sigma}
\prec\vec{\lambda}_{\rho}$. 
Consider the spectral decompositions of the states $\rho$ and $\rho \boxtimes \sigma$,
\begin{align*}
\rho = & \sum_{j=1}^{d^n} \lambda_j |\psi_j\rangle \langle \psi_j|,\\
\rho \boxtimes\sigma =& \sum_{j=1}^{d^n} \nu_j |\xi_j\rangle \langle \xi_ j| \;,
\end{align*}
where both $\ket{\psi_j}$ and $\ket{\xi_j}$ are  normalized eigenstates.
We need to prove that 
\begin{eqnarray*}
(\nu_1,...,\nu_{d^n}) \prec(\lambda_1,...,\lambda_{d^n}) \;.
\end{eqnarray*}
In the definition of order, we assume that $\lambda_1\ge \cdots \ge \lambda_{d^n}$ and $\nu_1\ge \cdots \ge \nu_{d^n} $.
Define
\begin{align*}
\tau_j 
=  
|\psi_j\rangle \langle \psi_j|\boxtimes \sigma \;,
\end{align*}
where each $\tau_j$ is a quantum state.
Moreover,
\begin{align}\label{0102shi1}
\sum_{j=1}^{d^n} \lambda_j \tau_j =  \sum_{j=1}^{d^n}  \lambda_j |\psi_j\rangle \langle \psi_j| \boxtimes \sigma = \rho\boxtimes \sigma \;,
\end{align}
and
\begin{align}\label{0102shi3}
\sum_{j=1}^{d^n} \tau_j = \sum_{j=1}^{d^n} \proj{\psi_j}\boxtimes \sigma=\left(\sum^{d^n}_{j=1}\proj{\psi_j}\right)\boxtimes\sigma = I \boxtimes \sigma = I \;,
\end{align}
where the last equality comes from Lemma \ref{lem:conv_iden}.
By (\ref{0102shi1}),
\begin{align}\label{0102shi4}
\nu_k = \langle \xi_k| \rho \boxtimes \sigma | \xi_k \rangle = \sum_{j=1}^{d^n} \lambda_j \langle \xi_k| \tau_j  | \xi_k \rangle \;.
\end{align}
Let us consider the $d^n\times d^n$ matrix $M=\left( m_{kj}\right)_{k,j=1}^{d^n}$, where each entry $m_{kj}$ is defined as
\begin{align*}
m_{kj} = \langle \xi_k| \tau_j  | \xi_k \rangle \;.
\end{align*}
By definition, each $m_{kj}\ge 0$, and
\begin{align}
\sum_{k } m_{kj} =& \sum_{k }  \langle \xi_k| \tau_j  | \xi_k \rangle = \Tr{\tau_j}= 1 \;,\\
\sum_{j } m_{kj} =&\sum_{j } \langle \xi_k| \tau_j  | \xi_k \rangle = \langle \xi_k| I  | \xi_k \rangle =1 \;, \label{0102shi2}
\end{align}
where (\ref{0102shi2}) comes from (\ref{0102shi3}).
Thus $M$ is a doubly stochastic matrix.
By (\ref{0102shi4}),
\begin{align}\label{0109shi3}
(\nu_1,...,\nu_{d^n})^T = M (\lambda_1,...,\lambda_{d^n}) ^T \;.
\end{align}
Based on Proposition 1.A.3 in \cite{MOA79}, $
(\nu_1,...,\nu_{d^n}) \prec(\lambda_1,...,\lambda_{d^n})$.
\end{proof}

\begin{prop}[\bf Convolution increases subentropy of states]
Let $\rho, \sigma$ be two $n$-qudit states,

(1) If $G$ is odd-parity positive, then the subentropy in Definition~\ref{def:suben} satisfies,
\begin{align}
Q(\rho \boxtimes \sigma) \ge  Q(  \sigma) \;.
\end{align} 

(2) If $G$ is even-parity positive, 
then 
\begin{align}
Q(\rho \boxtimes \sigma) \ge Q(\rho ) \;.
\end{align}

(3) If $G$ is positive, then  
\begin{eqnarray*}
Q(\rho\boxtimes \sigma)
\geq \max\set{Q(\rho), Q(\sigma)} \;.
\end{eqnarray*}
\end{prop}

\begin{prop}[\bf Convolution increases R\'enyi entropy of states]\label{prop:entropy}
Let $\rho, \sigma$ be two $n$-qudit states, and $\alpha \in [-\infty, +\infty]$. 

(1) If $G$ is odd-parity positive, then the generalized R\'enyi entropy in Definition~\ref{def:gen_ren_en} satisfies,
\begin{align}\label{ineq:entrop_con}
H_\alpha(\rho \boxtimes \sigma) \ge  H_\alpha(  \sigma) \;.
\end{align} 

(2) If $G$ is even-parity positive, 
then 
\begin{align}\label{ineq:entrop_con_2}
H_\alpha(\rho \boxtimes \sigma) \ge H_\alpha(\rho ) \;.
\end{align}

(3) If $G$ is positive, then  
\begin{eqnarray*}
H_{\alpha}(\rho\boxtimes \sigma)
\geq \max\set{H_{\alpha}(\rho), H_{\alpha}(\sigma)} \;.
\end{eqnarray*}
\end{prop}

We obtain the entropy inequalities for the convolution of quantum states based on the convolution-majorization property. This can be regarded as a quantum analog of the entropy power inequality, which is an important topic for classical convolution.

\begin{Rem}
If $G$ is not odd-parity positive (resp., not even-parity positive), 
then there may exist $\rho, \sigma$ such that 
$H_{\alpha}(\rho\boxtimes\sigma)<H_\alpha(\sigma)$ (resp., $H_{\alpha}(\rho\boxtimes\sigma)<H_\alpha(\rho)$.)
For example,
when $g_{10}=0$ and $g_{00}= g_{01}=g_{11}=1$,
let $\sigma=I/d^n$ and $\rho$ be any pure state associated with the maximal abelian Weyl group $\{w(\vec p,\vec 0): \vec p\in \Z_d^n\}$.
By Proposition \ref{prop:chara_conv},
$ \rho \boxtimes \sigma$ has characteristic function
\begin{align*}
\Xi_{ \rho \boxtimes \sigma} (\vec p, \vec q) = \Xi_\rho (\vec p, \vec q) \Xi_\sigma (\vec 0, \vec q) = \left\{
\begin{aligned}
&\Xi_\rho (\vec p, \vec q), && \vec q=\vec 0\\
&0, && \vec q\neq \vec 0
\end{aligned}
\right.
\end{align*}
which is equal to the characteristic function $\Xi_{\rho}$.
Hence $\rho \boxtimes \sigma= \rho$,
and $H_\alpha(\rho \boxtimes \sigma) =0< H_\alpha(\sigma )= n\log d$.
Similarly,
\eqref{ineq:entrop_con_2} may fail if $G$ is not even-parity positive.
\end{Rem}

In classical cases it was proved that
$
H\left(\frac{X_1+X_2+...X_{N+1}}{\sqrt{N+1}}\right)
\geq H\left(\frac{X_1+X_2+...X_{N}}{\sqrt{N}}\right) $,
where  $X_1, X_2,...$ are  i.i.d. square-integrable random variables \cite{artstein2004JAMS}. This is
 a classical analog of the second law of thermodynamics. Here, let us consider the 
 behavior of quantum entropies under our quantum convolution. 
Let us take the convolution repeatedly and define $\boxtimes^{N+1}\rho=(\boxtimes^N\rho)\boxtimes\rho$ inductively, where $\boxtimes^0\rho=\rho$.
We find that the quantum R\'enyi entropy $H_{\alpha}(\boxtimes^N\rho)$ is increasing w.r.t. the number $N$ of convolutions.

\begin{thm}[\bf Second law of thermodynamics  for quantum convolution]\label{0212prop2}
Let $G$ be even-parity positive, 
and  $\rho$ be an $n$-qudit state and  $\alpha\in [-\infty, +\infty]$.
Then the quantum R\'enyi entropy satisfies 
\begin{eqnarray*}
H_{\alpha}(\boxtimes^{N+1}\rho)
\geq H_{\alpha}(\boxtimes^N\rho) \;,\quad \forall N\geq 0 \;.
\end{eqnarray*}
\end{thm}
\begin{proof}
This is an immediate consequence of Proposition \ref{prop:entropy}.

\end{proof}

\begin{lem}\label{lem:eq_entr}
Let  $G$ be even-parity positive, and let $\rho =  \sum_{j=1}^{J} \mu_j P_j$ be the spectral decomposition of the quantum state $\rho$
where $P_j$ is the projection to the eigenspace corresponding to the eigenvalue $\mu_j$, and $\mu_i\neq\mu_j$ for any $i\neq j$.
For any $\alpha\not\in \{-\infty,0,\infty\}$,
the equality 
$$H_\alpha(\rho\boxtimes \sigma) = H_\alpha(\rho)$$
holds iff each $Q_j : = P_j \boxtimes \sigma$ is a projection of the same rank as $P_j$  and $Q_{j_1}\perp Q_{j_2}$ whenever $j_1\neq j_2$.
\end{lem}
\begin{proof}

We follow the notation in the proof of Theorem \ref{thm:major}.

("$\Rightarrow$"): Recall that we assume $\alpha\not\in \{-\infty,0,\infty\}$.
Since $ (\nu_1,...,\nu_{d^n}) \prec (\lambda_1,...,\lambda_{d^n})$ 
(and recall that we assume $\{\nu_j\}$ and $\{\lambda_j\}$ are non-increasing), we have
\begin{align}\label{0109shi2}
\sum_{j=1}^k \nu_j \le \sum_{j=1}^k \lambda_j,\quad k=1,2,..., d^n \;.
\end{align}
Moreover, $H_\alpha(\rho \boxtimes \sigma) = H_\alpha(\rho)$ iff every equality in (\ref{0109shi2}) holds, which means 
\begin{align*}
    (\nu_1,...,\nu_{d^n}) = (\lambda_1,...,\lambda_{d^n})\;.
\end{align*}

Assume $K$ is the largest number such that $\nu_1=\nu_K$. 
For any $k\le K$,
\begin{align}\label{0116shi1}
\nu_k = \sum_{j}m_{kj} \lambda_j = \left(\sum_{j=1}^K m_{kj}\right) \nu_k + \sum_{j>K} m_{kj}\nu_j \;.
\end{align}
If there exists $m_{kj}>0$ for some $j>K$, then
\begin{align*}
(\ref{0116shi1})< \left(\sum_{j=1}^K m_{kj}\right) \nu_k + \sum_{j>K} m_{kj} \nu_K= \nu_K \;,
\end{align*}
a contradiction.
Hence  $m_{kj}=0$ whenever $k\le K$ and $j>K$,
and therefore for every $k\le K$ we have
\begin{align*}
\sum_{j=1}^K m_{kj} =1 \;.
\end{align*}
The matrix $M$ is  doubly stochastic, so we also have $m_{kj}=0$ when $k>K$ and $j\le K$.
By the definition of $m_{kj}$,
we have $\langle \xi_k| \tau_j  | \xi_k \rangle=0$ whenever $k\le K< j$ or $j\le K<k$.
Denote $\rho_k= | \xi_k \rangle\langle \xi_k|$,
then $\tau_j \rho_k=0$ when $k\le K< j$ or when $j\le K<k$.
Therefore $\tau_j \sum_{k=1}^K \rho_k=0$ when $j>K$,
and $\tau_j \sum_{k=K+1}^{d^n} \rho_k=0$ when $j\le K$.

Denote $\psi_j=|\psi_j\rangle \langle \psi_j|$,
then by definition $\tau_j = \psi_j \boxtimes \sigma$.
When $k\le K$,
\begin{align*}
\langle \xi_k| \sum_{j=1}^K \tau_j | \xi_k \rangle = \sum_{j=1}^K m_{kj} =1 \;,
\end{align*}
that is
\begin{align}\label{0109shi4}
\langle \xi_k|\left(\frac 1K\sum_{j=1}^K \psi_j\right) \boxtimes \sigma| \xi_k \rangle =\frac 1K,\quad \forall k\le K \;.
\end{align}
Let us denote $\Phi= \left(\frac 1K\sum_{j=1}^K \psi_j\right) \boxtimes \sigma$.
Taking $\alpha=2$ in \eqref{ineq:entrop_con_2}, we have
\begin{align*}
 \left\|\Phi \right\|_2^2 \le \left\| \frac 1K \sum_{j=1}^K \psi_j \right\|_2^2 = \frac{1}{ K} \;.
\end{align*}
Expanding the matrix $\Phi$ under the basis $\{\xi_k\}$, we have
\begin{align*}
\sum_{k,k'}\left|\langle \xi_k|\Phi| \xi_{k'} \rangle\right|^2 \le \frac{1}{K} \;,
\end{align*}
and compared with (\ref{0109shi4}) we have that
\begin{align*}
\left(\frac 1K\sum_{j=1}^K \psi_j\right) \boxtimes \sigma=\Phi = \frac{1}{K}\sum_{k=1}^K | \xi_{k} \rangle\langle \xi_k| \;. 
\end{align*}
That is
\begin{eqnarray*}
P_1\boxtimes\sigma=Q_1 \;,
\end{eqnarray*}
where $P_1$ is the spectral projection of $\rho$ corresponding to the eigenvalue $v_1$, and 
$Q_1$ is the spectral projection of $\rho\boxtimes \sigma$ corresponding to $\lambda_1$. 
Repeat this process (by replacing $\rho$ by $(\rho - \nu_1 P_1)/\trace[\rho - \nu_1 P_1]$ )
and we obtain that, for each spectral projection $P_j$ of $\rho$, $Q_j:= P_j\boxtimes\sigma$ is a projection of the same rank as $P_j$, and $Q_j$ is a spectral projection of $\rho\boxtimes \sigma$.  The converse direction 
("$\Leftarrow$") is elementary, so  the proof is complete.
\end{proof}

Using this Lemma, we can study the conditions under which  equality holds.

\begin{thm}[\bf The case of equality]\label{0112thm1}
Let $G$ in \eqref{0125shi1} be positive and invertible, let 
 $\alpha\notin  \{-\infty, 0, +\infty\}$, let $\sigma$ be a quantum state,
and let
\begin{align}\label{0214shi1}
S=\{ w(- g_{10}^{-1} g_{11}\vec p, g^{-1}_{01} g_{00}\vec q): |\Xi_\sigma (\vec p, \vec q)|=1\}\;.
\end{align}
The equality 
\begin{align}\label{0110shi8}
H_\alpha(\rho \boxtimes  \sigma) =  H_\alpha(\rho)
\end{align}
holds for the state $\rho$, if and only if $\rho$ is in the abelian C*-algebra generated by $S$, i.e.,
$
\rho 
$ is a convex sum of MSPSs associated with $S$.
\end{thm}
\begin{proof}
First, we show $S$ is abelian. 
By Lemma \ref{0106lem2},
for every 
$(- g_{10}^{-1} g_{11}\vec p_1, g_{01}^{-1} g_{00}\vec q_1)$ and $(- g_{10}^{-1} g_{11}\vec p_2, g^{-1}_{01} g_{00}\vec q_2)$ in $ S$ we have $\ep{(\vec p_1,\vec q_1), (\vec p_2,\vec q_2 )}_s=0 $,
hence
\begin{align*}
&\ep{(- g_{10}^{-1} g_{11}\vec p_1, g_{01}^{-1} g_{00}\vec q_1), (- g_{10}^{-1} g_{11}\vec p_2, g_{01}^{-1} g_{00}\vec q_2)}_s \\
=& - g_{10}^{-1} g_{11} g_{01}^{-1} g_{00}\ep{(\vec p_1,\vec q_1), (\vec p_2,\vec q_2 )}_s=0 \;.
\end{align*}
Therefore, $S$ is abelian.
Similarly, 
the group\begin{align*}
S_1=\{ w( \vec p,  \vec q): |\Xi_\sigma (-N g_{10}\vec p, g_{01} \vec q)|=1\}
\end{align*}
is also an abelian Weyl group.
 
Now let $\rho$ be a state in the  C*-algebra generated by $S$,
and we show the equality \eqref{0110shi8} holds.
For each MSPS $\rho_j$ associated with $S$ we have $$|\Xi_{\rho_j}(\vec p,\vec q)| =\left\{
\begin{aligned}
 &1 &&  w( \vec p,\vec q) \in S,\\
&0 && w( \vec p,\vec q) \not \in S.
\end{aligned}\right.
$$ 
Combined with the definition \eqref{0214shi1} of $S$, we have  
$$|\Xi_{\rho_j \boxtimes \sigma} (\vec p,\vec q)| = |\Xi_{\rho_j}(N g_{11} \vec p, g_{00} \vec q) \Xi_{\sigma}(-Ng_{10} \vec p, g_{01} \vec q) | =\left\{
\begin{aligned}
 &1 &&  w( \vec p,\vec q) \in S_1,\\
&0 && w( \vec p,\vec q) \not \in S_1.
\end{aligned}\right.
$$
That is,
$\rho_j \boxtimes \sigma$ is an MSPS associated with the abelian group $S_1$.
Moreover,
$\rho_j \boxtimes \sigma = \rho_i \boxtimes \sigma$ if and only if $\rho_j = \rho _i$,
hence the map $\rho_j \mapsto \rho_j\boxtimes \sigma$ is a bijection from the set of MSPSs associated with $S$ to the set of MSPSs associated with $S_1$.
Therefore,
if $ \rho=\sum \mu_j \rho_j
 $ is a linear sum of MSPSs associated with $S$,
 then $ \rho\boxtimes \sigma =\sum \mu_j (\rho_j\boxtimes \sigma)
 $
is a linear sum (with the same coefficients) of MSPSs associated with $S_1$.
Hence, the equality \eqref{0110shi8} holds.

On the other hand, let us consider the case where  \eqref{0110shi8} holds.
Let $P$ be any spectral projection of $\rho$ and assume rank$(P)=R$.
In the following, we show $P$ is a stabilizer projection associated with $S$.
By Lemma \ref{lem:eq_entr},
$P \boxtimes \sigma$ is a projection $Q$ of rank $R$, hence $\|\frac1R P\|_2 = \|\frac1R Q\|_2$ and $\|\Xi_{\frac1R P}\|_2 = \|\Xi_{\frac1R Q}\|_2$.
While
\begin{align*}
\Xi_{\frac1RQ}(\vec p, \vec q) = \Xi_{ \frac1R P}(Ng_{11}\vec p, g_{00}\vec q) \Xi_\sigma (-Ng_{10}\vec p, g_{01}\vec q) \;,\quad \forall  (\vec p,\vec q) \in V^n \;.
\end{align*}
Thus, whenever $(\vec p,\vec q) \in V^n$ satisfies $\Xi_{ \frac1R P}(Ng_{11}\vec p, g_{00}\vec q)\neq 0$,
we  have $$|\Xi_\sigma (-Ng_{10}\vec p, g_{01}\vec q)|=1\;.$$
Equivalently,
\begin{align*}
\Xi_{ \frac1R P}(- g_{10}^{-1} g_{11}\vec p, g_{01}^{-1} g_{00}\vec q)\neq 0 \quad \Rightarrow \quad |\Xi_\sigma ( \vec p, \vec q)|=1 \;.
\end{align*}
That is, the characteristic function $ \Xi_{ \frac1R P}$ of $\frac1R P$ is supported on  $S$,
therefore
\begin{align*}
\frac1R P = \frac{1}{d^n} \sum_{ (\vec p, \vec q)\in S} \Xi_{ \frac1R P}( \vec p, \vec q) w(\vec p, \vec q) \in C^*(S) \;.
\end{align*}
Therefore $P$  is a stabilizer projection associated with $S$.
Hence $\rho$ is a linear sum of projections in the abelian C*-algebra generated by $S$.
\end{proof}

We have the following corollary for the case where $\sigma$ is a pure stabilizer state.

\begin{cor}
Suppose $G$ is positive, 
 $\alpha\notin \{-\infty, 0, +\infty\}$, and $\sigma$ is a pure stabilizer state whose stabilizer group $S$ is a maximal abelian subgroup of the Weyl group and $|S|=d^n$.
Denote
\begin{align*}
S_1 = \{ w(- g_{10}^{-1} g_{11}\vec p, g_{01}^{-1} g_{00}\vec q): (\vec p,\vec q)\in S\} \;.
\end{align*}
Then the  equality $H_\alpha(\rho \boxtimes  \sigma) =  H_\alpha(\rho) $ holds for some state $\rho$ iff $\rho$ is a convex combination of pure stabilizer states in $C^*(S_1)$, the (abelian) C* algebra generated by $S_1$.
\end{cor}

\subsection{Quantum Fisher information inequality}
We consider another important family of information measures that occur in quantum information theory, namely quantum Fisher information. In particular, we focus on divergence-based quantum Fisher information~\cite{Konig14}.

\begin{Def}[\bf Quantum Fisher information \cite{Konig14}] 
Given a smooth one-parameter family of states $\set{\rho_{\theta}}_{\theta}$, the divergence-based quantum Fisher information at $0$
is 
\begin{eqnarray*}
J(\rho_{\theta};\theta)|_{\theta=0}
:=\frac{d^2}{d\theta^2}D(\rho_0||\rho_{\theta})|_{\theta=0} \;.
\end{eqnarray*}
\end{Def}
Since the first derivative $\frac{d}{d\theta} D(\rho_0||\rho_{\theta})|_{\theta=0}=0$, the second derivative $J(\rho_{\theta};\theta)|_{\theta=0}$
quantifies the sensitivity of the divergence with respect to the change of the parameter $\theta$.
Since we only consider the divergence-based quantum Fisher information $J(\rho_{\theta};\theta)|_{\theta=0}$ in this work, 
we call it  quantum Fisher information for simplicity. 
If $\set{\rho_{\theta}}_{\theta}$ is a family of parameterized states defined by 
$\rho_{\theta}=\exp(i\theta H)\rho\exp(-i\theta H)$ with respect to a Hermitian operator $H$ for all $\theta\in \real$ , then 
the quantum Fisher information can be written as
\begin{align*}
J( \rho; H )=\frac{d^2}{d\theta^2}D(\rho||\rho_{\theta})|_{\theta=0} = \Tr{ \rho [H, [H, \log \rho]]} \;.
\end{align*}

In $n$-qudit systems,
we denote $X_k$ (resp., $Z_k$) to be the Pauli $X$ (resp., $Z$) operator on the $k$-th qudit.
For $R= X_k$ or $Z_k$ ($1\le k\le n$),
denote $|j\rangle_R $ to be an eigenvector of $R$  corresponding to the eigenvalue $\chi(j)$ with $j\in \mathbb{Z}_d$,
\begin{align}\label{Def:jR}
R |j\rangle_R = \chi(j)|j\rangle_R \;.
\end{align}
Define the completely dephasing channel $\Delta_R$ with respect to the eigenbasis of  the operator $R$ as
\begin{eqnarray}\label{Def:CDC}
\Delta_R (\rho) = \sum_j \langle j|\rho | j\rangle_R\proj{j}_R \;.
\end{eqnarray}
The Hermitian operator  $H^R_j$ for $j\in \mathbb{Z}_d$ is defined as 
\begin{align}\label{Defn:HRj}
H_j^R = \proj{j}_R \;,
\end{align}
and the corresponding parameterized unitary $U_j^R(\theta)$ is
\begin{align*}
U_j^R(\theta) = \exp( i \theta H_j^R ) \;.
\end{align*}
For any quantum state $\rho$, we consider the family of parameterized states
$\set{\rho^{R,j}_{\theta}}$ with
\begin{eqnarray*}
\rho^{R,j}_{\theta}
=U_j^R(\theta)\rho U_j^R(\theta)^\dag \;, \quad \theta\in\real \;,
\end{eqnarray*}
and its quantum Fisher information is 
\begin{align}\label{0103shi5}
J( \rho; H_j^R ) = \trace[ \rho [H_j^R, [H_j^R, \log \rho]]] \;.
\end{align}
We also denote 
\begin{align}\label{0103shi6}
J(\rho) = \sum_{k=1}^n \sum_{j=1}^d \left[J(\rho; H_j^{X_k}) + J(\rho; H_j^{Z_k}) \right]\;.
\end{align}

Now, let us consider the connection between the quantum convolutional channel
and completely dephasing channels $\Delta_X$ and $\Delta_Z$. 
We consider a $2n$-qudit system $\mathcal{H}_{AB}$ where subsystem $A$ denotes 
the first $n$-qudit system, and subsystem $B$ denotes the second $n$-qudit system.
And thus $X_k$ and $Z_{k}$  for $k\in [n]$ denote the Pauli $X$ and $Z$ operators acting on 
the first $n$-qudit system,  subsystem $A$. Likewise  $X_{n+k}$ and $Z_{n+k}$  for $k\in [n]$ denote the Pauli $X$ and $Z$ operators acting on the
second $n$-qudit system, subsystem $B$.
\begin{lem}\label{0103lem4}
Let $\rho_{AB}\in \mathcal{D}(\mathcal{H}_{AB})$ be a quantum state on $2n$-qudits,
and $\mathcal{E}_G$ be the quantum convolutional channel.

(1) If $G$ is odd-parity positive, then
\begin{align}\label{0103shi4_0}
\mathcal{E}_G\circ (id_A\otimes\Delta_{R_{n+k}}) (\rho_{AB}) &= \Delta_{R_k}\circ\mathcal{E}_G (\rho_{AB}) \;,
\end{align}
where $R= X$ or $Z$.

(2) If $G$ is even-parity positive, then
\begin{align}\label{0103shi4}
\mathcal{E}_G\circ (\Delta_{R_k} \otimes id_B) (\rho_{AB}) &= \Delta_{R_k}\circ\mathcal{E}_G (\rho_{AB}) \;,
\end{align}
where $R= X$ or $Z$.
\end{lem}
\begin{proof}
We prove the case where $G$ is even-parity positive and $R_k=Z_k$, and the other cases can be proved in the same way.
By Proposition \ref{prop:comm_wel},  we have
\begin{align*}
\mathcal{E}_G ( Z_k^p \rho_{AB} Z_k^{-p}) = Z_k^{g_{00}p} \mathcal{E}_G ( \rho_{AB}) Z_k^{-g_{00}p} \;.
\end{align*}
Then
\begin{align*}
\mathcal{E}_G\circ \Delta_{Z_k} \otimes id_B (\rho_{AB})=& \frac 1d \sum_{p\in \Z_d} \mathcal{E}_G( Z_k^p \rho_{AB} Z_k^{-p}) \\
=&\frac 1d \sum_{p\in \Z_d} Z_k^{g_{00}p} \mathcal{E}_G ( \rho_{AB}) Z_k^{-g_{00}p}\\
=&\Delta_{Z_k}\circ\CEE_G (\rho_{AB}) \;,
\end{align*}
where the last equality comes from the fact that $G$ is even-parity positive.
\end{proof}

\begin{thm}[\bf Convolution decreases Fisher information of states]\label{thm:fisher}
Let $\rho_A, \rho_B$ be two $n$-qudit states.

(1) If $G$ is odd-parity positive, then
\begin{align*}
J( \rho_A \boxtimes \rho_B) \le  J(\rho_B) \;.
\end{align*}

(2) If $G$ is even-parity positive, then

\begin{align*}
J( \rho_A \boxtimes \rho_B) \le J(\rho_A) \;.
\end{align*}

(3) If $G$ is positive, then 
\begin{align*}
J( \rho_A \boxtimes \rho_B) \le \min\set{J(\rho_A), J(\rho_B)} \;.
\end{align*}

\end{thm}
\begin{proof}
We prove  $J( \rho_A \boxtimes \rho_B) \le  J(\rho_A) $ for $G$ being  even-parity positive, and the other case can be proved in the same way.

By direct calculation, we find that 
\begin{align}\label{0103shi1}
\sum_{j=1}^d [H_j^R, [H_j^R,  \rho]] = 2 (\rho - \Delta_R(\rho)) \;.
\end{align}
Based on the definition of quantum Fisher information in \eqref{0103shi6}, 
we have
\begin{align*}
J(\rho_A) 
=& \sum_{k=1}^n \sum_{j=1}^d J(\rho_A; H_j^{X_k}) + J(\rho_A; H_j^{Z_k})\\
=& \sum_{k=1}^n \sum_{j=1}^d \Tr{ \rho_A \left[H_j^{X_k}, \left[H_j^{X_k}, \log \rho_A\right]\right]} + \Tr{ \rho_A \left[H_j^{Z_k}, \left[H_j^{Z_k}, \log \rho_A\right]\right]}\\
=&  2\sum^n_{k=1}\left[\Tr{\rho_A\log\rho_A}-\Tr{\rho_A\Delta_{X_k}(\log\rho_A)}\right.\\
&\qquad\qquad \left.+\Tr{\rho_A\log\rho_A}
-\Tr{\rho_A\Delta_{Z_k}(\log\rho_A)}\right] \;.
\end{align*}

Let $\CEE_{\rho_B}$ be the quantum channel with the given quantum state $\rho_B$ defined as
\begin{align*}
\CEE_{\rho_B} (\rho) = \rho \boxtimes \rho_B \;.
\end{align*}
Similarly,  $J(\rho_A\boxtimes\rho_B)$ can also be written as 
 \begin{align*}
  J(\CEE_{\rho_B}(\rho_A)) 
 = 2\sum^n_{k=1} \big[
 & \trace \left[\CEE_{\rho_B}(\rho_A)\log\CEE_{\rho_B}(\rho_A)\right]
 -\Tr{\CEE_{\rho_B}(\rho_A)\Delta_{X_k}(\log\CEE_{\rho_B}(\rho_A))} 
 \\
 &  +\Tr{\CEE_{\rho_B}(\rho_A)\log\CEE_{\rho_B}(\rho_A)}-\Tr{\CEE_{\rho_B}(\rho_A)\Delta_{Z_k}(\log\CEE_{\rho_B}(\rho_A))} \big] \;.
 \end{align*}
Therefore, we only need to prove that
\begin{align}
\nonumber &\Tr{\rho_A\log\rho_A}-\Tr{\rho_A\Delta_{R_k}(\log\rho_A)}\\
&\qquad\qquad\geq
\label{0103shi10}\Tr{\CEE_{\rho_B}(\rho_A)\log\CEE_{\rho_B}(\rho_A)}-\Tr{\CEE_{\rho_B}(\rho_A)\Delta_{R_k}(\log\CEE_{\rho_B}(\rho_A))} \;.
\end{align}
One can rewrite the left-hand side as, 
\begin{align*}
&\Tr{\rho_A\log\rho_A}-\Tr{\rho_A\Delta_{R_k}(\log\rho_A)}\\
&\ \ =  \trace\left[ \rho_A   \log \rho_A\right] -  \trace \left[ \rho_A   \log \Delta_{R_k} ( \rho_A)\right] \\
&\ \ \qquad\qquad +   \trace \left[ \rho_A  \log \Delta_{R_k} ( \rho_A)\right]
- \trace\left[  \Delta_{R_k} ( \rho_A ) \log \rho_A\right] \\
&\ \ = \trace\left[ \rho_A   \log \rho_A\right] -  \trace \left[ \rho_A   \log \Delta_{R_k} ( \rho_A)\right] \\
&\ \ \qquad\qquad +  \trace \left[ \Delta_{R_k} (\rho_A )  \log \Delta_{R_k} ( \rho_A)\right]-   \trace\left[  \Delta_{R_k} ( \rho_A ) \log \rho_A\right] \\
&\ \ =D(  \rho_A\| \Delta_{R_k} ( \rho_A ))+ D(  \Delta_{R_k} ( \rho_A )\| \rho_A ) \;,
\end{align*}
where the  second equality comes from the fact that, for every state $\rho$,
\begin{align*}
\trace \left[ \rho    \log \Delta_{R_k} ( \rho )\right] =   \trace \left[\Delta_{R_k}( \rho )   \log \Delta_{R_k} ( \rho )\right] \;.
\end{align*}
Since  $\text{supp}(\rho_A)\subseteq \text{supp} (\Delta_{R_k} ( \rho_A ))$, i.e., $D(  \rho_A\| \Delta_{R_k} ( \rho_A ))$ is well-defined and $\leq \log d^n$. If $supp (\Delta_{R_k} ( \rho_A )) \subseteq supp(\rho_A) $, then $D(  \Delta_{R_k} ( \rho_A )\| \rho_A ) $ is also well-defined, otherwise $D(  \Delta_{R_k} ( \rho_A )\| \rho_A ) =+\infty$ and \eqref{0103shi10} also holds.
Similarly, the right-hand side can be rewritten as
\begin{align*}
&\Tr{\CEE_{\rho_B}(\rho_A)\log\CEE_{\rho_B}(\rho_A)}-\Tr{\CEE_{\rho_B}(\rho_A)\Delta_{R_k}\log\CEE_{\rho_B}(\rho_A)}\\
 & \qquad\qquad= D(\Delta_{R_k} ( \CEE_{\rho_B}(\rho_A ))\| \CEE_{\rho_B}(\rho_A) ) + D( \CEE_{\rho_B}( \rho_A)\| \Delta_{R_k} (\CEE_{\rho_B}( \rho_A ))) \;.
\end{align*}
Based on  Lemma \ref{0103lem4}, if $G$ is even-parity positive,
\begin{align}\label{0103shi11}
\Delta_{R_k} \circ \CEE_{\rho_B} (\rho) = \Delta_{R_k} (\rho \boxtimes \rho_B)
= \CEE_{\rho_B} \circ \Delta_{R_k} (\rho) \;.
\end{align}
Hence
\begin{align*}
&\Tr{\CEE_{\rho_B}(\rho_A)\log\CEE_{\rho_B}(\rho_A)}-\Tr{\CEE_{\rho_B}(\rho_A)\Delta_{R_k}(\log\CEE_{\rho_B}(\rho_A))}\\
&\qquad\qquad =  D(\Delta_{R_k} ( \CEE_{\rho_B}(\rho_A ))\| \CEE_{\rho_B}(\rho_A) ) +  D( \CEE_{\rho_B}( \rho_A)\| \Delta_{R_k} (\CEE_{\rho_B}( \rho_A )))\\
&\qquad\qquad = D(  \CEE_{\rho_B}(\Delta_{R_k} ( \rho_A ))\| \CEE_{\rho_B}(\rho_A) ) + D( \CEE_{\rho_B}( \rho_A)\| \CEE_{\rho_B}(\Delta_{R_k} ( \rho_A )))\\
&\qquad\qquad \leq  D(  \Delta_{R_k} ( \rho_A )\| \rho_A ) +  D(  \rho_A\| \Delta_{R_k} ( \rho_A ))\\
&\qquad\qquad =\Tr{\rho_A\log\rho_A}-\Tr{\rho_A\Delta_{R_k}(\log\rho_A)} \;,
\end{align*}
where the inequality comes from the monotonicity of relative entropy under quantum channels.
Therefore, \eqref{0103shi10} holds and the proof is complete.
\end{proof}

We introduce a measure representing the number of non-identity local Pauli operators in the unitary $w(\vec p, \vec q)$.

\begin{Def}
For any $(\vec{p}, \vec q )\in V^n$, denote 
\begin{eqnarray*}
\interleave (\vec p, \vec q) \interleave
=\sum^n_{k=1}\delta_{p_k\neq 0} + \sum^n_{k=1}\delta_{q_k\neq 0}
\end{eqnarray*}
to be the number of nonzero coordinates in $(\vec p,\vec q)$.
\end{Def}

Now, we define the Liouvillian,  based on Pauli X and Z operators. With this, we establish the quantum de Bruijn inequality for qudit systems.

\begin{Def}[\bf Liouvillian]\label{0103def1}
The Liouvillian $\CLL$ on $n$-qudits is:
\begin{align}\label{0103shi8}
\CLL =\sum_{k=1}^n \left(\CLL_{X_k}+ \CLL_{Z_k}\right) \;,
\end{align}
and with $H^R_j$ defined in \eqref{Def:jR}--\eqref{Defn:HRj},
\begin{align*}
\CLL_R(\rho) = -\frac 14 \sum_{j=1}^d \left[H_j^R, \left[H_j^R, \rho\right]\right] \;.
\end{align*}
\end{Def}

\begin{prop}\label{0103prop1}
The  Liouvillian $\CLL$ has the following properties:

(1) The Liouvillian is Hermitian with respect to the inner product given by the trace, i.e.,
\be
\Tr{\CLL(A)B}=\Tr{A\CLL(B)}\;.
\ee

(2) The action of $\CLL$ on Weyl operators has the form
\begin{eqnarray*}
\CLL(w(\vec{p}, \vec q))
=-\frac{1}{2}\interleave (\vec{p}, \vec q)\interleave w(\vec{p}, \vec q) \;,
\end{eqnarray*}
so
\begin{align*}
e^{t\CLL } (w(\vec{p},\vec q)) = \exp \left(-\frac 12 \interleave (\vec{p}, \vec q) \interleave t\right) w(\vec{p}, \vec q) \;.
\end{align*}
\end{prop}
\begin{proof}

(1) By direct calculation, we have 
\begin{align}
\sum_{j=1}^d [H_j^R, [H_j^R,  \rho]] = 2 (\rho - \Delta_R(\rho)) \;,
\end{align}
and thus 
\begin{eqnarray}\label{0103shi9}
\mathcal{L}_{R}(\rho)=-\frac{1}{2}(\rho-\Delta_R(\rho)) \;.
\end{eqnarray}
Hence
\begin{align}\label{0103shi2}
\CLL(\rho) = -n\rho + \frac 12 \sum_{k=1}^n\left( \Delta_{X_k} (\rho)+ \Delta_{Z_k} (\rho) \right)\;.
\end{align}
Thus $ \CLL$ is Hermitian.

(2) For any Weyl operator $w(\vec{p}, \vec q)$ , we have 
\begin{align*}
\Delta_{X_k} (w(\vec{p}, \vec q)) = \left\{
\begin{aligned}
&w(\vec{p}, \vec q) \;, && p_k=0 \;,\\
&0 \;, && p_k\neq 0 \;,
\end{aligned}\right.
\end{align*}
and
\begin{align*}
\Delta_{Z_k} (w(\vec p, \vec q)) = \left\{
\begin{aligned}
&w(\vec{p}, \vec q) \;, && q_k=0 \;,\\
&0 \;, && q_k\neq 0 \;.
\end{aligned}\right.
\end{align*}
Due to  \eqref{0103shi2},  we have
\begin{align*}
\CLL(w(\vec{p}, \vec q)) =& -nw(\vec p ,\vec q) + \frac 12 \sum_{k=1}^n \left[\Delta_{X_k} (w(\vec{p}, \vec q))+ \Delta_{Z_k} (w(\vec{p}, \vec q))\right]\\
=& -n w(\vec{p}, \vec q)+ \frac 12 \left( \sum_{ k:\; p_k= 0} 1+ \sum_{ k:\; q_k= 0}1 \right)w(\vec{p}, \vec q)\\
=& -\frac 12 \interleave (\vec{p}, \vec q) \interleave w(\vec{p}, \vec q) \;.
\end{align*}
\end{proof}

We give the following proposition to provide some commutation relation between 
the convolutional channel and quantum Markov semigroup, which may be of independent interest.
\begin{prop}
Let $\rho_{AB}\in \mathcal{D}(\mathcal{H}_{AB})$ be an $2n$-qudit state,
and $\mathcal{E}_G$ be the quantum convolutional channel.

(1) If $G$ is odd-parity positive, 
then 
\begin{align*}
\mathcal{E}_G\circ e^{t \CLL_B} (\rho_{AB}) = e^{t \CLL} \circ \mathcal{E}_G(\rho_{AB}) \;.
\end{align*}

(2) If $G$ is even-parity positive, 
then 
\begin{align*}
\mathcal{E}_G\circ e^{t \CLL_A} (\rho_{AB}) = e^{t \CLL} \circ \mathcal{E}_G(\rho_{AB}) \;.
\end{align*}

(3) If $G$ is positive, then
\begin{align*}
\mathcal{E}_G\circ e^{t_A \CLL_A} \otimes e^{t_B \CLL_B} (\rho_{AB}) = e^{t \CLL} \circ \mathcal{E}_G(\rho_{AB}) \;,
\end{align*}
where $t= t_A+t_B$.
\end{prop}

\begin{proof}
We only prove the case (3), and the other two cases can be proved in a similar way.
For any Weyl operator $w(\vec p,\vec q)$,
\begin{align*}
& e^{t_A \CLL_A} \otimes e^{t_B \CLL_B} \circ \CEE^\dag_G (w(\vec{p}, \vec q)) \\
= & e^{t_A \CLL_A} \otimes e^{t_B \CLL_B} (w(Ng_{11}\vec{p}, g_{00}\vec q) \otimes w(-Ng_{10}\vec{p}, g_{01} \vec q)) \\
=& \exp\left(-\frac 12 \interleave(Ng_{11}\vec{p}, g_{00}\vec q)\interleave t_A -\frac 12 \interleave(-Ng_{10}\vec{p}, g_{01} \vec q) \interleave t_B\right)\\ &\quad \times w(Ng_{11}\vec{p}, g_{00}\vec q) \otimes w(-Ng_{10}\vec{p}, g_{01} \vec q)\\
=& \exp\left(-\frac 12 \interleave (\vec{p}, \vec q) \interleave t\right) w(Ng_{11}\vec{p}, g_{00}\vec q) \otimes w(-Ng_{10}\vec{p}, g_{01} \vec q)\\
=& \CEE^\dag_G \left( \exp\left(-\frac 12 \interleave(\vec{p}, \vec q) \interleave t\right)w(\vec{p}, \vec q) \right)\\
=& \CEE^\dag_G \circ e^{t\CLL}( w(\vec{p}, \vec q) ) \;,
\end{align*}
where the first and the fourth equalities come from Lemma \ref{lem:adj_con}, the second and the last equalities come from  Proposition \ref{0103prop1}, and the third equality comes from the fact that $G$ is positive.
\end{proof}

\begin{thm}[\bf Quantum de Bruijn for qudits]
For any quantum state $\rho\in\mathcal{D}(\mathcal{H}^{\ot n})$, we have
\begin{align*}
\frac{d}{dt} \bigg|_{t=0}  H\left( e^{t\CLL } (\rho) \right)= \frac 14 J(\rho) \;.
\end{align*}
\end{thm}
\begin{proof}
\begin{align*}
\frac{d}{dt} \bigg|_{t=0} H\left( e^{t\CLL } (\rho) \right) =& - \trace[ \CLL(\rho) \log \rho]\\
=& - \Tr{\rho \CLL(\log \rho)} \\
=& -\sum_{ k=1}^n\left( \Tr{\rho \CLL_{X_k}   (\log \rho)} - \Tr{\rho \CLL_{Z_k}   (\log \rho)}\right)\\
=& \frac 14 \sum_{ k=1}^n \sum_{j=1}^d \left(\Tr{\rho \left[H_j^{X_k}, \left[H_j^{X_k}, \log \rho\right]\right]} + \Tr{\rho \left[H_j^{Z_k}, \left[H_j^{Z_k}, \log \rho\right]\right]}\right)\\
=& \frac14 J(\rho) \;,
\end{align*}
where the second equation comes from the fact that $\CLL$ is Hermitian.
\end{proof}

The quantum de Bruijn identity for qudits establishes a connection between the quantum Fisher information and the rate of quantum entropy change under the Liouvillian $\CLL$.

\subsection{Stabilizer states in the convolution }

In this section, we study the role of stabilizer states in the convolutional channel.
First, let us discuss what kind of input states $\rho,\sigma$ 
will make the output state have the minimal output entropy. 
\begin{thm}\label{thm:min_outen}
Let  $G$ be positive and invertible, and let $\rho$ and $\sigma$ be two $n$-qudit states.
Then the output state $\mathcal{E}_G(\rho\ot\sigma)$ has the minimal 
output entropy iff both $\rho,\sigma$ are pure stabilizer states,
and their stabilizer groups $S_1$ and $S_2$
satisfy 
\begin{align}\label{0125shi2}
S_1 = \{ w(- g_{10}^{-1} g_{11}\vec p, g_{01}^{-1} g_{00}\vec q):   w(\vec p, \vec q)\in S_2\} \;.
\end{align}
\end{thm}

\begin{proof}
If $\rho\boxtimes\sigma$ has the minimal output entropy,
then it is a pure state.
By Proposition \ref{prop:entropy}, both $\rho$ and $\sigma$ are pure,
so the equality $H(\rho\boxtimes\sigma)= H(\rho) = H(\sigma) $ holds.
Then Theorem \ref{0112thm1} tells us that $\rho$ and $\sigma$ are pure  stabilizer states, 
and \eqref{0125shi2} is satisfied.
\end{proof}

Classical capacity, in general, quantifies the maximal rate of reliable classical information transmission. Here, we study the classical capacity of the convolutional channel and explore its relation to the stabilizerness of the input states.
We also consider the 
Holevo capacity of the quantum channels, which can be used to quantify the classical capacity of a memory-less quantum channel~\cite{Schumacher97,Holevo98}.

\begin{Def}[\bf Holevo capacity]
The Holevo capacity $\chi_H(\mathcal{E})$ of a quantum channel $\mathcal{E}$ is 
\begin{eqnarray}
\chi_H(\mathcal{E})
=\max_{\set{p_i,\rho_i}}
H\left(\sum_ip_i\mathcal{E}(\rho_i)\right)
-\sum_ip_iH(\mathcal{E}(\rho_i)) \;,
\end{eqnarray}
where the maximum is taken over all ensembles $\set{p_i,\rho_i}$ of possible input states $\rho_i$ occurring with
probabilities $p_i$.
\end{Def}

\begin{thm}[\bf Holevo capacity bound: general case]\label{0118thm1}
Let  $G$ be positive and invertible, and $\sigma$ be an $n$-qudit state.
The Holevo capacity of the quantum channel $\mathcal{E}_{\sigma}$  satisfies
\begin{eqnarray}\label{eq:lo_up_hol}
n\log d-H(\CMM(\sigma))\leq \chi_H(\mathcal{E}_{\sigma})\leq  n\log d-H(\sigma) \;.
\end{eqnarray}
In particular, if $\sigma\in MSPS$,
then 
\begin{eqnarray*}
\chi_H(\mathcal{E}_{\sigma})
=n\log d-H(\sigma) \;.
\end{eqnarray*}
\end{thm}
\begin{proof}
Based on Proposition \ref{prop:entropy}, 
we have $\sum_ip_iH(\mathcal{E}_{\sigma}(\rho_i))\geq \sum_ip_iH(\sigma)=H(\sigma)$. The entropy
$H(\sum_ip_i\mathcal{E}_{\sigma}(\rho_i))\leq\log d^n$, hence we have 
\begin{eqnarray*}
H\left(\sum_ip_i\mathcal{E}_{\sigma}(\rho_i)\right)
-\sum_ip_iH(\mathcal{E}_{\sigma}(\rho_i))
\leq \log d^n-H(\sigma) \;.
\end{eqnarray*}

For the lower bound,
let $S$ be the abelian Weyl group associated with $\mathcal{M}(\sigma)$ and assume $|S|=d^r$. 
Consider the following two sets of Weyl operators:
\begin{eqnarray*}
S_1&=&\set{w(Ng_{11}\vec p, g_{00}\vec q): w(-Ng_{10}\vec p, g_{01}\vec q)\in S} \;,\\
S_2&=&\set{w(\vec p, \vec q): w(-Ng_{10}\vec p,g_{01} \vec q)\in S} \;.
\end{eqnarray*}
Clearly, $|S_1|=|S_2|=|S|=d^r$, and both $S_1$ and $S_2$ are abelian groups.

Let $\rho$ be an MSPS associated with $S_1$. Then  
the characteristic function of $\rho$ is supported on  $S_1$. 
Since
\begin{eqnarray*}
\Xi_{\rho\boxtimes \sigma}
(\vec p, \vec q)
=\Xi_{\rho}(Ng_{11}\vec p, g_{00}\vec q)
\Xi_{\sigma}(-Ng_{10}\vec p, g_{01}\vec q) \;,
\end{eqnarray*}
$\Xi_{\rho\boxtimes\sigma}(\vec p, \vec q)\neq 0$ implies that $w(Ng_{11}\vec p, g_{00}\vec q)\in S_1$ .
By the definition of $S_1$, we have $w(-Ng_{10}\vec p, g_{01}\vec q)\in S$,  then $w(\vec p, \vec q)\in S_2$, that is 
$\text{supp}\Xi_{\rho\boxtimes\sigma}\subseteq S_2$. 
Moreover, if $w(\vec p, \vec q)\in S_2$, then by the definition 
of $S_1, S_2$ and $\rho$, we have $|\Xi_{\rho\boxtimes\sigma}(\vec p, \vec q)|=1$. Hence, $\rho\boxtimes\sigma$
is an MSPS associated with $S_2$. Therefore
\begin{eqnarray*}
H(\rho\boxtimes\sigma)=\log d^{n-r} \;.
\end{eqnarray*}

Let $\{ \rho_i\}_{i=1}^{d^r}$ be the $d^{r}$ MSPSs associated with the abelian Weyl group $S_1$, 
and let the ensemble  be $\set{p_i=\frac{1}{d^r}, \rho_i}^{d^r}_{i=1}$, then
we have
\begin{eqnarray*}
\sum_ip_i\rho_i=I_n/d^n \;.
\end{eqnarray*}
Thus, $H(\sum_ip_i\rho_i\boxtimes\sigma)=H(I_n/d^n\boxtimes\sigma)=\log d^n$, and
$\sum_ip_iH(\rho_i\boxtimes\sigma)=\log d^{n-r}$. Therefore, $\chi_H(\mathcal{E}_{\sigma})\geq \log d^n-\log d^{n-r}=\log d^n-H(\CMM(\sigma))$.
\end{proof}

\begin{thm}[\bf Holevo capacity bound: pure state case]\label{thm:holv_stab}
Let  $G$ be positive and invertible,  and $\sigma$ be a pure state.
The  quantum channel $\mathcal{E}_{\sigma}$  has the maximal  Holevo capacity
$n\log d$ iff 
$\sigma$ is a stabilizer state.

\end{thm}
\begin{proof}

If the pure state
$\sigma$ is a stabilizer state, the size of the stabilizer group is $d^n$, and thus
 the Holevo 
capacity 
$\chi_H(\mathcal{E}_{\sigma})=n\log d$ by  \eqref{eq:lo_up_hol} as
$H(\CMM(\sigma))=H(\sigma)=0$.

On the other hand, if there exists an ensemble $\set{p_i,\rho_i}$ such that 
\begin{eqnarray*}
H\left(\sum_ip_i\mathcal{E}_{\sigma}(\rho_i)\right)
-\sum_ip_iH(\mathcal{E}_{\sigma}(\rho_i))
=\log d^n \;,
\end{eqnarray*}
then 
\begin{eqnarray*}
H\left(\sum_ip_i\mathcal{E}_{\sigma}(\rho_i)\right)&=&n\log d \;,\\
H(\mathcal{E}_{\sigma}(\rho_i))&=&0 \;.
\end{eqnarray*}
That is, the output entropy $H(\mathcal{E}_G(\rho_i\ot\sigma))=H(\mathcal{E}_{\sigma}(\rho_i))=0$. By 
Theorem \ref{thm:min_outen}, $\sigma$ is  a stabilizer state.
\end{proof}

The above two theorems consider the case where the parameter matrix $G$ is positive. 
The following result deals with the case where $G$ is 
just odd-parity positive.

\begin{prop}\label{0211prop1}
Let the nontrivial parameter matrix $G$ be
odd-parity positive but not even-parity positive, and $\sigma$ be an $n$-qudit state. 
Then 
\begin{eqnarray*}
\chi_H(\mathcal{E}_{\sigma})\leq H\left(\frac{I_n}{d^n}\boxtimes \sigma\right)-H(\sigma) \;.
\end{eqnarray*}
\end{prop}
\begin{proof}
Since the
nontrivial parameter matrix $G$ is 
odd-parity positive and not even-parity positive, 
either $g_{00}=0$ or $g_{11}=0$. 
Without loss of generality, we consider the case where $g_{11}=0$, and $g_{00}, g_{01}, g_{10}\neq 0$. 

Since $g_{00}\neq 0$, by Proposition \ref{prop:comm_wel} we have
\begin{align*}
\mathcal{E}_G ( (Z_k^p \ot I) \rho_{AB} (Z_k^{-p} \ot I)) = Z_k^{g_{00}p} \mathcal{E}_G ( \rho_{AB}) Z_k^{-g_{00}p} \;.
\end{align*}
Then
\begin{align*}
\mathcal{E}_G\circ (\Delta_{Z_k} \otimes id_B) (\rho_{AB})
=& \frac 1d \sum_{p\in \Z_d} \mathcal{E}_G(  (Z_k^p \ot I) \rho_{AB} (Z_k^{-p} \ot I))\\ 
=&\frac 1d \sum_{p\in \Z_d} Z_k^{g_{00}p} \mathcal{E}_G ( \rho_{AB}) Z_k^{-g_{00}p} =\Delta_{Z_k}\circ\CEE_G (\rho_{AB}) \;.
\end{align*}
Denote $\Delta_Z=\Delta_{Z_1}\ot \Delta_{Z_2}\ot...\ot\Delta_{Z_n}$, we have
\begin{eqnarray*}
\Delta_Z(\rho \boxtimes \sigma)
=(\Delta_Z\rho)\boxtimes \sigma \;.
\end{eqnarray*}
It is easy to verify that $\Xi_{\Delta_Z(\rho)}(\vec p, \vec q)=\Xi_{\rho}(\vec p, \vec q)\delta_{\vec q, \vec 0}$.
Moreover, the characteristic function of $(\Delta_Z\rho)\boxtimes \sigma$ is 
\begin{align*}
&\Xi_{(\Delta_Z\rho)\boxtimes \sigma}(\vec p,\vec q)\\
=&\Xi_{\Delta_Z(\rho)}(\vec 0, g_{00}\vec q)\Xi_{\sigma}(-Ng_{10}\vec p, g_{01}\vec q)\\
=&\Xi_{\rho}(\vec 0, g_{00}\vec q)\delta_{\vec q, \vec 0}\Xi_{\sigma}(-Ng_{10}\vec p, g_{01}\vec q)
=\Xi_{I_n/d^n\boxtimes \sigma} \;.
\end{align*}
Therefore
\begin{eqnarray*}
(\Delta_Z\rho)\boxtimes \sigma
=I_n/d^n\boxtimes \sigma \;,
\end{eqnarray*}
and 
\begin{eqnarray*}
H(\rho\boxtimes\sigma)
\leq H(\Delta_Z(\rho \boxtimes \sigma))
=H((\Delta_Z\rho)\boxtimes \sigma)
=H(I_n/d^n\boxtimes\sigma) \;.
\end{eqnarray*}
Since $G$ is odd-parity positive, by Proposition \ref{prop:entropy}  we have
\begin{eqnarray*}
H(\rho\boxtimes\sigma)\geq H(\sigma) \;.
\end{eqnarray*}
Therefore, 
\begin{eqnarray*}
\chi_H(\mathcal{E}_{\sigma})\leq H(I_n/d^n\boxtimes \sigma)-H(\sigma) \;.
\end{eqnarray*}

The case  where $g_{00}=0$, and $g_{11}, g_{01}, g_{10}\neq 0$ can be proved in the same way by replacing 
$\Delta_Z$ by $\Delta_X$.
\end{proof}

From Proposition \ref{0211prop1} 
we have the following Corollary on the relationship between the 
Holevo capacity and the coherence of a state $\sigma$ for the convolution channels 
$G=[0,1;1,1]$ and $G=[1,1;1,0]$.  
Here the coherence is quantified by the 
relative entropy. That is, given an orthonormal basis $\set{\ket{i}}$, 
 the relative entropy of coherence is
 $C_r(\rho):=D(\rho||\Delta(\rho))$ \cite{BaumgratzPRL14,Plenio17}, where $\Delta$  is the complete dephasing channel with respect to the given basis.

\begin{cor}
Let the parameter matrix $G=[0,1;1,1]$ (resp., $G=[1,1;1,0]$), and $\sigma$ be an $n$-qudit state. Then the Holevo capacity of the convolutional channel $\mathcal{E}_{\sigma}$ satisfies
\begin{eqnarray*}
\chi_H(\mathcal{E}_{\sigma})\leq C_{r,X}(\sigma) \;\; ( \text{resp., }C_{r,Z}(\sigma)) \;,
\end{eqnarray*}
where $C_{r,X}(\sigma)$ (resp., $C_{r,Z}(\sigma)$) is the relative entropy of coherence 
with respect to the eigenbasis of Pauli X (or Z) operators.
\end{cor}
\begin{proof}
If  $g_{11}=0$ and $g_{00}=g_{01}=g_{10}=1$,  
we have
\begin{align*}
\Xi_{ I/d^n\boxtimes \sigma} (\vec p, \vec q) = \Xi_{I/d^n} ( \vec 0, \vec q) \Xi_\sigma (\vec p, \vec q) 
= \delta_{ \vec q,\vec 0 }  \Xi_\sigma (\vec p, \vec q) 
= \Xi_\sigma (\vec p, \vec 0) \delta_{ \vec q,\vec 0 } 
=\Xi_{\Delta_Z(\sigma)}(\vec p, \vec q) \;.
\end{align*}
Hence  $I/d^n\boxtimes \sigma=\Delta_Z(\sigma)$,
and by Proposition \ref{0211prop1} we have
\begin{eqnarray*}
\chi_H(\mathcal{E}_{\sigma})\leq H(\Delta_Z(\sigma))-H(\sigma)=C_{r, Z}(\sigma) \;,
\end{eqnarray*}
and the proof is complete.
\end{proof}
\goodbreak

\subsection{Examples}\label{subsec:examp}
Having considered the general case, let us now consider some specific examples in detail.

\subsubsection{Hadamard Convolution on Qudits}\label{Hadamard convolution}

As our first example, let the parameter matrix $G$ equal the Hadamard 
matrix 
\begin{equation}\label{eq:had}
H=\left[
\begin{array}{cc}
1&1\\
1&d-1
\end{array}
\right]
\equiv
\left[
\begin{array}{cc}
1&1\\
1&-1
\end{array}
\right],
\end{equation}
which is invertible for odd $d$.

\begin{Def}[\bf Hadamard Convolution]\label{Def:Had_Conv}
The key unitary $U_H$ corresponding to the Hadamard matrix \eqref{eq:had} is 
\begin{eqnarray}
U_H=\sum_{\vec{i},\vec{j}}
\ket{\vec{i}}\bra{\vec i+\vec j}
\ot \ket{\vec j}\bra{\vec i-\vec j},
\end{eqnarray}
 where  $| \vec i \rangle = |  i_1 \rangle \otimes \cdots \otimes |  i_n \rangle \in \CHH^{\otimes n} $.
The convolution  of two $n$-qudit states $\rho$ and $\sigma$ is
\begin{eqnarray}
\rho\boxtimes_H\sigma=\Ptr{B}{U_H\rho\ot\sigma U^\dag_H}.
\end{eqnarray}
The corresponding convolutional channel $\mathcal{E}_H$ is $\mathcal{E}_H(\cdot)
=\Ptr{B}{U_H(\cdot) U^\dag_H}$.

\end{Def}
\begin{prop}
Given two $n$-qudit states $\rho$ and $\sigma$, the characteristic function 
satisfies
\begin{align}
\Xi_{ \rho \boxtimes_H \sigma} (\vec p, \vec q) = \Xi_\rho (2^{-1}\vec p, \vec q) \Xi_\sigma (2^{-1}\vec p, \vec q),
\forall (\vec p, \vec q)\in V^n.
\end{align}
\end{prop}

\begin{prop}[\bf Hadamard convolution is abelian]
For any two $n$-qudit states $\rho$ and $\sigma$, 
\begin{eqnarray}
\rho\boxtimes_H\sigma=
\sigma\boxtimes_H\rho.
\end{eqnarray}

\end{prop}

\begin{prop}[\bf Wigner function positivity]
Given two $n$-qudit states $\rho$ and $\sigma$, the discrete Wigner function of  $\rho\boxtimes_H\sigma$  satisfies
\begin{eqnarray}
W_{\rho\boxtimes_H\sigma}(\vec u, \vec v)
=\sum_{\vec{u}_1,\vec{v}_1}
W_{\rho}(\vec{u}_1,2\vec{v}_1)
W_{\sigma_{\vec u, \vec v}}(\vec{u}_1, 2\vec{v}_1)
=\frac{1}{d^n}\Tr{\rho \sigma_{\vec u, \vec v}}\geq 0,
\end{eqnarray}
where $\sigma_{\vec u, \vec v}=w(\vec u, 2\vec v)T(\vec 0,\vec 0)\sigma T(\vec 0, \vec 0)w(\vec u, 2\vec v)^\dag$ (See
the properties of operators $T(\vec p, \vec q)$ in Remark \ref{rem:phas_op}).
\end{prop}
\begin{proof}
By Proposition \ref{prop:wig_conv},
\begin{eqnarray*}
W_{\rho\boxtimes_H\sigma}(\vec u, \vec v)
&=&\sum_{\vec{u}_1,\vec{v}_1}
W_{\rho}(\vec{u}_1,2\vec{v}_1)
W_{\sigma}(\vec{u}-\vec{u}_1, 2(\vec{v}-\vec{v}_1))\\
&=&\sum_{\vec{u}_1,\vec{v}_1}
W_{\rho}(\vec{u}_1,2\vec{v}_1)
W_{\sigma_0}(\vec{u}_1-\vec u, 2(\vec{v}_1-\vec{v}))\\
&=&\sum_{\vec{u}_1,\vec{v}_1}
W_{\rho}(\vec{u}_1,2\vec{v}_1)
W_{w(\vec u, 2\vec v)\sigma_0w(\vec u, 2\vec v)^\dag}(\vec{u}_1, 2\vec{v}_1)\\
&=&\sum_{\vec{u}_1,\vec{v}_1}
W_{\rho}(\vec{u}_1,2\vec{v}_1)
W_{\sigma_{\vec u, \vec v}}(\vec{u}_1, 2\vec{v}_1),
\end{eqnarray*}
where $\sigma_0=T(0,0)\sigma T(0,0)$, 
and the third equality used the fact that 
$$W_{w(\vec u, \vec v)\rho w(\vec u, \vec v)^\dag}(\vec p, \vec q)=
W_{\rho}(\vec p-\vec u, \vec q-\vec v ) \;,$$
and the last equality comes from $\sigma_{\vec u,\vec v}=w(\vec u, 2\vec v)\sigma_0w(\vec u, 2\vec v)^\dag$.
\end{proof}

\begin{lem}
Let $\rho$ and $\sigma$ be two $n$-qudit states with $\CMM(\rho) = \CMM(\sigma)$. 
Then  we have 
\begin{eqnarray}
\CMM(\rho\boxtimes_H\sigma)=\CMM(\rho)\boxtimes_H\sigma
=\rho\boxtimes_H\CMM(\sigma)=\CMM(\rho)\boxtimes_H\CMM(\sigma).
\end{eqnarray} 
\end{lem}
\begin{proof}
Let $S$ be the abelian Weyl group associated with $\CMM(\rho) (= \CMM(\sigma))$,
then the characteristic function of $\CMM(\rho)$ is supported on  $S$.
Therefore
\begin{align*}
\Xi_{\CMM(\rho) \boxtimes_H \sigma } (\vec p, \vec q) 
= & \Xi_{\CMM(\rho)} (2^{-1}\vec p, \vec q)\; \Xi_{\sigma} (2^{-1}\vec p, \vec q)\\
=&  \Xi_{\rho} (2^{-1}\vec p, \vec q) \;\Xi_{ \sigma } (2^{-1}\vec p, \vec q) \delta_{(2^{-1}\vec p, \vec q)\in S}\\
=&  \Xi_{\rho\boxtimes_H \sigma} (\vec p, \vec q) \delta_{(2^{-1}\vec p, \vec q)\in S} \;.
\end{align*}
And for any $(\vec p, \vec q)$, $|\Xi_{\rho \boxtimes_H \sigma} (\vec p, \vec q)| = 1 $ if and only if $|\Xi_{\rho} (2^{-1}\vec p, \vec q)| = |\Xi_{\sigma} (2^{-1}\vec p, \vec q)| =1$, if and only if $(2^{-1}\vec p, \vec q)\in S$.
That is,
\begin{align*}
    \Xi_{\CMM(\rho\boxtimes_H \sigma)} = \Xi_{\rho\boxtimes_H \sigma} (\vec p, \vec q) \delta_{(2^{-1}\vec p, \vec q)\in S} = \Xi_{\CMM(\rho) \boxtimes_H \sigma} (\vec p, \vec q)= \Xi_{\rho \boxtimes_H \CMM(\sigma)} (\vec p, \vec q)\;,
\end{align*}
 for every $(\vec p, \vec q)$. 
Therefore $\CMM(\rho\boxtimes_H\sigma)=\CMM(\rho)\boxtimes_H\sigma
=\rho\boxtimes_H\CMM(\sigma)$. Moreover,   
\begin{align*}
\Xi_{\CMM(\rho) \boxtimes_H \CMM(\sigma) } (\vec p, \vec q) = \Xi_{\rho} (2^{-1}\vec p, \vec q) \Xi_{ \sigma } (2^{-1}\vec p, \vec q) \delta_{(2^{-1}\vec p, \vec q)\in S} \;.
\end{align*}
Hence $\CMM(\rho\boxtimes_H\sigma)=\CMM(\rho)\boxtimes_H\CMM(\sigma)$, and the result holds.
\end{proof}

\begin{lem}[\bf Commutativity with Clifford unitaries]
For any Clifford unitary $U$, there exists a Clifford unitary $U_1$ such that
\begin{eqnarray}\label{230520shi1}
U_1(\rho\boxtimes_{H}\sigma) U_1^\dag
=(U\rho U^\dag) \boxtimes_H(U\sigma U^\dag) \;,\quad \forall \rho,\sigma\;.
\end{eqnarray}
\end{lem}
\begin{proof}
We may assume $U$ satisfies 
\[
U^\dag =\mu(M) w(\vec p_0, \vec q_0)\;,
\]
for some symplectic matrix $M$ and $(\vec p_0, \vec q_0)\in V^n$,
where $\mu(M)$ is a unitary satisfying
\[\mu(M) w(\vec p,\vec q) \mu(M)^\dag
=w(M(\vec p,\vec q)) \;,\quad \forall (\vec p,\vec q)\in V^n.\]
Hence for every $(\vec p, \vec q ) \in V^n $,
\begin{eqnarray*}
\mathcal{E}^\dag_H(w(\vec p,\vec q))
=w(2^{-1}\vec p,\vec q)\ot w(2^{-1}\vec p, \vec q) \;,\quad\forall (\vec p, \vec q ) \in V^n \;,
\end{eqnarray*}
and
\begin{eqnarray*}
&&(U^\dag\otimes U^\dag) \mathcal{E}^\dag_H(w(\vec p,\vec q)) (U\otimes U) \\
&=& (U^\dag\otimes U^\dag) w(2^{-1}\vec p,\vec q)\ot w(2^{-1}\vec p, \vec q) (U\otimes U)\\
&=&  \chi( 2\vec p_0 \cdot \vec q - \vec p \cdot \vec q_0 )  w(M(2^{-1}\vec p, \vec q) )  \ot   w(M(2^{-1}\vec p, \vec q) ) \;.
\end{eqnarray*}

Let $U_1$ be the Clifford unitary satisfying 
\[U_1^\dag = \mu(M_1)w(2\vec p_0, \vec q_0) \;,\]
where 
\begin{eqnarray*}
M_1  =
\left(
\begin{matrix}
  2I_n  &  0 \\
  0 &  I_n
\end{matrix}\right)  M\left(
\begin{matrix}
 2^{-1}I_n &  0 \\
  0 &  I_n
\end{matrix}\right) .
\end{eqnarray*}
Then
\begin{eqnarray*}
&&\mathcal{E}^\dag_H( U_1^\dag w(\vec p,\vec q) U_1) \\
&=&   \mathcal{E}^\dag_H (\chi( 2\vec p_0 \cdot \vec q - \vec p \cdot \vec q_0 ) w(M_1 (\vec p,\vec q)))\\
&=& \chi( 2\vec p_0 \cdot \vec q - \vec p \cdot \vec q_0 )   w\left(\left(
\begin{matrix}
 2^{-1}I_n &  0 \\
  0 &  I_n
\end{matrix}\right) M_1 (\vec p,\vec q)\right) \otimes w\left(\left(
\begin{matrix}
 2^{-1}I_n &  0 \\
  0 &  I_n
\end{matrix}\right) M_1 (\vec p,\vec q)\right)\\
&=&\chi( 2\vec p_0 \cdot \vec q - \vec p \cdot \vec q_0 ) w(M(2^{-1}\vec p, \vec q) )  \ot   w(M(2^{-1}\vec p, \vec q) )\\
&=& (U^\dag\otimes U^\dag) \mathcal{E}^\dag_H(w(\vec p,\vec q)) (U\otimes U).
\end{eqnarray*}
Thus \eqref{230520shi1} holds.
\end{proof}

\subsubsection{Discrete Beam Splitter Convolution on Qudits}

Our second example  is the convolution whose  parameter matrix $G$ is 
\begin{equation}\label{0204shi5}
G=\left[
\begin{array}{cc}
s&t\\
t& -s
\end{array}
\right]\;,
\end{equation}
with $s^2+t^2\equiv 1\mod d$.
This  is a discrete version of the condition 
$(\sqrt{\lambda})^2+(\sqrt{1-\lambda})^2=1$ for rotation that occurs in CV  beam splitter.  In fact, 
the condition $s^2+t^2\equiv 1 \mod d$ can be satisfied for any prime number $d\ge 7$ (See Appendix \ref {sec:apen_numT}.)

\begin{Def}[\bf Discrete beam splitter]\label{Def:disc_BS}
Given $s^2+t^2\equiv 1 \mod d$, 
the key unitary  $U_{s,t}$ is
\begin{align}\label{1231shi1}
 U_{s,t} = \sum_{\vec i,\vec j\in \mathbb{Z}^n_d} |s\vec i+t\vec j \rangle \langle \vec i| \otimes | t\vec i-s\vec j\rangle \langle \vec j|\; ,
 \end{align}
 where the state $| \vec i \rangle = |  i_1 \rangle \otimes \cdots \otimes |  i_n \rangle \in \CHH^{\otimes n} $.
The convolution  of two $n$-qudit states $\rho$ and $\sigma$ is
\begin{align}\label{eq:conv_B}
\rho \boxtimes_{s,t} \sigma = \Ptr{B}{ U_{s,t} (\rho \otimes \sigma) U^\dag_{s,t}}\;.
\end{align}
\end{Def}

\begin{prop}
Given two $n$-qudit states $\rho$ and $\sigma$, 
the characteristic function satisfies
\begin{align*}
\Xi_{ \rho \boxtimes_{s,t} \sigma} (\vec p, \vec  q) = \Xi_\rho (s\vec p, s\vec q) \;\Xi_\sigma (t\vec p ,t\vec q) \;, \quad \forall (\vec{p}, \vec q)\in V^n \;.
\end{align*}
\end{prop}

\begin{prop}\label{prop:wign_bs}
Given two  $n$-qudit states $\rho$ and $\sigma$, 
the discrete Wigner function of  $\rho\boxtimes_{s,t}\sigma$ 
satisfies
\begin{align*}
W_{\rho\boxtimes_{s,t}\sigma}(\vec{u}, \vec v)=\sum_{\vec{p},\vec q}W_{\rho}(\vec{p}, \vec q)  W_{\sigma_{\vec{u},\vec v, t}}(t^{-1}s\vec{p}, t^{-1}s\vec q) \;,
\end{align*}
where  $\sigma_{\vec{u},\vec v,t}=w(t^{-1}\vec{u}, t^{-1}\vec v)T(\vec{0},\vec 0)\sigma T(\vec{0}, \vec 0)w(t^{-1}\vec{u}, t^{-1}\vec v)^\dag$.
\end{prop}
\begin{proof}
By Proposition \ref{prop:wig_conv},
\begin{align*}
W_{\rho\boxtimes_{s,t}\sigma}(\vec{u}, \vec v)
=&\sum_{\vec{u}_1,\vec v_1} W_{\rho}(s^{-1}\vec{u}_1, s^{-1}\vec v_1) W_{\sigma }(t^{-1} (\vec u-\vec u_1), t^{-1} (\vec v-\vec v_1))\\
=&\sum_{\vec{u}_1,\vec v_1} W_{\rho}(s^{-1}\vec{u}_1, s^{-1}\vec v_1) W_{\sigma_0 }(t^{-1} \vec u_1-t^{-1}\vec u, t^{-1} \vec v_1-t^{-1}\vec v)\\
=&\sum_{\vec{u}_1,\vec v_1} W_{\rho}(s^{-1}\vec{u}_1, s^{-1}\vec v_1) W_{w(t^{-1}\vec{u}, t^{-1}\vec v)\sigma_0 w(t^{-1}\vec{u}, t^{-1}\vec v)^\dag}(t^{-1} \vec u_1, t^{-1} \vec v_1)\\
=&\sum_{\vec{u}_1,\vec v_1} W_{\rho}(\vec{u}_1, \vec v_1) W_{\sigma_{\vec u,\vec v,t}}(t^{-1} s\vec u_1, t^{-1} s\vec v_1) \;,
\end{align*}
where $\sigma_0=T(0,0)\sigma T(0,0)$, 
and the third equality used the fact that 
$$W_{w(\vec u, \vec v)\rho w(\vec u, \vec v)^\dag}(\vec p, \vec q)=
W_{\rho}(\vec p-\vec u, \vec q-\vec v ) \;,$$
and the last equality comes from $\sigma_{\vec u,\vec v,t}=w(t^{-1}\vec{u}, t^{-1}\vec v)\sigma_0 w(t^{-1}\vec{u}, t^{-1}\vec v)^\dag$.
\end{proof}

\begin{Rem}[\bf Wigner function positivity]
If $s\equiv t\mod d$ for the beam splitter, then $W_{\rho\boxtimes_{s,t}\sigma}$ is always nonnegative for any input states $\rho$ and $\sigma$.
This is because
\begin{align}
 W_{\rho\boxtimes_{s,t}\sigma}(\vec{u}, \vec v)=\sum_{\vec{p},\vec q}W_{\rho}(\vec{p}, \vec q)  W_{\sigma_{\vec{u},\vec v, t}}(\vec{p},\vec q)
 =\frac{1}{d^n}\Tr{\rho \sigma_{\vec{u},\vec v, t}}\geq 0.
\end{align}

\end{Rem}

\begin{lem}\label{lem:com_Mcon}
Let $\rho$ and $\sigma$ be two $n$-qudit states with $\CMM(\rho) = \CMM(\sigma)$. 
Then  we have 
\begin{eqnarray}\label{0204shi6}
\CMM(\rho\boxtimes_{s,t}\sigma)=\CMM(\rho)\boxtimes_{s,t}\sigma
=\rho\boxtimes_{s,t}\CMM(\sigma)=\CMM(\rho)\boxtimes_{s,t}\CMM(\sigma).
\end{eqnarray} 
\end{lem}
\begin{proof}
Let $S$ be the abelian Weyl group associated with $\CMM(\rho) = \CMM(\sigma)$,
then the characteristic function of $\CMM(\rho)$ is supported on  $S$.
Therefore
\begin{align*}
\Xi_{\CMM(\rho) \boxtimes_{s,t} \sigma } (\vec p, \vec q) 
= & \Xi_{\CMM(\rho)} (s\vec p, s\vec q)\; \Xi_{\sigma} (t\vec p, t\vec q)\\
=&  \Xi_{\rho} (s\vec p, s\vec q) \;\Xi_{ \sigma } (t\vec p, t\vec q) \delta_{(\vec p, \vec q)\in S}\\
=&  \Xi_{\rho\boxtimes_{s,t} \sigma} (\vec p, \vec q) \delta_{(\vec p, \vec q)\in S} \;.
\end{align*}
And for any $(\vec p, \vec q)$, $|\Xi_{\rho \boxtimes_{s,t} \sigma} (\vec p, \vec q)| = 1 $ if and only if $|\Xi_{\rho} (s\vec p, s\vec q)| = |\Xi_{\sigma} (t\vec p, t\vec q)| =1$, if and only if $(\vec p, \vec q)\in S$.
That is,
$\Xi_{\CMM(\rho\boxtimes_{s,t} \sigma)} = \Xi_{\rho\boxtimes_{s,t} \sigma} (\vec p, \vec q) \delta_{(\vec p, \vec q)\in S} = \Xi_{\CMM(\rho) \boxtimes_{s,t} \sigma} (\vec p, \vec q)= \Xi_{\rho \boxtimes_{s,t} \CMM(\sigma)} (\vec p, \vec q)$ for every $(\vec p, \vec q)$. 
Therefore $\CMM(\rho\boxtimes_{s,t}\sigma)=\CMM(\rho)\boxtimes_{s,t}\sigma
=\rho\boxtimes_{s,t}\CMM(\sigma)$. Moreover, we also have 
\begin{align*}
\Xi_{\CMM(\rho) \boxtimes_{s,t} \CMM(\sigma) } (\vec p, \vec q) = \Xi_{\rho} (s\vec p, s\vec q) \Xi_{ \sigma } (t\vec p, t\vec q) \delta_{(\vec p, \vec q)\in S} \;.
\end{align*}
Hence $\CMM(\rho\boxtimes_{s,t}\sigma)=\CMM(\rho)\boxtimes_{s,t}\CMM(\sigma)$, and \eqref{0204shi6} holds.
\end{proof}

\begin{lem}[\bf Commutativity with Clifford unitaries]
For any Clifford unitary $U$, there exists a Clifford unitary $U_1$ such that
\begin{eqnarray*}
U_1(\rho\boxtimes_{s,t}\sigma) U_1^\dag
=(U\rho U^\dag) \boxtimes_{s,t}(U\sigma U^\dag) \;.
\end{eqnarray*}
\end{lem}
\begin{proof}
First, based on Theorem 3 in \cite{{Gross06}}, 
Clifford unitaries have the following properties:

(1) For  any symplectic matrix $M$, there is a unitary operator $U(M)$
such that 
\begin{eqnarray*}
U(M)w(\vec v) U(M)^\dag
=w(M\vec v) \;.
\end{eqnarray*}

(2)  Up to a phase, any Clifford operation is of the form
\begin{eqnarray*}
U=w(\vec p,\vec q)U(M)\;,
\end{eqnarray*}
for some $(\vec p,\vec q) \in V^n$.
Without loss of generality, we may assume 
$U=w(\vec p,\vec q) U(M)$, then $U_1=w((s+t )(\vec p,\vec q))U(M)$
satisfies the equality.
\end{proof}

\begingroup
\setlength{\tabcolsep}{6pt} 
\renewcommand{\arraystretch}{1.5} 

\begin{table}[htbp]
\centering
 \resizebox{\textwidth}{20mm}{
\begin{tabular}{ |c|c|c|c| } 
\hline
Beam splitter& CV quantum systems & DV quantum systems\\
\hline
Parameter& $(\sqrt{\lambda},\sqrt{1-\lambda}),\;\lambda\in[0,1]$ & $(s,t)$,\; $s^2+t^2\equiv1 \mod d$ \\
\hline
\multirow{2}{*}{Convolution}& 
$\rho\boxtimes_{\lambda}\sigma=\Ptr{B}{U_{\lambda}\rho\ot\sigma U^\dag_{\lambda}}$, & $\rho\boxtimes_{s,t}\sigma=\Ptr{B}{U_{s,t}\rho\ot\sigma U^\dag_{s,t}}$,\\
& $U_{\lambda}$: beam splitter  & $ U_{s,t}$: discrete beam splitter  \\
\hline
{Characteristic function } &  $\Xi_{\rho\boxtimes_{\lambda}\sigma}(\vec{x})=\Xi_{\rho}(\sqrt{\lambda}\vec{x})\;\Xi_{\sigma}(\sqrt{1-\lambda}\vec{x})$&$\Xi_{\rho\boxtimes_{s,t}\sigma}(\vec{x})=\Xi_{\rho}(s\vec{x})\;\Xi_{\sigma}(t\vec{x})$\\
\hline
Wigner function  & $W_{\rho\boxtimes_{\lambda}\sigma}(\vec{b})=\int d\vec{x}\;W_{\rho}\left(\vec{x}\right)W_{\sigma_{\vec{b},\lambda}}\left(\frac{\sqrt{\lambda}}{\sqrt{1-\lambda}}\vec{x}\right)$
&$W_{\rho\boxtimes_{s,t}\sigma}(\vec{b})=\sum_{\vec{x}}W_{\rho}(\vec{x})W_{\sigma_{\vec{b},t}}(t^{-1}s\vec{x})$\\
\hline
\multirow{2}{*}{Quantum entropy  power inequality}& $H(\rho\boxtimes_{\lambda}\sigma)\geq \lambda H(\rho)+(1-\lambda)H(\sigma)$ \cite{Konig14},&
$H_{\alpha}(\rho\boxtimes_{s,t}\sigma)\geq \max\set{H_{\alpha}(\rho),H_{\alpha}(\sigma)}$,\\
&$e^{H(\rho\boxtimes_{\lambda}\sigma)/n}\geq \lambda e^{H(\rho)/n}+(1-\lambda)e^{H(\sigma)/n}$ \cite{Konig14,Palma14}
& $\alpha\in[-\infty,+\infty]$(Proposition \ref{prop:entropy})\\
\hline 
\multirow{2}{*}{Quantum Fisher information inequality}&
$w^2J(\rho\boxtimes_{\lambda}\sigma)\leq w^2_1J(\rho)+w^2_2J(\sigma)$,&
$J(\rho\boxtimes_{s,t}\sigma)\leq \min\set{J(\rho),J(\sigma)}$\\
& $w=\sqrt{\lambda}w_1+\sqrt{1-\lambda}w_2$ \cite{Konig14}&(Theorem \ref{thm:fisher})\\ 
\hline
\end{tabular}}
\vskip 5pt
\caption{\label{tab:sum_B}We compare our discrete beam splitter (displayed in the third column) with results for  CV quantum systems (in the second column.)}

\end{table}

\endgroup

\subsubsection{Discrete Amplifier Convolution on Qudits}

The third example of the convolution is the discrete amplifier.
Let  the  parameter matrix $G$ be
\begin{equation}\label{0205shi1}
G=\left[
\begin{array}{cc}
l&-m\\
-m& l
\end{array}
\right]\;,
\end{equation}
where $l^2-m^2\equiv 1\mod d$.
This  is a discrete version of the condition 
$(\sqrt{\kappa})^2-(\sqrt{\kappa-1})^2=1$ with $\kappa\in[1,\infty)$ that occurs in CV squeezing unitary.  
In fact, 
the condition $l^2-m^2\equiv 1\mod d$ can be satisfied for any prime number $d\ge 7$ (See Appendix \ref {sec:apen_numT}.)

\begin{Def}[\bf Discrete  amplifier]\label{Def:dis_sq}
Given $l^2-m^2\equiv 1 \mod d$, 
the unitary operator $V_{l,m}$ is
\begin{align}\label{eq:squeez}
 V_{l,m} = \sum_{\vec i,\vec j\in \mathbb{Z}^n_d} |l\vec i+m\vec j \rangle \langle \vec i| \otimes | m\vec i+l\vec j\rangle \langle \vec j| \;.
 \end{align}
 The convolution of two $n$-qudit states $\rho$ and $\sigma$ 
is 
\begin{align}\label{eq:conv_S}
\rho \boxtimes_{l,m} \sigma = \Ptr{B}{ V_{l,m} (\rho \otimes \sigma) V^\dag_{l,m}} \;.
\end{align}
\end{Def}

\begin{prop}
Given two $n$-qudit states $\rho$ and $\sigma$, 
 the characteristic function 
satisfies
\begin{align*}
\Xi_{ \rho \boxtimes_{l,m} \sigma} (\vec p,\vec q) = \Xi_\rho (l\vec p,l\vec{q}) \Xi_\sigma (m\vec p,-m\vec q)
=\Xi_\rho (l\vec p,l\vec{q}) \Xi_{\sigma^T} (m\vec p,m\vec q) \;, \quad \forall (\vec p,\vec q)\in V^n \;,
\end{align*}
where $ \sigma^T$ is the transpose of $\sigma$ as a matrix in the Pauli $Z$ basis.
\end{prop}

\begin{prop}\label{prop:wign_am}
Given two $n$-qudit states $\rho$ and $\sigma$,  the  discrete Wigner function 
satisfies
\begin{align*}
W_{\rho\boxtimes_{l,m}\sigma}(\vec{u}, \vec v)
=\sum_{\vec{p}, \vec q}W_{\rho}(\vec{p}, \vec q)\,\,W_{\sigma^T_{\vec{u}, \vec v, m}}(m^{-1}l\vec{p}, m^{-1}l\vec q) \;,
\end{align*}
where  $\sigma^T_{\vec{u},\vec v,m}=w(m^{-1}\vec{u}, m^{-1}\vec v)T(\vec{0},\vec 0)\sigma^T T(\vec{0}, \vec 0)w(m^{-1}\vec{u}, m^{-1}\vec v)^\dag$, 
and  $ \sigma^T$ is the transpose of $\sigma$ as the matrix in the Pauli Z basis.
\end{prop}

\begingroup
\setlength{\tabcolsep}{6pt} 
\renewcommand{\arraystretch}{1.5} 

\begin{table*}[!htbp]
\centering
\resizebox{\textwidth}{20mm}{
\begin{tabular}{ |c|c|c|c| } 
\hline
Amplifier & CV quantum systems & DV quantum systems\\
\hline
Parameter& $(\sqrt{\kappa},\sqrt{\kappa-1}),\kappa\in[1,\infty)$ & $(l,m)$, $l^2-m^2\equiv1 \mod d$ \\
\hline
\multirow{2}{*}{Convolution}& 
$\rho\boxtimes_{\kappa}\sigma=\Ptr{B}{V_{\kappa}\rho\ot\sigma V^\dag_{\kappa}}$, & $\rho\boxtimes_{l,m}\sigma=\Ptr{B}{V_{l,m}\rho\ot\sigma V^\dag_{l,m}}$,\\
& $V_{\kappa}$ : squeezing unitary & $V_{l,m}$ : discrete squeezing unitary \\
\hline
Characteristic function &  $\Xi_{\rho\boxtimes_{\kappa}\sigma}(\vec p, \vec q)=\Xi_{\rho}(\sqrt{\kappa}\vec p, \sqrt{\kappa} \vec q)\Xi_{\sigma}(\sqrt{\kappa-1}\vec p, -\sqrt{\kappa-1}\vec q)$&$\Xi_{\rho\boxtimes_{l,m}\sigma}(\vec{p},\vec q)=\Xi_{\rho}(l\vec{p}, l\vec q)\Xi_{\sigma}(m\vec p, -m\vec q)$\\
\hline
\multirow{2}{*}{Quantum entropy  power inequality}&$e^{H(\rho\boxtimes_{\kappa}\sigma)/n}\geq \kappa e^{H(\rho)/n}+(\kappa-1)e^{H(\sigma)/n}$  &
$H_{\alpha}(\rho\boxtimes_{l,m}\sigma)\geq \max\set{H_{\alpha}(\rho),H_{\alpha}(\sigma)}$,\\
& \cite{Palma14}
& $\alpha\in[-\infty,+\infty]$ (Proposition \ref{prop:entropy})\\
\hline 
\multirow{2}{*}{Quantum Fisher information inequality}&
$w^2J(\rho\boxtimes_{\kappa}\sigma)\leq w^2_1J(\rho)+w^2_2J(\sigma)$,&
$J(\rho\boxtimes_{l,m}\sigma)\leq \min\set{J(\rho),J(\sigma)}$\\
& $w=\sqrt{\kappa}w_1+\sqrt{\kappa-1}w_2$ \cite{Konig14}&(Theorem \ref{thm:fisher})\\ 
\hline
\end{tabular}}
\vskip 5pt
\caption{\label{tab:sum_S}We compare our results on the discrete amplifier (displayed in the third column) with results in  CV quantum systems (in the second column.)}
\end{table*}
\endgroup

\subsubsection{Convolution on qubits: CNOT}
Consider the convolution on qubits. There are only four choices of 
nontrivial invertible parameter matrices $G$ as follows 
\begin{equation*}
G=\left[
\begin{array}{cc}
1&0\\
1& 1
\end{array}
\right]\;,\quad
\left[
\begin{array}{cc}
1&1\\
0& 1
\end{array}
\right]\;,\quad
\left[
\begin{array}{cc}
0&1\\
1& 1
\end{array}
\right]\;,\quad
\left[
\begin{array}{cc}
1&1\\
1& 0
\end{array}
\right].
\end{equation*}
Following Definition \ref{def:Cli_unitary},
the corresponding unitaries on the 2-qubit system are 
\begin{equation*}
CNOT_{2\to 1} \;,\quad CNOT_{1\to 2} \;,\quad  SWAP \cdot CNOT_{1\to 2} \;,\quad  SWAP \cdot CNOT_{2\to 1}\;.
\end{equation*}
Here the various CNOT gates are 
\[
CNOT_{2\to 1}=\sum_{i,j}\ket{i+j}\bra{i}\ot\ket{j}\bra{j}\;,
\quad
CNOT_{1\to 2}=\sum_{i,j}\ket{i}\bra{i}\ot\ket{i+j}\bra{j}\;,
\]
and
\[
SWAP=\sum_{i,j}\ket{j}\bra{i}\ot \ket{i}\bra{j}\;.
\]
The parameter matrices are even-parity positive for the first two choices, and odd-parity positive for the last two choices. 
 Hence, we also have the corresponding partial quantum entropy power inequality and quantum Fisher information inequality for these 4 choices of $G$ in qubit systems.
For example, for $CNOT_{B\to A}$, by Proposition \ref{prop:entropy} we have  
$H_{\alpha}(\rho\boxtimes\sigma)\geq H_{\alpha}(\rho)$.

\section{Quantum central limit theorem for states}\label{sec:CLL}

In this section, we use $\boxtimes$ to abbreviate the beam splitter convolution 
$\boxtimes_{s,t}$,
and for any state $\rho$, define  $\boxtimes^{N+1}\rho=(\boxtimes^N\rho)\boxtimes\rho$ inductively, where $\boxtimes^0\rho=\rho$.
Before considering the quantum central limit theorem, let us review the classical case. 
Let $X$ be a random variable with probability density function $f$. 
The central limit theorem states that, if $X$ has zero mean, 
then  $\frac{1}{\sqrt N}X_1+\cdots + \frac{1}{\sqrt N}X_N$ converges to a normal random variable. That is, the  probability density function $ *_Nf $ converges to a normal distribution  as $N\rightarrow \infty$,
where $*_Nf$ denotes the balanced  $N$-fold convolution of $f$. 
Here the condition that $X$ has zero mean cannot be removed.
For example,
if  $X\sim \CNN( 1,1)$,
$\frac{1}{\sqrt N}X_1+\cdots + \frac{1}{\sqrt N}X_N\sim \CNN(\sqrt N, 1)$ and it does not converge.
Hence, given a random variable $X$, we should consider the zero-mean variable $X-\mathbb{E} X$ instead of $X$, where $\mathbb{E}X$ is the mean value of $X$.

The quantum analogue of this centering was introduced in Definition \ref{Def:Zero_mean}.
A state $\rho$ has zero mean precisely when the
characteristic function of $\mathcal M(\rho)$ takes values in $\{0,1\}$,
and by Lemma~\ref{lem:zero_mean_pau} every state has zero mean after conjugation by a
suitable Weyl operator. Since conjugation by a Weyl operator is a
discrete phase-space displacement, this assumption involves no loss of
generality, and we impose it throughout this section.

\begin{thm}[\bf Central limit theorem for states and the magic gap]\label{thm:CLT_gap} 
Let $\rho$ be a zero-mean  $n$-qudit state with the MS $\CMM(\rho)$ and  magic gap $MG(\rho)$. Then  
\begin{eqnarray*}
\norm{\boxtimes^N\rho-\CMM(\rho)}_2
\leq (1-MG(\rho))^{N}\norm{\rho-\CMM(\rho)}_2 \;.
\end{eqnarray*}
If $\rho\neq \CMM(\rho)$, then  $MG(\rho)>0$, and the  convergence is exponentially fast with respect to the time of
convolution.
\end{thm}
\begin{proof}
Let $S$ be the abelian subgroup associated with $\CMM(\rho)$.
Since $\rho$ is zero-mean, we have
\begin{align*}
\Xi_{\CMM(\rho)}(\vec p, \vec q) = \left\{
\begin{aligned}
1&& (\vec p, \vec q)\in S \;,\\
0&& (\vec p, \vec q)\not\in S \;.
\end{aligned}\right.
\end{align*} 
By Lemma \ref{lem:com_Mcon}, it can be proved inductively that  
\begin{eqnarray*}
\CMM(\boxtimes^N\rho) 
=\CMM(\rho) \;.
\end{eqnarray*}
Thus,
\begin{eqnarray*}
\boxtimes^N\rho
-\CMM(\rho)
=\frac{1}{d^n}
\sum_{(\vec p, \vec q)\notin S}
\Xi_{\boxtimes^N\rho}(\vec p, \vec q)
w(\vec p, \vec q) \;.
\end{eqnarray*}
Moreover,  for any $(\vec p, \vec q) \notin S$,
\begin{eqnarray*}
|\Xi_{\boxtimes^N\rho}(\vec p, \vec q)|
=|\Xi_{\boxtimes^{N-1}\rho}(s\vec p, s\vec q)
||\Xi_{\rho}(t\vec p, t\vec q)|
\leq (1-MG(\rho))^{N}
|\Xi_{\rho}(t\vec p, t\vec q)| \;.
\end{eqnarray*}
Therefore, 
\begin{align*}
\norm{\boxtimes^N\rho-\CMM(\rho)}_2^2
=&\frac{1}{d^n}
\sum_{(\vec p, \vec q)\notin S}
|\Xi_{\boxtimes^N\rho}(\vec p, \vec q)|^2
\leq \frac{1}{d^n}(1-MG(\rho))^{2N}
\sum_{(\vec p, \vec q)\notin S}|\Xi_{\rho}(t\vec p, t\vec q)|^2\\
\leq& (1-MG(\rho))^{2N}\norm{\rho-\CMM(\rho)}^2_2 \;.
\end{align*}
\end{proof}

The above theorem shows exponential decay with respect to the number of repeated convolutions;  the exponential rate of convergence is controlled by the magic gap. This is quite different from the classical central limit theorem, where the rate of convergence is on the order of $O(1/\sqrt{N})$.

\begin{Rem}
All the results in this section rely on the definition of the $N$-fold convolution $\boxtimes^N \rho$.
Although we assumed that the parameter matrix $G$ for each convolution is the same in the definition of $\boxtimes^N \rho$, all conclusions in this section can relax this assumption. 
Specifically, given a sequence of pairs of parameters $(s_1,t_1), (s_2,t_2)$,... with $s_i^2+t_i^2\equiv 1 \mod d$ for every $i$, 
then we can define $\boxtimes^{N+1}\rho=(\boxtimes^N\rho)\boxtimes_{s_N, t_N}\rho$ inductively, where $\boxtimes^0\rho=\rho$.
Theorem \ref{0212prop2} and Theorem \ref{thm:CLT_gap} also hold with this newly defined multiple convolution.
\end{Rem}

\section{A framework for convolution of quantum channels}
\label{sect:ConvolutionChannels}
In this section, we focus on the convolutions of $n$-qudit channels, i.e., the quantum channels
acting on $n$-qudit systems. To study the convolution of quantum channels, we will use the Choi-Jamiołkowski isomorphism \cite{Choi75,Jamio72}.
By the Choi-Jamiołkowski isomorphism, any  quantum
channel $\Lambda$ from $\mathcal{H}_{A}$ to $\mathcal{H}_{A'}$ can be represented by its Choi state 
\begin{eqnarray*}
J_{\Lambda}=id_{A}\ot \Lambda (\proj{\Phi}) \;,
\end{eqnarray*}
where $| \Phi \rangle = \frac{1}{\sqrt{d^n}} \sum_{\vec j \in \Z_d^n}\ket{\vec j}_{A}\ot\ket{\vec j}_{A'} $.
For any input state $\rho$, the output state of the quantum channel $\Lambda(\rho)$ can be represented via the Choi state $J_{\Lambda}$ as 
\begin{eqnarray}\label{0127shi2}
\Lambda(\rho)
=d^n\Ptr{A}{J_{\Lambda} ( \rho^T_{A}\ot I_{A'})} \;.
\end{eqnarray}
On the other hand, for any operator $J$ on $\mathcal{H}_A\ot\mathcal{H}_{A'}$, 
the map 
\begin{eqnarray*}
\rho\to
d^n\Ptr{A}{J ( \rho^T_{A}\ot I_{A'})},
\end{eqnarray*}
is (1) completely positive if and only if $J$ is positive, (2) trace-preserving 
if and only if $\Ptr{A'}{J}=I_n/d^n$.

\begin{lem}[\bf Convolution of Choi states is Choi]
Given a nontrivial parameter matrix $G$, and  two quantum states $\rho_{AA'},
\sigma_{AA'}$ on $\mathcal{H}_A\ot\mathcal{H}_{A'}$ with 
$\Ptr{A'}{\rho_{AA'}} =\Ptr{A'}{\sigma_{AA'} } =I_A/d^n$. 
Then
$\rho_{AA'}\boxtimes\sigma_{AA'}$ also satisfies $\Ptr{A'}{\rho_{AA'}\boxtimes\sigma_{AA'}}=I_A/d^n$.
\end{lem}
\begin{proof}
By Proposition \ref{prop:chara_conv}, we have
\begin{align*}
&\Xi_{\rho\boxtimes\sigma}(\vec{p}_A,\vec{p}_{A'}, \vec{q}_A, \vec{q}_{A'})\\
&\qquad=\Xi_{\rho}(Ng_{11}\vec{p}_A,Ng_{11}\vec{p}_{A'},g_{00} \vec{q}_A, g_{00}\vec{q}_{A'})\\
&\qquad\qquad \times\Xi_{\sigma}(-Ng_{10}\vec{p}_A,-Ng_{10}\vec{p}_{A'}, g_{01}\vec{q}_A, g_{01}\vec{q}_{A'})\;.
\end{align*}
Since 
$\Ptr{A'}{\rho_{AA'}} =\Ptr{A'}{\sigma_{AA'}} =I_A/d^n$, we have 
$$\Xi_{\rho}(\vec{p}_A,\vec{0}_{A'}, \vec{q}_A, \vec{0}_{A'})=0, \quad\Xi_{\sigma}(\vec{p}_A,\vec{0}_{A'}, \vec{q}_A, \vec{0}_{A'})=0, $$
for any $(\vec{p}_A,\vec{q}_{A})\neq (\vec 0, \vec 0)$.
The parameter matrix $G$ is nontrivial, which means at most one of $g_{11}, g_{10}, g_{01}, g_{00}$ is  $0\mod d$. 
Then 
\begin{align*}
\Xi_{\rho\boxtimes\sigma}(\vec{p}_A, \vec{0}_{A'}, \vec{q}_A, \vec{0}_{A'})
&=\Xi_{\rho}(Ng_{11}\vec{p}_A, \vec{0}_{A'},g_{00} \vec{q}_A, \vec{0}_{A'})\\ 
&\qquad \times
\Xi_{\sigma}(-Ng_{10}\vec{p}_A, \vec{0}_{A'}, g_{01}\vec{q}_A, \vec{0}_{A'})
=0,
\end{align*}
for any $(\vec{p}_A,\vec{q}_{A})\neq (\vec 0, \vec 0)$.
Therefore $$\Ptr{A'}{\rho_{AA'}\boxtimes\sigma_{AA'}} = \frac{1}{d^n} \sum_{\vec{p}_A, \vec{q}_A} \Xi_{\rho\boxtimes\sigma}(\vec{p}_A, \vec{0}_{A'}, \vec{q}_A, \vec{0}_{A'}) w( \vec p_A, \vec q_A) =I_A/d^n, $$
and the proof is complete.
\end{proof}

\subsection{Definition of convolution of channels}

\begin{Def}[\bf Convolution of channels]\label{def:con_chan}
Given two $n$-qudit channels $\Lambda_1$ and $\Lambda_2$, the convolution $\Lambda_1\boxtimes \Lambda_2$ 
is the quantum channel  with the Choi state
\begin{align*}
J_{ \Lambda_1\boxtimes \Lambda_2} := J_{\Lambda_1} \boxtimes J_{\Lambda_2}\;,
\end{align*}
where the state $ J_{\Lambda_1} \boxtimes J_{\Lambda_2}$ is the convolution of the Choi states  $J_{\Lambda_1}$ and $J_{\Lambda_2}$.
\end{Def}

In the definition above, we use the convolution of Choi states to induce the convolution of quantum channels. As this does not provide a direct formula for the convolution of quantum channels, we now derive such an expression.
Denote the right inverse of the channel $\mathcal{E}$ to be
\begin{eqnarray}\label{eq:inver_conv_chn}
\mathcal{E}^{-1}(\rho)
=U^\dag \left(\rho\ot \frac{I_n}{d^n}\right) U \;,
\end{eqnarray}
which satisfies that $\mathcal{E}\circ\mathcal{E}^{-1}=id$. In addition, we observe that
$\mathcal{E}^{-1}=\frac{1}{d^n}\mathcal{E}^\dag$.

\begin{thm}[\bf Convolution of quantum channels]
\label{thm:exact_conv_chan}
Given two $n$-qudit channels $\Lambda_1, \Lambda_2$, their convolution $\Lambda_1\boxtimes \Lambda_2$ 
is 
\begin{eqnarray}\label{eq:exp_box_chan}
\Lambda_1\boxtimes\Lambda_2(\cdot)
=\mathcal{E}_G\circ (\Lambda_1\ot\Lambda_2)\circ\mathcal{E}^{-1}_G(\cdot) \;,
\end{eqnarray}
where $\mathcal{E}_G$ is the convolutional channel  in 
\eqref{eq:convo_chn}, and  $\mathcal{E}^{-1}_G$ is the inverse of the convolutional channel in \eqref{eq:inver_conv_chn}.
\end{thm}
\begin{proof}
First, the Choi state can be rewritten in terms of Weyl operators as 
\begin{eqnarray}\label{0127shi3}
J_\Lambda = \frac{1}{d^{2n}} \sum_{(\vec p,\vec q)\in V^n  } w(\vec p,\vec q) \otimes \Lambda(w(-\vec p,\vec q )) \;,
\end{eqnarray}
and thus 
\begin{eqnarray*}
J_{\Lambda_1}\boxtimes J_{\Lambda_2}
=\frac{1}{d^{3n}}\sum_{\vec p,\vec q}
w(\vec{p},\vec{q})
\ot (
\Lambda_1(w(-Ng_{11}\vec p, g_{00}\vec q))
\boxtimes \Lambda_2(w(Ng_{10}\vec p, g_{01}\vec q ))
) \;.
\end{eqnarray*}
For any $n$-qudit state $\rho=\frac{1}{d^n}\sum_{\vec{p},\vec{q}}\Xi_{\rho}(\vec p, \vec q)w(\vec p, \vec q)$,
\begin{eqnarray*}
\Lambda_1\boxtimes\Lambda_2(\rho)
&=&d^n\Ptr{A}{J_{\Lambda_1}\boxtimes J_{\Lambda_2} \cdot \rho^T\ot I}\\
&=&  
\sum_{\vec{p},\vec{q}}\Xi_{\rho}(-\vec p, \vec q)
\Ptr{A}{J_{\Lambda_1}\boxtimes J_{\Lambda_2}w(-\vec p, -\vec q)\ot I}\\
&=&\frac{1}{d^{2n}}
\sum_{\vec{p},\vec{q}}\Xi_{\rho}(-\vec p, \vec q)
\Lambda_1(w(-Ng_{11}\vec p, g_{00}\vec q))
\boxtimes \Lambda_2(w(Ng_{10}\vec p, g_{01}\vec q ))\\
&=&\frac{1}{d^{2n}}
\sum_{\vec{p},\vec{q}}\Xi_{\rho}(\vec p, \vec q) \Lambda_1(w(Ng_{11}\vec p, g_{00}\vec q))
\boxtimes \Lambda_2(w(-Ng_{10}\vec p, g_{01}\vec q )) \;.
\end{eqnarray*}
Moreover, 
\begin{eqnarray*}
&&\mathcal{E}_G\circ (\Lambda_1\ot\Lambda_2) \circ\mathcal{E}^{-1}_G(\rho)\\
&=&\frac{1}{d^{2n}}\sum_{\vec{p},\vec{q}}\Xi_{\rho}(\vec p, \vec q)
\mathcal{E}_G\circ (\Lambda_{1}\ot\Lambda_2)\circ\mathcal{E}^\dag_G
(w(\vec p, \vec q))\\
&=&\frac{1}{d^{2n}}\sum_{\vec{p},\vec{q}}\Xi_{\rho}(\vec p, \vec q)
\mathcal{E}_G\circ (\Lambda_1\ot\Lambda_2) (w(Ng_{11}\vec p, g_{00}\vec q)\ot w(-Ng_{10}\vec p, g_{01}\vec q))\\
&=&\frac{1}{d^{2n}}\sum_{\vec{p},\vec{q}}\Xi_{\rho}(\vec p, \vec q)
\mathcal{E}_G(\Lambda_1(w(Ng_{11}\vec p, g_{00}\vec q))\ot \Lambda_2( w(-Ng_{10}\vec p, g_{01}\vec q)))\\
&=&\frac{1}{d^{2n}}\sum_{\vec{p},\vec{q}}\Xi_{\rho}(\vec p, \vec q)
\Lambda_1(w(Ng_{11}\vec p, g_{00}\vec q))
\boxtimes \Lambda_2(w(-Ng_{10}\vec p, g_{01}\vec q )) \;.
\end{eqnarray*}
Therefore, \eqref{eq:exp_box_chan} holds for any quantum state $\rho$.
\end{proof}

\begin{figure}[t]
  \center{\includegraphics[width=12cm]{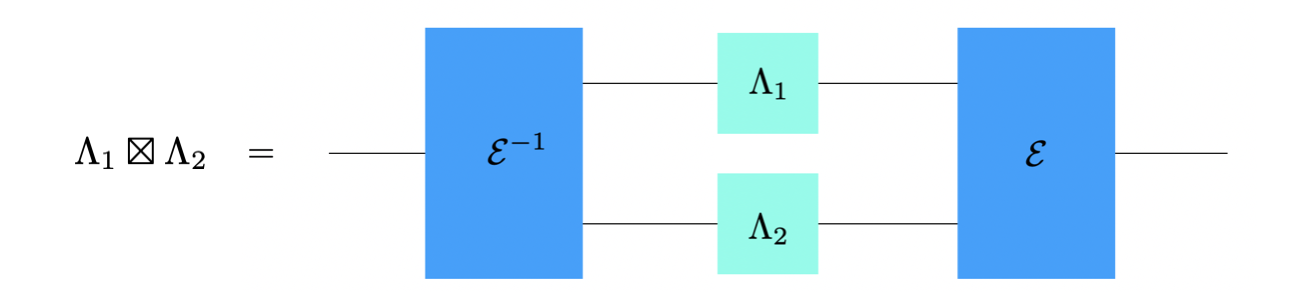} }    
  \caption{ The quantum circuit to realize the convolution of quantum channels.}
 \end{figure}

\begin{Def}[\bf Convolutional superchannel]
The convolutional superchannel $\Theta_G$ from $2n$-qudit to $n$-qudit is
\begin{eqnarray}
\Theta_G(\cdot)=\mathcal{E}_G\circ (\cdot)\circ \mathcal{E}^{-1}_G \;.
\end{eqnarray}
Hence $\Lambda_1\boxtimes \Lambda_2=\Theta_G(\Lambda_1\ot\Lambda_2)$ is the output channel of $\Theta_G$.
\end{Def}

The completely-depolarizing channel $\mathcal{R}$ on $n$-qudit systems is
\begin{eqnarray}\label{0212shi4}
\mathcal{R}(\rho)
=\Tr{\rho} \frac{I_n}{d^n} \;.
\end{eqnarray}
Then we have the following result.

\begin{prop}\label{prop:iden_pre}
Let  $\Lambda$ be an $n$-qudit channel.

(1) If $G$ is odd-parity positive, then 
\begin{eqnarray*}
\Lambda\boxtimes \mathcal{R}=\mathcal{R} \;.
\end{eqnarray*}

(2)  If $G$ is even-parity positive, then 
\begin{eqnarray*}
\mathcal{R}\boxtimes \Lambda=\mathcal{R} \;.
\end{eqnarray*}
\end{prop}
\begin{proof}
First, we prove the case where $G$ is odd-parity positive.
For any Weyl operator $w(\vec p,\vec q)$,
\begin{eqnarray*}
\Lambda\boxtimes \mathcal{R}(w(\vec p, \vec q))
&=&\mathcal{E}_G\circ \Lambda\ot \mathcal{R}\circ\mathcal{E}^{-1}_G(w(\vec p, \vec q))\\
&=&\frac{1}{d^n}\mathcal{E}_G\circ \Lambda\ot \mathcal{R}\circ\mathcal{E}^{\dag}_G(w(\vec p, \vec q))\\
&=&\frac{1}{d^n}
\mathcal{E}_G\circ \Lambda\ot \mathcal{R}(w(Ng_{11}\vec p, g_{00}\vec q)\ot w(-Ng_{10}\vec p, g_{01}\vec q))\\
&=&\mathcal{E}_G\left(\Lambda(w(Ng_{11}\vec p, g_{00}\vec q))
\ot \frac{I_n}{d^n}
\right)\delta_{\vec p,\vec 0}\delta_{\vec q, \vec 0}\\
&=&\Lambda(I_n)
\boxtimes \frac{I_n}{d^n}
\delta_{\vec p,\vec 0}\delta_{\vec q, \vec 0}\\
&=& I_n
\delta_{\vec p,\vec 0}\delta_{\vec q, \vec 0}\\
&=&\mathcal{R}(w(\vec p, \vec q)) \;,
\end{eqnarray*}
where the fourth equality comes from the fact that $G$ is odd-parity positive, and  the fifth equality comes from Lemma \ref{lem:conv_iden}.
Similar arguments also work for the case where $G$ is even-parity positive.
\end{proof}

\begin{prop}[\bf Convolutional stability for channels]
Given two stabilizer channels $\Lambda_1, \Lambda_2$, the convolution 
$\Lambda_1\boxtimes\Lambda_2$ is a stabilizer channel.
\end{prop}
\begin{proof}
Since  both $\mathcal{E}_G$ and $\mathcal{E}^{-1}_G$ are stabilizer channels, Theorem 
\ref{thm:exact_conv_chan} implies that  $\Lambda_1\boxtimes\Lambda_2=\mathcal{E}_G\circ \Lambda_1\ot \Lambda_2\circ\mathcal{E}^{-1}_G$
is a stabilizer channel for any stabilizer channels $\Lambda_1, \Lambda_2$.
\end{proof}

\subsection{Mean channel}
To study the quantum central limit theorem for channels, let us introduce the mean channel; this plays a similar role to the mean state.

\begin{lem}
Let $\rho_{AA'}$ be a quantum state with $\Ptr{A'}{\rho_{AA'}}= I_A/d^n$.
Then the MS $\CMM(\rho_{AA'})$ also satisfies $\Ptr{A'}{\CMM(\rho_{AA'})}= I_A/d^n$.
\end{lem}
\begin{proof}
The condition $\Ptr{A'}{\rho_{AA'}}= I_A/d^n$ 
means 
that 
\begin{eqnarray*}
\Xi_{\rho_{AA'}}(\vec p_A, \vec 0, \vec q_A, \vec 0) 
=0\;, \quad \forall (\vec{p}_A, \vec q_A)\neq (\vec 0, \vec 0)\;.
\end{eqnarray*}
Thus, by the definition of  $\CMM(\rho_{AA'})$,
\begin{eqnarray*}
\Xi_{\CMM(\rho_{AA'})}(\vec p_A, \vec 0, \vec q_A, \vec 0)  
=0\;,\quad \forall(\vec{p}_A, \vec q_A)\neq (\vec 0, \vec 0)\;,
\end{eqnarray*}
that is $\Ptr{A'}{\CMM(\rho_{AA'})}= I_A/d^n$. 
\end{proof}

Based on the above lemma, we have the following 
definition of the mean channel.
\begin{Def}[\bf Mean channel]
Given a quantum channel  $\Lambda$, 
the mean channel 
 $\CMM(\Lambda)$  is  the quantum channel with the 
 Choi state $J_{\CMM(\Lambda)} = \CMM(J_\Lambda)$, where  $\CMM(J_\Lambda)$
 is the MS of the Choi state $J_{\Lambda}$.
\end{Def}

\begin{prop}[\bf Mean channels are stabilizer channels]
For any quantum channel $\Lambda$, 
 $\mathcal{M}(\Lambda)$ is a stabilizer 
 channel.
\end{prop}

\begin{proof}
First, it is easy to verify that, for any $2n$-qudit pure stabilizer state $\proj{\psi}$ and any $n$-qudit pure stabilizer state $\phi$, the partial trace $\Ptr{A}{\proj{\psi} (\proj{\phi}\ot I)}$ can be written as a convex combination of pure stabilizer states. 
Since the MS $\mathcal{M}(J_{\Lambda}) $ can be written as a convex combination of pure stabilizer states,  
for any pure stabilizer state $\psi$,  the output state of the channel is
\begin{eqnarray*}
\mathcal{M}(\Lambda)(\proj{\psi})
=d^n\Ptr{A}{\mathcal{M}(J_{\Lambda})\proj{\psi}^T\ot I}
\end{eqnarray*}
can be written as  a convex combination of pure stabilizer states. Hence the mean channel $\mathcal{M}(\Lambda)$ is a stabilizer 
channel.
\end{proof}

We now consider the extremality of the mean channel. For this, we need to introduce the information measures for quantum 
channels. Specifically, we consider
the R\'enyi entropy of a quantum channel,  introduced in Ref. \cite{Gour21}. This can be used to quantify the capacity of merging   quantum channels.

\begin{Def}
Let $\Lambda$ be an $n$-qudit channel.
The R\'{e}nyi entropy of order $\alpha$ is
\begin{align*}
H_\alpha(\Lambda) = n\log d - D_\alpha(\Lambda\| \CRR) \;,
\end{align*}
where the channel $\CRR$ is the completely depolarizing channel \eqref{0212shi4}, and 
$$D_{\alpha}(\Lambda||\CRR):=\sup_{\rho_{AR}}D_{\alpha}(\Lambda\ot I_R(\rho_{AR})||\CRR\ot I_R(\rho_{AR}))\;,$$
with $\rho_{AR}$ running over bipartite states on 
$\mathcal{H}_A\ot \mathcal{H}_R$ for any ancilla system $R$.
\end{Def}

\begin{thm}\label{thm:exremality}
Let the local dimension $d$ be a prime number $\geq 7$, let $\alpha\in [1/2,+\infty]$, and let $\Lambda$ be an $n$-qudit channel. 
Then
\begin{eqnarray*}
H_{\alpha}(\CMM(\Lambda))
\geq H_{\alpha}( \Lambda).
\end{eqnarray*}
\end{thm}
\begin{proof}
The local dimension $d$ is a prime $\geq 7$, thus we can define the beam splitter convolution
$\boxtimes_{s,t} = \boxtimes$ with $s,t\neq 0 \mod d$.

(1) First  we show that, for every bipartite state $\rho$ on  $AA'$,  we can find a state $\sigma_{AA'}$  so that $\rho_{AA'} \boxtimes \sigma_{AA'}= \CMM(\rho_{AA'})$.
Moreover, if $\Ptr{A'}{\rho} = I_A/d^n$,
$\sigma$ can also be chosen to satisfy $\Ptr{A'}{\sigma} = I_A/d^n$.

Denote the abelian group associated with $\CMM(\rho)$ to be $S$,
and assume $S$ is generated by $ w(\vec p_1, \vec q_1),...,w(\vec p_r, \vec q_r)$.
Assume $\Xi_{\CMM(\rho)} (w(\vec p_i, \vec q_i))= \chi(k_i) $ with $k_i\in \Z_d$.
Let $ u_i = (1-s) t^{-1}k_i $ for every $i$,
and let $ \sigma$ be the  MSPS associated with $S$ such that
$\Xi_{\sigma}(\vec p_i, \vec q_i)= \chi(u_i) $ for every $i$.
Then for every $w(\vec p,\vec q)\in S$ we can write$ (\vec p, \vec q)  = \sum_i \lambda_i (\vec p_i, \vec q_i)  $ for some $\lambda_i \in \Z_d$, hence
\begin{align*}
\Xi_{\CMM(\rho)} ( \vec p, \vec q )= \chi(\sum_i \lambda_i  k_i) \;,
\end{align*}
and
\begin{align*}
\Xi_{\sigma} ( \vec p, \vec q  )= \chi(\sum_i \lambda_i  u_i) =  \chi((1-s) t^{-1} \sum_i \lambda_i  k_i) = \Xi_{\CMM(\rho)} ((1-s) t^{-1} \vec p,  (1-s)t^{-1} \vec q) \;.
\end{align*}
Therefore 
\begin{align*}
\Xi_{\rho \boxtimes \sigma} (\vec p, \vec q) =& \Xi_{ \rho} ( s\vec p, s\vec q) \Xi_{ \sigma} ( t\vec p, t\vec q)\\
=& \Xi_{ \rho} ( s\vec p, s\vec q) \Xi_{ \sigma} ( t\vec p, t\vec q) \delta_{w(\vec p, \vec q)\in S}\\
=& \Xi_{ \CMM(\rho)} ( s\vec p, s\vec q) \Xi_{ \sigma} ( t\vec p, t\vec q) \delta_{w(\vec p, \vec q)\in S}\\
=& \Xi_{ \CMM(\rho)} ( s\vec p, s\vec q) \Xi_{ \CMM(\rho)} ( (1-s)\vec p, (1-s)\vec q) \\
=& \Xi_{ \CMM(\rho)} (  \vec p, \vec q) \;.
\end{align*}
Thus $\rho \boxtimes \sigma= \CMM(\rho)$.
Moreover,
when $\Ptr{A'}{\rho} = I_A/d^n$,
we have $\Xi_{\rho}(\vec p_A, \vec p_{A'}, \vec q_A, \vec q_{A'}) = 0$  if $(\vec p_A, \vec q_A) \neq (\vec 0,\vec 0)$ and $ (\vec p_{A'},
\vec q_{A'})=(\vec 0, \vec 0)$.
Therefore $(\vec p_A, \vec p_{A'}, \vec q_A, \vec q_{A'})\not\in S $ when $(\vec p_A, \vec q_A) \neq (\vec 0,\vec 0)$ and $ (\vec p_{A'} ,
\vec q_{A'})=(\vec 0, \vec 0)$.
Therefore we have $\Ptr{A'}{\sigma} = I_A/d^n$.

(2) Let $J_\Lambda$ be the Choi state of $\Lambda$.
Then $J_\Lambda\ge 0$,
and $\trace_{A'}(J_\Lambda) = I_A/d^n$.
Based on the above observation,  we can find state $\sigma$ on $AA'$ such that $\trace_{A'}(\sigma) = I_A/d^n$ and $J_{\Lambda} \boxtimes \sigma = \CMM( J_{\Lambda}) = J_{\CMM(\Lambda)}$.
Therefore the map
\begin{align}\label{0207shi2}
\Lambda_{\sigma}: \rho_A \; \mapsto \; d^n \trace_A \left[\sigma \cdot (\rho_A^T\ot I_{A'}) \right] \;,\quad\forall \rho_A\in \mathcal{B}(\mathcal{H}^{\ot n})
\end{align}
is a quantum channel,
and the map
\begin{align*}
\Theta_\sigma:  \Lambda_1\mapsto \Lambda_1\boxtimes \Lambda_{\sigma}
\end{align*}
is a superchannel, 
satisfying $\Theta_{\sigma} (\Lambda) = \CMM(\Lambda)$. Moreover, we have 
$\Theta_{\sigma} (\mathcal{R})=\mathcal{R}\boxtimes\Lambda_{\sigma}=\mathcal{R}$.
Hence, by the monotonicity of R\'enyi relative entropy, 
\begin{align*}
D_{\alpha}(\CMM(\Lambda)||\mathcal{R}) = D_{\alpha}(\Theta_{\sigma} (\Lambda)|| \Theta_{\sigma}(\mathcal{R})) \leq D_{\alpha}(\Lambda||\mathcal{R}) \;,
\end{align*}
for any $\alpha\in [1/2, +\infty]$.
By the definition of the R\'enyi entropy of channels,
we have $H_{\alpha}(\CMM(\Lambda))\geq H_{\alpha}(\Lambda)$.
\end{proof}

The above result demonstrates the extremality of the mean channel in terms of entropic measures. This extends our extremality results from states to channels, and it also highlights the universality of our quantum convolution framework. This framework can be applied both to  states and to channels.

\subsection{Quantum entropy inequality for channels}
In this subsection, we will focus on the 
entropy power inequalities for the R\'enyi entropy of quantum channels.

\begin{thm}[\bf Convolution increases entropy of channels]\label{thm:entropy_power_chn}
Let $\Lambda_1, \Lambda_2$  be two $n$-qudit channels and $\alpha\in [1/2,+\infty]$.

(1) If $G$ is odd-parity positive, then 
\begin{align*}
H_\alpha( \Lambda_1\boxtimes \Lambda_2) \ge  H_\alpha(\Lambda_2) \;.
\end{align*}

(2) If $G$ is even-parity positive, then 
\begin{align*}
H_\alpha( \Lambda_1\boxtimes \Lambda_2) \ge  H_\alpha(\Lambda_1)  \;.
\end{align*}

(3) If $G$ is positive, then
\begin{align*}
H_\alpha( \Lambda_1\boxtimes \Lambda_2) \ge  \max\set{H_\alpha(\Lambda_1), H_\alpha(\Lambda_2)} \;.
\end{align*}

\end{thm}
\begin{proof}
(3) is a consequence of  (1) and (2). 
Let us start with the even-parity positive case; the odd-parity case can be proved in the same way.
For $\alpha\in [1/2, +\infty]$, 
the quantum R\'enyi  relative entropy is monotone under quantum superchannels~\cite{Tomamichel2015quantum}, 
which 
means that for any quantum superchannel $\Theta$ we have 
\begin{eqnarray*}
D_{\alpha}(\Lambda_1||\Lambda_2)\geq
D_{\alpha}(\Theta(\Lambda_1)||\Theta(\Lambda_2)) \;.
\end{eqnarray*}

Given a quantum channel $\Lambda_2$,  let us define 
the superchannel $\Theta_{\Lambda_2}(\cdot)$ 
to be 
\begin{eqnarray*}
\Theta_{\Lambda_2}(\Lambda)
=\mathcal{E}_G\circ (\Lambda\ot\Lambda_2)\circ \mathcal{E}^{-1}_G \;.
\end{eqnarray*}
Due to Proposition \ref{prop:iden_pre}, we have 
$\Theta_{\Lambda_2}(\mathcal{R})=\mathcal{R}\boxtimes \Lambda_2=\mathcal{R}$ as 
$G$ is even-parity positive. Hence
\begin{eqnarray*}
D_{\alpha}(\Lambda_1||\mathcal{R})
\geq
D_{\alpha}(\Theta_{\Lambda_2}(\Lambda_1)||\Theta_{\Lambda_2}(\mathcal{R}))
=D_{\alpha}(\Theta_{\Lambda_2}(\Lambda_1)||\mathcal{R}) \;,
\end{eqnarray*}
that is 
\begin{eqnarray*}
H_{\alpha}(\Lambda_1\boxtimes\Lambda_2)
\geq H_{\alpha}(\Lambda_1)
\end{eqnarray*}
for $\alpha\geq 1/2$. Similarly, we can prove that 
$H_{\alpha}(\Lambda_1\boxtimes\Lambda_2)
\geq H_{\alpha}(\Lambda_2)$ when $G$ is odd-parity positive.
\end{proof}

Let us take the convolution repeatedly and define $\boxtimes^{N+1}\Lambda=(\boxtimes^N \Lambda)\boxtimes\Lambda$, and $\boxtimes^0\Lambda=\Lambda$.
Similar to the state case, we show that the quantum R\'enyi entropy $H_{\alpha}(\boxtimes^N\Lambda)$ is also nondecreasing w.r.t. the number of convolutions $N$.

\begin{prop}[\bf Second law of thermodynamics for channel convolution]
Let $\Lambda$ be an $n$-qudit quantum channel and $\alpha\in [1/2, +\infty]$.
Then
\begin{eqnarray*}
H_{\alpha}(\boxtimes^{N+1}\Lambda)
\geq H_{\alpha}(\boxtimes^N\Lambda) \;, \forall N\geq 0 \;.
\end{eqnarray*}
\end{prop}
\begin{proof}
This is a corollary of Theorem \ref{thm:entropy_power_chn}.
\end{proof}

\subsection{Clifford unitary in the convolution}
Here we consider the role of Clifford unitaries in channel convolution.

\begin{lem}[\bf Corollary 8, \cite{Gour21}]\label{lem:chn_0}
Given an $n$-qudit channel $\Lambda$, the channel entropy $H(\Lambda)\geq -n\log d$, and 
$H(\Lambda)=-n\log d$ iff $\Lambda$ is a unitary channel.
\end{lem}

\begin{lem}\label{lem:chann_unitary}
Let the parameter matrix $G$ be positive and invertible, and $\Lambda_1$, $\Lambda_2$ be two $n$-qudit quantum unitary channels. 
If $\Lambda_1\boxtimes \Lambda_2$ is a unitary channel,
then all $\Lambda_1$, $\Lambda_2$ and $\Lambda_1\boxtimes \Lambda_2$ are Clifford unitaries.
Moreover, for every $(\vec{p},\vec q)\in V^n$, there is $(\vec{x}_0,\vec{y}_0)\in V^n$ such that 
\begin{eqnarray*}
\Lambda_1\boxtimes\Lambda_2(w(\vec p, \vec{q}))&=&w(\vec x_0,\vec y_0) \;, \\
\Lambda_1(w(Ng_{11}\vec p, g_{00}\vec q))&=&w(Ng_{11}\vec x_0, g_{00}\vec y_0) \;,\\
 \Lambda_2 (w(-Ng_{10}\vec p, g_{01}\vec q))&=&w(-Ng_{10}\vec x_0, g_{01}\vec y_0) \;,
\end{eqnarray*}
up to some phases.
\end{lem}
\begin{proof}
First, we have
$-n \log d =H(\Lambda_1\boxtimes\Lambda_2)\geq \max\set{H(\Lambda_1), H(\Lambda_2)}$ based on  Theorem \ref{thm:entropy_power_chn} and Lemma \ref{lem:chn_0}. Hence
$H(\Lambda_1)=H(\Lambda_2)=-n\log d$ and both $\Lambda_1$ and $\Lambda_2$ are unitary channels.
Second, for any $(\vec p,\vec q)\in V^n$, we have 
\begin{eqnarray}\label{eq:conv_equal}
\Lambda_1\boxtimes\Lambda_2(w(\vec p, \vec q))
=\frac{1}{d^n}
\Lambda_1(w(Ng_{11}\vec p, g_{00}\vec q))
\boxtimes \Lambda_2 (w(-Ng_{10}\vec p, g_{01}\vec q))\;.
\end{eqnarray}
Assume
\begin{align*}
\Lambda_1\boxtimes \Lambda_2(w(\vec p, \vec q)) =& 
\sum_{\vec x, \vec y} \lambda_0(\vec x, \vec y) w(\vec x, \vec y) \;,\\
\Lambda_1(w(Ng_{11}\vec p, g_{00}\vec q)) =& 
\sum_{\vec x, \vec y} \lambda_1(\vec x, \vec y) w(\vec x ,\vec y) \;,\\
 \Lambda_2 (w(-Ng_{10}\vec p, g_{01}\vec q)) =&
\sum_{\vec x, \vec y} \lambda_2(\vec x, \vec y) w(\vec x, \vec  y) \;,
\end{align*}
Then from \eqref{eq:conv_equal} one infers 
\begin{eqnarray*}
\lambda_0(\vec x, \vec y)
=\lambda_1(Ng_{11}\vec x, g_{00}\vec y)
 \lambda_2 (-Ng_{10}\vec x, g_{01}\vec y) \;.
\end{eqnarray*}
Since  $\Lambda_1$, $\Lambda_2$ and $\Lambda_1\boxtimes \Lambda_2$ are unitary channels, 
\begin{align}\label{0117shi2}
\sum_{\vec x,\vec y} |\lambda_j(\vec x, \vec y) |^2=1 \;,\quad j=0,1,2.
\end{align}
Hence, 
\begin{eqnarray*}
1= \sum_{\vec x,\vec y}|\lambda_0(\vec x, \vec y)|^2
=\sum_{\vec x,\vec y}
|\lambda_1(Ng_{11}\vec x, g_{00}\vec y)
 \lambda_2 (-Ng_{10}\vec x, g_{01}\vec y)|^2\\
 \leq 
 \min\left\{
 \sum_{\vec x,\vec y}|\lambda_1(\vec x, \vec y)|^2,
 \sum_{\vec x,\vec y}|\lambda_2(\vec x, \vec y)|^2
 \right\}=1 \;.
\end{eqnarray*}
Condition \eqref{0117shi2} says that the above inequality is an equality,
thus we can find some $\vec x_0, \vec y_0$ such that 
\begin{align*}
|\lambda_0(\vec x_0 \vec y_0) |= 
|\lambda_1(Ng_{11}\vec x_0, g_{00}\vec y_0) |= |\lambda_2(-Ng_{10}\vec x_0, g_{01}\vec y_0) |=1 \;,
\end{align*}
and the proof is complete.
\end{proof}

\begin{thm}\label{thm:min_out_cupchn}
Let $\Lambda_1$ and $\Lambda_2$ be two $n$-qudit quantum channels. 
The output channel $\Lambda_1\boxtimes \Lambda_2$ has the minimal channel entropy $-n\log d$ iff
 $\Lambda_1$ and $\Lambda_2$ are two Clifford unitary channels such that, for every $(\vec{p},\vec q)\in V^n$,
there exists $(\vec{x}_0,\vec{y}_0)\in V^n$ satisfying 
\begin{eqnarray*}
\Lambda_1(w(Ng_{11}\vec p, g_{00}\vec q))&=&w(Ng_{11}\vec x_0, g_{00}\vec y_0) \;,\\
 \Lambda_2 (w(-Ng_{10}\vec p, g_{01}\vec q))&=&w(-Ng_{10}\vec x_0, g_{01}\vec y_0) \;,
\end{eqnarray*}
up to some phases.
\end{thm}
\begin{proof}
In one direction,
if the output channel $\Lambda_1\boxtimes \Lambda_2$ has the minimal channel entropy $-n\log d$, then   $\Lambda_1\boxtimes \Lambda_2$ is a unitary channel, hence  the statement holds by Lemma \ref{lem:chann_unitary}. 

In the other direction, 
\begin{eqnarray*}
\Lambda_1\boxtimes \Lambda_2(w(\vec p, \vec q))
&=&\frac{1}{d^n}\Lambda_1(w(Ng_{11}\vec p, g_{00}\vec q))
\boxtimes  \Lambda_2 (w(-Ng_{10}\vec p, g_{01}\vec q))\\
&=&\frac{1}{d^n}w(Ng_{11}\vec x_0, g_{00}\vec y_0)
\boxtimes w(-Ng_{10}\vec x_0, g_{01}\vec y_0)
=w(\vec x_0, \vec y_0)
\end{eqnarray*}
up to some phases.
So $\Lambda_1\boxtimes \Lambda_2$ is a unitary channel with channel entropy $-n \log d$.
\end{proof}

Given a quantum channel $\Lambda$, denote $\Theta_{\Lambda}$ to be the superchannel that maps a channel $\Lambda_1$ to $\Lambda_1\boxtimes \Lambda$.
In the following, we consider the Holevo capacity of superchannels $\Theta_{\Lambda}$,
which is defined as follows.

\begin{Def}[\bf Holevo capacity for superchannel]
The Holevo capacity of a superchannel $\Theta$ is 
\begin{align}\label{0117shi4}
\chi_H(\Theta) = \sup_{p_i\,,\CNN_i} \left\{ H\left(\sum_i p_i \Theta(\CNN_i)\right) - \sum_i p_i H\left(\Theta (\CNN_i)\right) \right\} \;,
\end{align}
where each $p_i\ge 0$, \   $ \sum_i p_i=1$,
and each $\CNN_i$ is a quantum channel.
\end{Def}

\begin{figure}[t]
  \center{\includegraphics[width=12cm]{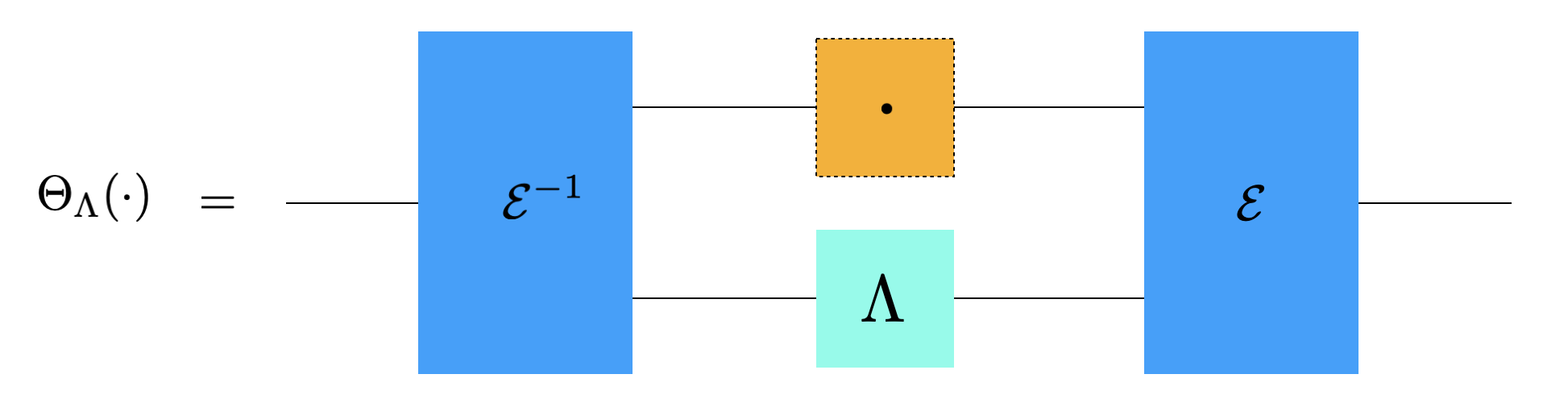} }    
  \caption{ The quantum circuit to realize the superchannel $\Theta_{\Lambda}$.}
 \end{figure}

\begin{thm}
Let $\Lambda$ be a unitary channel.
The Holevo capacity 
$\chi_H(\Theta_{\Lambda})\leq 2n\log d$, and $\Theta_{\Lambda}$ achieves the maximal Holevo capacity $2n\log d$,  if and only if 
$\Lambda $ is a Clifford unitary channel.
\end{thm}
\begin{proof}
First, 
$\chi_H(\Theta_{\Lambda})\leq 2n\log d$ comes from the fact that 
the channel entropy satisfies $-n\log d\leq H(\mathcal{N})\leq n\log d$.
If $\chi_H(\Theta_{\Lambda})=2n\log d $, then $H(\Theta_\Lambda (\CNN_i)) = - n\log d $ for some channel $\CNN_i$,
which implies that $\Lambda \boxtimes\CNN_i $ is unitary. 
By Theorem \ref{thm:min_out_cupchn}, 
$\Lambda$ is a Clifford unitary channel. 

On the other hand, if $\Lambda$ is a Clifford unitary channel, say $\Lambda(\cdot)=U_{\Lambda}(\cdot) U^\dag_{\Lambda}$, 
then there exists a subset $\set{(\vec x_i, \vec y_i)}^{2n}_{i=1}$ of  $V^n$ such that 
\begin{align*}
U_{\Lambda}X^{g_{01}}_{i} U^\dag_{\Lambda}
=&w(-Ng_{10}\vec x_{2i-1}, g_{01}\vec y_{2i-1})\;,  \\
U_{\Lambda} Z^{-Ng_{10}}_i U^\dag_{\Lambda}
=&w(-Ng_{10}\vec x_{2i}, g_{01}\vec y_{2i})\;,
\end{align*}
up to phases, and
\begin{eqnarray*}
\inner{(\vec{x}_i,\vec{y}_i)}{(\vec{x}_j,\vec{y}_j)}_s
=\Omega_{ij} \;, \quad \forall i,j\;,
\end{eqnarray*}
where $\Omega$ is the $2n\times 2n $ matrix 
\begin{equation*}
    \Omega=
    \left[
    \begin{array}{cc}
      0   & 1 \\
      -1  & 0
    \end{array}
    \right]^{\oplus n} \;.
\end{equation*}

Define another unitary quantum channel $\mathcal{N}$ by 
\begin{eqnarray*}
\mathcal{N}(X^{g_{00}}_{i} )
&=&w(Ng_{11}\vec x_{2i-1}, g_{00}\vec y_{2i-1}) \;,\\
\mathcal{N}( Z^{Ng_{11}}_i )
&=&w(Ng_{11}\vec x_{2i}, g_{00}\vec y_{2i}) \;,
\end{eqnarray*}
then $\mathcal{N}$ is 
a Clifford unitary channel. 
Moreover, we have 
\begin{eqnarray*}
\mathcal{N}\boxtimes \Lambda(X_i)
&=&w(x_{2i-1}, y_{2i-1}) \;,\\
\mathcal{N}\boxtimes \Lambda(Z_i)
&=&w(x_{2i}, y_{2i}) \;.
\end{eqnarray*}
Now, for any $(\vec p, \vec q)\in V^n$,
let $\mathcal{N}_{\vec p,\vec q}$ be the quantum channel such that 
\begin{eqnarray*}
\mathcal{N}_{\vec p, \vec q}(\rho)
=\mathcal{N}\left(w(\vec p,\vec q)\rho w(\vec p, \vec q)^\dag\right),
\end{eqnarray*}
for any quantum state $\rho$.
Let  $\lambda_{\vec p, \vec q}=\frac{1}{d^{2n}}$ for each $(\vec p,\vec q)$,
then
\begin{eqnarray*}
\sum_{\vec p,\vec q}\lambda_{\vec p, \vec q}
\mathcal{N}_{\vec p,\vec q}
=\mathcal{N}\circ \mathcal{R}
=\mathcal{R} \;,
\end{eqnarray*}
and thus 
\begin{eqnarray*}
\sum_{\vec p,\vec q}\lambda_{\vec p, \vec q}
\Theta_{\Lambda}
(\mathcal{N}_{\vec p,\vec q})
=\Theta_{\Lambda}(\mathcal{R})
=\mathcal{R} \;.
\end{eqnarray*}
Hence
\begin{eqnarray*}
H\left(\sum_{\vec p, \vec q}\lambda_{\vec p, \vec q} \mathcal{N}_{\vec p, \vec q}\right)
=n\log d \;.
\end{eqnarray*}

One also observes that
\begin{eqnarray*}
\Theta_{\Lambda}(\mathcal{N}_{\vec p,\vec q})(\rho)
=
\mathcal{N}_{\vec p,\vec q}
\boxtimes \Lambda(\rho)
=\mathcal{N}\boxtimes \Lambda
 \left(w(g_{11}N\vec p, g_{00}\vec q)\rho w(g_{11}N\vec p, g_{00}\vec q)^\dag\right) \;,
\end{eqnarray*}
i.e., $\Theta_{\Lambda}(\mathcal{N}_{\vec p,\vec q})$ is Clifford unitary. 
Hence
\begin{eqnarray*}
H\left(\Theta_{\Lambda}(\mathcal{N}_{\vec p,\vec q})\right)
=-n\log d \;.
\end{eqnarray*}
Therefore $
\chi_H(\Theta_{\Lambda})=2n\log d
$.
\end{proof}

\subsection{Central limit theorem for channels}

Denote $\boxtimes^{N+1}\Lambda=(\boxtimes^N \Lambda)\boxtimes\Lambda$, and $\boxtimes^0\Lambda=\Lambda$, where $\boxtimes$ is short for a beam splitter convolution 
$\boxtimes_{s,t}$.

\begin{Def}[\bf Zero-mean channel]\label{0128def2}
An $n$-qudit channel $\Lambda$ is called a zero-mean channel if 
its Choi state $J_\Lambda$ is a zero-mean state.
\end{Def}

\begin{cor}\label{0212cor5}
A channel $\Lambda$   has zero mean if and only if the characteristic function of $\Lambda(w(\vec x,\vec y))$ takes values in $\{0,1\}$ for every $(\vec x,\vec y)\in V^n$.
\end{cor}

\begin{Def}[\bf Diamond distance, \cite{Aharonov98}]
Given two $n$-qudit channels $\Lambda_1$ and $\Lambda_2$, the diamond distance between $\Lambda_1$ and $\Lambda_2$ is
\begin{eqnarray*}
\norm{\Lambda_1-\Lambda_2}_{\diamond}
=\sup_{\rho\in \mathcal{D}(\mathcal{H}_S\ot \mathcal{H}_R)}\norm{(\Lambda_1-\Lambda_2)\ot id_R(\rho)}_1 \;,
\end{eqnarray*}
where $id_R$ is the identity mapping on the ancilla system.
\end{Def}

\begin{Def}[\bf Magic gap of a quantum channel]
Given a quantum channel 
$\Lambda$, 
the magic gap of $\Lambda$ is  the magic gap of the Choi state $J_{\Lambda}$, i.e., 
\begin{eqnarray*}
MG(\Lambda)
:=MG(J_{\Lambda}) \;.
\end{eqnarray*}
\end{Def}

\begin{lem}\label{lem:zero_wel_pre}
Let $\Lambda$ be a zero-mean channel.
For every $(\vec p,\vec q)\in V^n$ and $t\in \Z_d$, we have
\begin{align*}
\Xi_{\Lambda (w(\vec p,\vec q)) }(\vec a,\vec b) = 1 \quad \Rightarrow \quad \Xi_{\Lambda( w(t\vec p,t\vec q)) }(t\vec a,t\vec b) = 1 \;.
\end{align*}

\end{lem}
\begin{proof}
Since $\Lambda$ has zero mean,
by Corollary \ref{0212cor5},
both $\Lambda (w(\vec p,\vec q))$ and $\Lambda (w(t\vec p,t\vec q))$ have characteristic functions taking values in $\{0,1\}$.
When 
\begin{align*}
\Xi_{\Lambda (w(\vec p,\vec q)) }(\vec a,\vec b) = 1 \;,
\end{align*}
by \eqref{0127shi3} we have
\begin{align*}
\Xi_{J_\Lambda} [ (-\vec p,\vec q), (\vec a,\vec b) ]= 1 \;.
\end{align*}
By Lemma \ref{0106lem1},
\begin{align*}
|\Xi_{J_\Lambda} [( -t\vec p,t\vec q), (t\vec a,t\vec b)] |= 1 \;.
\end{align*}
We assume $\Lambda$ has zero mean,
thus we have
\begin{align*}
\Xi_{J_\Lambda}[ ( -t\vec p,t\vec q), (t\vec a,t\vec b) ]= 1 \;.
\end{align*}
Again by \eqref{0127shi3},
\begin{align*}
\Xi_{\Lambda (w(t\vec p,t\vec q)) }(t\vec a,t\vec b) = 1 \;.
\end{align*}
The proof is complete.
\end{proof}

\begin{thm}
Given an $n$-qudit zero-mean channel $\Lambda$, we have 
\begin{eqnarray*}
\norm{\boxtimes^N \Lambda-\CMM(\Lambda)}_{\diamond}
\leq d^{2n} (1-MG(\Lambda))^N\norm{J_{\Lambda}-J_{\CMM(\Lambda)}}_2 \;.
\end{eqnarray*}
\end{thm}

\begin{proof}
We have the estimates
\begin{align*}
\norm{\boxtimes^N \Lambda-\CMM(\Lambda)}_{\diamond}
\leq& d^n \norm{J_{\boxtimes^N \Lambda}-J_{\CMM(\Lambda)}}_1
=d^n \norm{\boxtimes^NJ_{\Lambda}-\CMM(J_{\Lambda})}_1\\
\leq& d^{2n}\norm{\boxtimes^NJ_{\Lambda}-\CMM(J_{\Lambda})}_2
\leq  d^{2n} (1-MG(\Lambda))^N\norm{J_{\Lambda}-J_{\CMM(\Lambda)}}_2 \;,
\end{align*}
where the first inequality comes from the fact that $\norm{ \Lambda_1-\Lambda_2}_{\diamond}\leq d^n\norm{J_{\Lambda_1}-J_{\Lambda_2}}_1$ \cite{Watrous18}, the equality comes from the definition of mean channels, 
and the next inequality comes from the fact that 
$\norm{\cdot}_1\leq \sqrt{d^{2n}}\norm{\cdot}_2$.
The last inequality uses Theorem \ref{thm:CLT_gap}.
\end{proof}

\section{Conclusion and open problems}
In this work, we introduce a framework of convolution to study stabilizer states and channels in DV quantum systems. We have collected some of these results without proof in a companion paper~\cite{BGJ23a},  and also mentioned
several open problems in the last section of our companion work \cite{BGJ23a}. 
In addition, one might consider the following:

(1) We have proved the extremality of the mean state for every quantum R\'enyi entropy. We thank Eric Carlen for asking: is this related to the extremality of Wehrl  entropy in the CV case~\cite{LiebCMP1973,Lieb14}?

(2) We prove the extremality of the mean channel for local dimension being prime $\geq 7$. Does it hold for other prime dimensions?

(3) Here we consider the entropy inequality on convolution of channels by using the channel entropy defined in 
\cite{Gour21}, which depends on the optimization of some joint states on system and ancilla system. Besides Clifford unitaries, what are other quantum channels such that the equality in Theorem \ref{thm:entropy_power_chn} holds?
Can one have  entropy inequalities for convolution with a different definition of channel entropy or convolution?

(4)  Can one consider the convolution of channels in CV systems, e.g., define the convolution in CV systems by using \eqref{eq:exp_box_chan} in Theorem \ref{thm:exact_conv_chan}? More research is necessary on the convolution of quantum channels.

Aside from  the questions mentioned above, further applications and connections of our quantum convolution have been explored in quantum capacity within quantum information theory~\cite{BJ2025a} and in additive combinatorics~\cite{BGJ2025a}. 
 We believe that our quantum convolution will shed further  insights into both mathematics and quantum information theory.

\section{Acknowledgments}
 We thank Roy Garcia, Jiange Li, Seth Lloyd,  and Sijie Luo for helpful discussions.
This work was supported in part by ARO Grant W911NF-19-1-0302, ARO MURI Grant W911NF-20-1-0082, and NSF Eager Grant 2037687.

\section{Appendix}
We recall some basic facts about convolution for CV systems, namely majorization and the existence of parameters in the discrete beam splitter and amplifier.
\subsection{Convolution in CV quantum systems}
In CV quantum systems, Gaussian states (and processes which can be represented in terms of Gaussian distributions) are primary tools in studying CV quantum information~\cite{SethRMP12}. 
One important property of Gaussian states is extremality within all CV states, under some constraint on the covariance matrix~\cite{Holevo99,HolevoMutual99,CerfPRL04,Eisert07}.  
Gaussian states also
minimize the output entropy or maximize the achievable rate of communication by Gaussian channels. One sees this using quantum entropy-power inequalities on the convolution of CV states~\cite{Konig13,KonigPRL13,Konig14,Palma14,PalmaIEE16,Qi2016,Huber17,PalmaPRA15,PalmaIEEE17,PalmaPRL17,PalmaIEEE19}. This statement is a quantum analogue of Shannon's entropy power inequality~\cite{Shannon,Stam,Lieb78}.
These states have both been realized in experiment, and also applied in quantum information tasks, such as quantum teleportation~\cite{Vaidman94,Braunstein98,Tittel98},
quantum-enhanced sensing~\cite{Caves81,Bondurant84,Tan08,Zhuang17},  quantum-key distribution~\cite{Grosshans02} and 
quantum-speed limits~\cite{Becker21}.

Similar to DV systems,  computational processes with only Gaussian states and processes can be efficiently simulated  on a classical computer~\cite{Bartlett02,Mari12,Veitch_2013}. Hence, non-Gaussian states and processes are necessary to implement universal quantum computing~\cite{Lloyd99,BartlettPRL02}. To quantify the non-Gaussian nature of a quantum state or process,  the framework of resource theory has been used~\cite{Albarelli18,Takagi18,Chabaud20}.
CV quantum systems have also been considered as a platform to implement quantum computation and realize quantum advantage. Several sampling tasks have been proposed~\cite{Lund14,Douce17,Hamilton17,Cerf17}, including Gaussian boson sampling, a modification of the original boson sampling proposed by Aaronson and Arkhipov~\cite{aaronson2011computational}. This  has attracted much attention and has been realized experimentally; it is claimed that they  beat classical computers~\cite{Pan20,Pan21,Jonathan22}.

Let us consider the CV quantum system with $n$ modes, which has $2n$ canonical degrees of freedom (see the review paper \cite{SethRMP12}).
Let $\hat{q}_k,\hat{p}_k$ be the “position” and “momentum” operators of the $k$-th mode. 
Let us define
\begin{eqnarray*}
\hat{R}:=(\hat{q}_1,\hat{p}_1,...,\hat{q}_n,\hat{p}_n)^T \;,
\end{eqnarray*}
and these operators satisfy the following relations
\begin{eqnarray*}
[\hat{R}_k,\hat{R}_l]=i\Omega_{kl} \;,
\end{eqnarray*}
where
\begin{equation*}
  \Omega=\left[
\begin{array}{cc}
    0 &  1\\
    -1 & 0
\end{array}  
\right]^{\oplus n} \;,
\end{equation*}
and $\oplus n$ is the $n$-fold direct sum.
The Weyl displacement
 operators are defined as 
 \begin{eqnarray*}
 D(\vec{x})=\exp(i\vec{x}^T\Omega\hat{R}), \forall \vec{x}\in\real^{2n} \;.
 \end{eqnarray*}
 The characteristic function $\Xi_{\rho}$ in CV systems is defined as
 \begin{eqnarray*}
 \Xi_{\rho}[\vec{x}]=\Tr{\rho D(\vec{x})} \;.
 \end{eqnarray*}
 The Wigner function $W_{\rho}$ in CV systems is defined as the Fourier transform of the characteristic function, 
\begin{eqnarray*}
W_{\rho}(\vec{x})
=\int_{\real^{2n}}
\frac{d^{2n}\vec{\xi}}{(2\pi)^{2n}}
\exp(-i\vec{x}^T\Omega \vec{\xi})
\Xi_{\rho}(\vec{\xi}) \;.
\end{eqnarray*}

 The beam splitter $U_{\lambda}$, whose action on $2n$ modes is defined by the symplectic matrix
 \begin{equation*}
 U_{\lambda}
 =\left[
 \begin{array}{cc}
     \sqrt{\lambda}I_2 & \sqrt{1-\lambda}I_2 \\
     -\sqrt{1-\lambda}I_2 &\sqrt{\lambda}I_2
 \end{array}
 \right]^{\oplus n} \;.
 \end{equation*}
 The convolution of two $n$-mode quantum states $\rho\boxtimes_{\lambda}\sigma$ by the beam splitter is  defined as 
 \begin{eqnarray*}
 \rho\boxtimes_{\lambda}\sigma
=\Ptr{2}{U_{\lambda}\rho \ot\sigma 
U^\dag_{\lambda}} \;,
 \end{eqnarray*}
 where $\Ptr{2}{\cdot}$ denotes the partial trace taken on the second $n$ modes.
 
 The squeezing unitary  $V_{\kappa}$, whose action on $2n$ modes is defined by the symplectic matrix
 \begin{equation*}
     V_{\kappa}
=\left[
\begin{array}{cc}
   \sqrt{\kappa} I_2 &\sqrt{\kappa-1}Z_2  \\
   \sqrt{\kappa-1}Z_2  &\sqrt{\kappa} I_2 
\end{array}
\right]^{\oplus n} \;,
 \end{equation*}
 where $Z_2=diag[1,-1]$.
  The convolution of two $n$-mode quantum states $\rho\boxtimes_{\kappa}\sigma$ by the squeezing unitary is defined as 
 \begin{eqnarray*}
 \rho\boxtimes_{\kappa}\sigma
=\Ptr{2}{V_{\kappa}\rho \ot\sigma 
V^\dag_{\kappa}} \;.
 \end{eqnarray*}

\subsection{Majorization}\label{0107app1}
In this Appendix, we introduce definitions and related results on majorization used in the paper.
For the details, we refer to \cite{MOA79}.
Let $\vec x=(x_1,...,x_n)$ be any vector in $\R^n$.
Let
$ x_1^{\downarrow} \ge x_2^{\downarrow}\ge \cdots \ge x_n^{\downarrow}$ denote the components of $\vec x$ in decreasing order,
and denote $\vec x^{\downarrow} = ( x_1^{\downarrow} , ...,x_n^{\downarrow} )$ to be the decreasing rearrangement of $\vec x$. Similarly, we let $\vec x^{\uparrow} = ( x_1^{\uparrow} , ...,x_n^{\uparrow} )$ denote the increasing rearrangement of $\vec x$.
Let $\vec y=(y_1,...,y_n)$ also be a vector in $\R^n$.
We write $\vec x\le \vec y$ if $x_k\le y_k$ for every $1\le k\le n$.
When $\sum_{k=1}^n x_k = \sum_{k=1}^n y_k $, we will say $\vec x$ is majorized by $\vec y$ and write $\vec x \prec\vec y$ if 
\begin{align*}
\sum_{j=1}^k x_k^\downarrow \le \sum_{j=1}^k y_k^\downarrow \;,\quad \forall 1\le k\le n \;.
\end{align*}

A matrix $M = (M_{ij})$ is called doubly stochastic if $M$ has non-negative
entries and each row and each column sums to 1. 
Proposition 1.A.3 in \cite{MOA79} states that $\vec x= \vec y M$ for some doubly stochastic matrix $M$ iff $\vec x\prec \vec y$.
A  function $f:\R^n \rightarrow \R$  is said to be Schur convex if $\vec x \prec \vec y \Rightarrow f (\vec x) \le f (\vec y)$. And $f$ is said to be Schur concave if $- f$ is Schur convex.

\begin{lem}[\bf Page 14, \cite{MOA79}]
    The following conditions are equivalent: 
    
(i) $\vec x\prec \vec y$;

(ii) $\vec x=\vec yP$ for some doubly stochastic matrix $P$;

(iii) $\sum_ig(x_i)\leq \sum_ig(y_i)$ for all continuous convex functions $g$;

(iv) $\sum_i|x_i-a|\leq \sum_i|y_i-a|$ for all $a\in \real$;

(v) $\vec x$ is in the convex hull of the $n!$ permutations of $\vec y$.
     
\end{lem}

\subsection{Existence of parameters in discrete beam splitter and amplifier: a number theory guarantee }\label{sec:apen_numT}

\begin{lem}\label{0212lem1}
For every prime number $d$, there  exist  $\lfloor \frac{d+1}{8}\rfloor$ pairs $s,t\in \Z_d$, with $s\not\equiv 0,\pm1$,   such that
\begin{align}\label{0120shi1}
s^2+t^2\equiv 1 \mod d\;.
\end{align}
There also exist $\lfloor \frac{d-3}{4}\rfloor$ pairs $s,t\in \Z_d$ with $s\not\equiv 0,\pm1$, such that 
\begin{align}\label{coshsinh}
s^2-t^2\equiv 1 \mod d\;.
\end{align}
Thus solutions exist in both cases for every prime $d\ge7$.
\end{lem}

\begin{proof}
For $m\in\Z_{d}$, consider $s\equiv\frac{2m}{m^2+1}$, 
and $t\equiv\frac{m^2-1}{m^2+1}$ as a potential solution of \eqref{0120shi1}.  
One can not choose $m=0,1,d-1$, 
for in those cases $s$ or $t$ are zero, 
leaving $d-3$ possible choices for $m$. 
In addition one must rule out two values of $m$ in case $m^{2}=-1\in\Z_{d}$, 
for then $m^{2}+1$ will not have a multiplicative inverse. 
One can count distinct pairs of $(s,t)$ up to the equivalence relation $\{(\pm s,\pm t), (\pm t, \pm s)\}$.  If a solution satisfies $s=t$, then the equivalence class has 4 elements;
otherwise, the equivalence class has 8 elements.  So the number of distinct solutions $(s,t)$ is one of the following numbers (depending on which of these is an integer):
\begin{align*}
\frac{d-3}{8} \;,\;\text{ or } \;\; \frac{d-5}{8} \;,\;\text{ or }\;\;\frac{d-3-4}{8}+1 \;,\;\text{ or }\;\; \frac{d-5-4}{8}+1\;,
\end{align*}
or after simplification  $\lfloor \frac{d+1}{8}\rfloor$.  This establishes the stated number of solutions to \eqref{0120shi1}.
In order to establish solutions to \eqref{coshsinh}, consider $s=\frac{m+m^{-1}}{2}$ and $t=\frac{m-m^{-1}}{2}$.  
Again we cannot choose $m=0,1,d-1$ as $s$ or $t$ would vanish.  
In addition one must rule out two values of $m$ in case $m^{2}=-1\in\Z_{d}$, 
for then $m=-m^{-1}$ and $s$ would vanish.
The counting thus remains the same as for the study of \eqref{0120shi1}.
\end{proof} 

\bibliographystyle{siam}
\bibliography{reference}{}
\end{document}